\documentclass{article}
\usepackage{arxiv}

\usepackage{amsmath,amsfonts,amsthm}
\usepackage{amssymb}
\usepackage{thmtools, thm-restate}
\usepackage{hyperref}
\usepackage{bm}
\usepackage[title]{appendix}
\usepackage[explicit,compact]{titlesec}             
\usepackage{dsfont}
\usepackage[square]{natbib}
\usepackage{graphicx}
\usepackage{xcolor}
\usepackage{booktabs}

\usepackage[bottom]{footmisc}

\newcommand{\cor}{\mathrm{cor}}

\newcommand{\x}{\mathbf{x}}
\newcommand{\y}{\mathbf{y}}

\newcommand{\indicator}[2]{\mathds{1}_{#1}\big({#2}\big)}
\newcommand{\apprx}{\mathrm{apprx}}
\newcommand{\subN}{\underline{N}}
\newcommand{\residue}{\mathcal{O}\Big(\frac{1}{\sqrt{\subN}}\Big)}

\newcommand{\RMSet}[2]{\mathcal{R}_{\mu,#1}^{#2}(\mu^\rho_{#1}, \nu^\rho_{#1})}
\newcommand{\RNSet}[2]{\mathcal{R}_{\nu,#1}^{#2}(\mu^\rho_{#1}, \nu^\rho_{#1})}

\newcommand{\empMu}[1]{\mathrm{Emp}_{\mu}(#1)}
\newcommand{\empNu}[1]{\mathrm{Emp}_{\nu}(#1)}

\newcommand{\bfX}{\mathbf{X}}
\newcommand{\bfY}{\mathbf{Y}}
\newcommand{\bfU}{\mathbf{U}}
\newcommand{\bfV}{\mathbf{V}}
\newcommand{\bfx}{\mathbf{x}}
\newcommand{\bfy}{\mathbf{y}}

\usepackage{cases}

\newcommand{\distH}[1]{\mathrm{dist}_\mathrm{H}( #1 )}

\newcommand{\Co}[1]{\mathrm{Co}(#1)}

\newcommand{\lowervalue}{\underline{J}}

\newcommand{\uppervalue}{\bar{J}}

\DeclareMathOperator*{\argmax}{argmax}
\DeclareMathOperator*{\argmin}{argmin}

\allowdisplaybreaks

\theoremstyle{plain}
\newtheorem{theorem}{Theorem}
\newtheorem{lemma}{Lemma}
\newtheorem{corollary}{Corollary}
\newtheorem{proposition}{Proposition}

\theoremstyle{definition}
\newtheorem{definition}{Definition}
\newtheorem{assumption}{Assumption}
\newtheorem{remark}{Remark}
\newtheorem{problem}{Problem}

\newcommand{\A}{\mathcal{A}}
\newcommand{\B}{\mathcal{B}}
\newcommand{\E}{E}

\newcommand{\M}{\mathcal{M}}
\newcommand{\N}{\mathcal{N}}
\renewcommand{\P}{\mathcal{P}}

\newcommand{\R}{\mathcal{R}}

\newcommand{\U}{\mathcal{U}}
\newcommand{\V}{\mathcal{V}}

\newcommand{\X}{\mathcal{X}}
\newcommand{\Y}{\mathcal{Y}}
\newcommand{\Z}{\mathcal{Z}}

\newcommand{\st}{~\mathrm{s.t.}~}

\newcommand{\ED}{ED}
\newcommand{\MF}{MF}

\newcommand\norm[1]{\left\lVert#1\right\rVert}
\newcommand{\abs}[1]{\left\lvert#1\right \rvert}
\newcommand{\expectation}[1]{\mathbb{E}\left[#1\right]}
\newcommand{\expct}[2]{\mathbb{E}_{#2}\left[#1\right]}

\newcommand{\dtv}[1]{\mathrm{d}_{\mathrm{TV}} \big( #1 \big)}
    
\newcommand{\halfspace}{\kern 0.2em}

\usepackage{fancyhdr}

\fancyhfoffset[L]{0cm} 
\fancyhfoffset[R]{0cm} 
\newcommand{\figref}[1]{Figure~\ref{#1}}

\title{\LARGE \bf  Zero-Sum Games between Large-Population Teams: \\
Reachability-based Analysis under Mean-Field Sharing
}

\author{
	Yue Guan%
	\thanks{Yue Guan is a PhD student with the School of Aerospace Engineering, Georgia Institute of Technology, Atlanta, GA, USA. Email:
		{\tt\small yguan44@gatech.edu}}
	\qquad Mohammad Afshari%
	\thanks{Mohammad Afshari is a Postdoctoral Fellow with the Institute for Robotics and Intelligent Machines, Georgia Institute of Technology, Atlanta, GA, USA. Email:
		{\tt\small mafshari@gatech.edu}}
 	\qquad Panagiotis Tsiotras%
 	\thanks{Panagiotis Tsiotras is the David \& Andrew Lewis Chair Professor with the School of Aerospace Engineering, Georgia Institute of Technology, Atlanta, GA, USA. Email: {\tt\small tsiotras@gatech.edu}}
}

\begin{document}
\maketitle

\setlength{\abovedisplayskip}{5pt}
\setlength{\belowdisplayskip}{5pt}

\begin{abstract}
    This work studies the behaviors of two large-population teams competing in a discrete environment. 
    The team-level interactions are modeled as a zero-sum game while the agent dynamics within each team is formulated as a collaborative mean-field team problem. 
    Drawing inspiration from the mean-field literature, we first approximate the large-population team game with its infinite-population limit. Subsequently, we construct a fictitious centralized system and transform the infinite-population game to an equivalent zero-sum game between two coordinators. 
    We study the optimal coordination strategies for each team via a novel reachability analysis and later translate them back to decentralized strategies that the original agents deploy.
    We prove that the strategies are $\epsilon$-optimal for the original finite-population team game,
    and we further show that the suboptimality diminishes when team size approaches infinity.
    The theoretical guarantees are verified by numerical examples.
\end{abstract}

\section{Introduction}

Multi-agent decision-making arises in many applications, ranging from warehouse robots~\citep{li2021lifelong}, surveillance missions~\citep{tian2020search} to organizational economics~\citep{gibbons2013handbook}. 
While the majority of the literature formulates the problems within either the cooperative or competitive settings, results on mixed collaborative-competitive team behaviors are relatively sparse. 
In this work, we consider a competitive team game, where two teams, each consisting of a large number of intelligent agents, compete at the team level, while the agents within each team collaborate. 
Such hierarchical interactions are of particular interest to military operations~\citep{tyler2020autonomous} and other multi-agent systems that operate in adversarial environments.

There are two major challenges when trying to solve such competitive team problems:
\begin{enumerate}
    \item Large-population team problems are \emph{computationally} challenging since the solution complexity increases exponentially with the number of agents. 
    It is well-known that team optimal control problems belong to the NEXP complexity class~\citep{bernstein2002complexity}.
    \item 
    Competitive team problems are \emph{conceptually} challenging due to the elusive nature of the opponent team. In particular, one may want to impose assumptions on the opponent team to obtain tractable solutions, but these assumptions may not be valid. It is thus unclear whether one can deploy techniques that are readily available in the large-population game literature.
\end{enumerate}

The scalability challenge is also present in dynamic games; however, it has been successfully resolved for a class of models known as mean-field games (MFGs)~\citep{huang2006large,lasry2007mean}.
The salient feature of mean-field games is that the \emph{selfish} agents are weakly coupled in their dynamics and rewards, and coupling is only through the mean field (i.e., the empirical distribution of the agents). 
When the population is large, the game can be approximately solved by considering the infinite-population limit, at which point the interactions among agents are reduced to a game between a typical agent and the infinite population. 

The mean-field approximation idea has been extended to the single-team setting~\citep{arabneydi:2014}, where a group of homogeneous agents are weakly coupled but receive the same team reward and thus collaborate. 
Under the assumption that all agents apply the same strategy, a dynamic programming decomposition is developed leveraging the common-information approach~\citep{nayyar2013decentralized} so that all agents within the team deploy the same strategy prescribed by a fictitious coordinator. 
The identical strategy assumption significantly simplifies the problem, but the optimality of such approach is only guaranteed for the linear-quadratic models~\citep{arabneydi2015team}.
However, in competitive team setting, although one may restrict the strategies used by her/his team to be identical, extending the same assumption to the opponent team may lead to a substantial underestimation of the opponent's capabilities and thus requires further justification. 

\subsection{Main Contributions}

We address the two aforementioned challenges for the class of discrete zero-sum mean-field team games (ZS-MFTGs), which is an extension to the mean-field (single) team problems. 
Importantly, ZS-MFTG models competitive team behaviors and draws focus to the justifiable approximation of the opponent team strategies. 
We focus on discrete-time models with finite state and action spaces, different from the continuous setting in the concurrent work~\citep{sanjari2022nash}.

We develop a dynamic program that constructs $\epsilon$-optimal strategies to the proposed large-population competitive team problem. 
Notably, our approach finds an optimal solution at the infinite-population limit and considers only \emph{identical} team strategies.
This avoids both the so-called ``curse of dimensionality" issue in multi-agent systems and the book-keeping of individual strategies.
Our main results provide a sub-optimality bound on the exploitability for our proposed solution in the original finite-population game, even when the opponent team is allowed to use non-identical strategies.
Specifically, we show that the sub-optimality decreases at the rate of $\mathcal{O}(\underline{N}^{-0.5})$, where $\underline{N}$ is the size of the smaller team.

Our results stem from a novel reachability-based analysis of the mean-field approximation. 
In particular, we show that any finite-population team behavior can be effectively approximated by an infinite-population team that uses identical team strategies. 
This result allows us to approximate the original problem with two competing infinite-population teams and transform the resulting problem into a zero-sum game between two \emph{fictitious} coordinators.
Specifically, each coordinator observes the mean-fields (state distributions) of both teams and prescribes a local policy for all agents within its team (see Figure 1).
Such transformation leads to a simple dynamic program based on the common-information approach~\citep{nayyar2013decentralized}.

\vspace{-0.05in}
\subsection{Related Work}
\vspace{-0.05in}
\paragraph{Mean-field games.}

The mean-field game model was introduced in~\citep{huang2006large, huang2007large, lasry2007mean} to address the scalability issues in large population games. 
The salient feature of mean-field games is that \textit{selfish} agents are weakly-coupled in their dynamics and rewards through the state distribution (mean-field) of the agents. 
If the population is sufficiently large, then an approximately optimal solution for these models can be obtained by solving the infinite-population limit, and the solution in this limiting case is the mean-field equilibrium (MFE). 
The infinite-population game is easier to solve since, when the population is asymptotically large, the action of a single agent does not affect the mean-field. 
In the continuous setting, the MFE is characterized by a Hamilton-Jacobi-Bellman equation (HJB) coupled with a transport equation~\citep{lasry2007mean}. 
The HJB equation describes the optimality conditions for the strategy of a typical agent, and the transport equation captures the evolution of the mean-field. 
The existence and uniqueness of the MFE have been established in~\citep{huang2006large}.
Two desirable features of the MFE are:
(i) the resultant strategy is decentralized and identical across all agents, i.e., an agent selects its action only based on the local information of its own state, and no information regarding the mean-field is needed; 
and (ii) the complexity of the MFE solution does not depend on the number of agents. 
The mean-field results have been extended to discrete environments in~\citep{cui2021approximately,guan2022shaping}, using entropy-regularization.
We refer the readers to~\citep{lauriere2022learning} for a detailed overview of the main results in the mean-field game literature. 

The main differences between our setup and the current mean-field game literature are the following:
(i) we seek team-optimal strategies, while MFG seeks Nash equilibrium strategies. 
In particular, we provide performance guarantees when the whole opponent team deviates, while MFG only considers single-agent deviations.
(ii) The MFE assumes that all agents apply the same strategy thus the mean-field flow can be solved offline.
As a result, the MFE strategy is open-loop to the mean-field.
However, under the setting in this paper, different opponent team strategies can lead to different mean-field trajectories, and thus our optimal strategy requires feedback on the mean-field to respond to different strategies deployed by the opponent team.

\paragraph{Mean-field teams.}

The single-team problem under with mean-field-type coupling has been considered in~\citep{arabneydi:2014}, where all agents receive the same team reward and thus the problem is collaborative. 
It transforms the team problem into a hierarchical control structure, where a fictitious coordinator broadcasts a local policy based on the distribution of agents, and the individual agents then act according to this policy.
The work~\citep{arabneydi:2014} assumes that all agents within the team apply the same strategy and the optimality for the finite-population game is \emph{only} guaranteed in the LQG setting~\citep{mahajan2015linear}.
Our work extends this formulation to the two-team zero-sum setting and justifies the identical team strategy assumptions on the opponent team. 

The concurrent work~\citep{sanjari2022nash} studies a similar team against team problem but in a continuous state and action setting. 
Through modeling the randomized policies as Borel probability measures, the authors showed that under suitable topology induced by certain information structure, a Nash equilibrium exists for the team game, and it is exchangeable for the finite population and symmetric (identical) for the infinite population. 
It is also shown that common randomness is not necessary for an approximate Nash equilibrium for large population games.
Our work differs in the following aspects:
(i)
The concurrent work~\citep{sanjari2022nash} relies on the convexity of the team best-response correspondence and the Kakutani fixed point theorem to establish the existence of a Nash equilibrium. 
However, due to discrete nature of our formulation, the convexity no longer holds, and thus a Nash equilibrium may not exist as shown in Numerical Example~1. 
Consequently, we developed a novel reachability-based analysis to study the single-sided optimality based on the lower and upper game values;
(ii) the mean-field sharing information structure allows an easy transformation of the team against team problem into a zero-sum continuous game between two coordinators, which can be solved via dynamic programming. 
(iii) our reachability-based approach and the additional Lipschitz assumptions allow us to provide the  convergence rate of the finite-population team performance to its infinite-population limit.

\vspace{-0.1in}
\paragraph{Markov team games and information structure.}

Different from mean-field games, team problems focus on fully cooperative scenarios, where all agents share a common reward function.
The theory of teams was first investigated in~\citep{marschak1955elements,radner1962team}. 
From a game-theoretic perspective, this cooperative setting can also be viewed as a special case of Markov potential games~\citep{zazo2016dynamic}, with the potential function being the common reward. 
When all agents have access to the present and past observations along with the past actions of all other agents, such a system
is equivalent to a centralized one, which enables the use of single-agent algorithms, such as value iteration or Q-learning~\citep{bertsekas1996neuro}.
However, such a centralized system suffers from scalability issues, since the joint state and action spaces grow exponentially with the number of agents. 
Furthermore, in many applications, decentralization is a desirable or even required trait, due to reasons such as limited communication bandwidth, privacy, etc.
This work, therefore, investigates the decentralized team control problem under the information structure known as the mean-field sharing~\citep{arabneydi:2014}.

Information structure is one of the most important characteristics of a multi-agent system, since it largely determines the tractability of the corresponding decentralized decision problem~\citep{4048470,ho1980team}.
In~\citep{witsenhausen1971separation}, information structures are classified into three classes: classical, quasi-classical, and non-classical.
The characterizing feature of the classical information structure is centralization of information, i.e., all agents know the information available to all agents acting in the past. 
A system is called quasi-classical or partially nested if the following condition holds: 
If agent $i$ can influence agent $j$, then agent $j$ must know the information available to agent $i$.
All other information structures are non-classical. 
In the team  game literature, information structures that are commonly used include 
state sharing~\citep{aicardi1987decentralized}, 
belief sharing~\citep{yuksel2009stochastic}, 
partial-history sharing~\citep{nayyar2013decentralized}, 
mean-field sharing~\citep{arabneydi:2014}, etc. 
This work, in particular, assumes a mean-field sharing structure, where each agent observes its local state and the mean-fields of both teams. 
As the agents do not receive information regarding the other agents' individual states, the mean-field sharing is a non-classic information structure, thus posing a challenge to attain tractable solutions.

As decentralized team problems belong to the NEXP complexity class~\citep{bernstein2002complexity},
no efficient algorithm is available, in general. 
Nonetheless, it is possible to develop a dynamic programming decomposition for specific information structures. 
Three such generic approaches are discussed in~\citep{mahajan2012information}: namely,
person-by-person approach,
designer's approach,
and common-information approach. 
While we use the common-information approach in this work, we refer the interested reader to~\citep{mahajan2012information} for the other two approaches.

Team problems in a mixed competitive-cooperative setting have also been considered.
For example, \citep{lagoudakis2002learning} proposed a centralized learning algorithm for zero-sum Markov games in which each side is composed of multiple agents collaborating against an opponent team of agents. 
Later, in~\citep{hogeboom2023zero}, the authors studied the information structure and the existence of equilibria in games involving teams against teams.

\vspace{-0.1in}
\paragraph{Common-information approach.} 

The common information-based approach was proposed in~\citep{nayyar2013decentralized} to investigate optimal strategies in decentralized stochastic control problems with partial history-sharing information structure. 
In this model, each agent shares part of its information with other agents at each time step. Based on the common information available to all agents, a fictitious coordinator is designed. It is shown that the optimal control problem at the coordinator level is a centralized Partially Observable Markov Decision Process (POMDP) and can be solved using dynamic programming. The coordinator solves the POMDP problem and chooses a \emph{prescription} for each agent that maps that agent's local information to its actions. The common information approach was used to solve decentralized stochastic control problems including the control sharing information structure~\citep{mahajan2011}, mean-field teams~\citep{arabneydi:2014}, and LQG systems~\citep{mahajan2015linear}.

Most relevant to this work, \citep{Kartik2021} considered a general model of zero-sum stochastic game between two competing teams, where the information available to each agent can be decomposed into common and private information. 
Exploiting this specific information structure, an expanded game between two virtual players (fictitious coordinators) is constructed based on the sufficient statistics known as the common information belief. 
The authors showed that the expanded game can be solved via dynamic programming, and the upper and lower values of the expanded games provide bounds for the original game values.
Although we utilize a similar common information approach as in~\citep{Kartik2021}, this work focuses more on the large-population aspect of the problem, i.e., the mean-field limits of the team games along with the performance guarantees under the mean-field approximation.

In this work, we treat the mean fields of both (infinite-population) teams as common information. 
Instead of having a single coordinator, we introduce two fictitious coordinators, one for each team, and, consequently, formulate an equivalent zero-sum game between the two coordinators. 
Since we consider the equivalent coordinator game at the infinite-population limit, the two mean fields provide enough information to characterize the status of the system under the mean-field dynamics, which leads to a full-state information structure instead of a partially observable one.
Although the original environment is discrete, the resultant zero-sum coordinator game has continuous state and action spaces%
\footnote{The joint state of the coordinator game is the mean-fields of the two teams, which are distributions over state space and are continuous objects. 
The actions used by the coordinators are the local policies, which are distributions over the action spaces and are also continuous.}.
Under certain assumptions, we show that the optimal max-min and min-max value functions are Lipschitz continuous with respect to the mean-fields. 
This justifies the solution of the continuous game via a discretization scheme.

\subsection{Road Map}

Our approach can be summarized in the road map in Figure~\ref{fig:road-map}.

\begin{enumerate}
    \item 
    In Section~\ref{sec:formulation}, we present the formulation of the zero-sum mean-field team games (ZS-MFTG) with \textit{finite number of agents}. 
    \item 
    In Section~\ref{sec:mean-field-approx},
    we employ the \textbf{mean-field approximation} and considers the \textit{infinite-population limit} where identical team strategies are deployed.   
    \item In Section~\ref{sec:coordinator-game}, we use the \textbf{common information approach}~\citep{nayyar2013decentralized} and reformulate the infinite-population problem as a \textit{zero-sum game between two coordinators}. This allows us to efficiently find the max-min and min-max optimal strategies through dynamic programming.
    \item In Section~\ref{sec:performance}, we establish that the optimal strategies for the coordinator game (computed under the common information approach, mean-field approximation, and identical team  strategies) remain $\epsilon$-optimal for the original finite-population game. 
    Specifically, if one team applies the proposed strategy, the opponent team can only increase its performance by $\mathcal{O}\big( 1/\sqrt{\min\{N_1, N_2\}}\big)$, even if the opponent team uses a non-identical team strategy to exploit.
\end{enumerate}

\begin{figure}[b]
    \centering
    \includegraphics[width=0.7\linewidth]{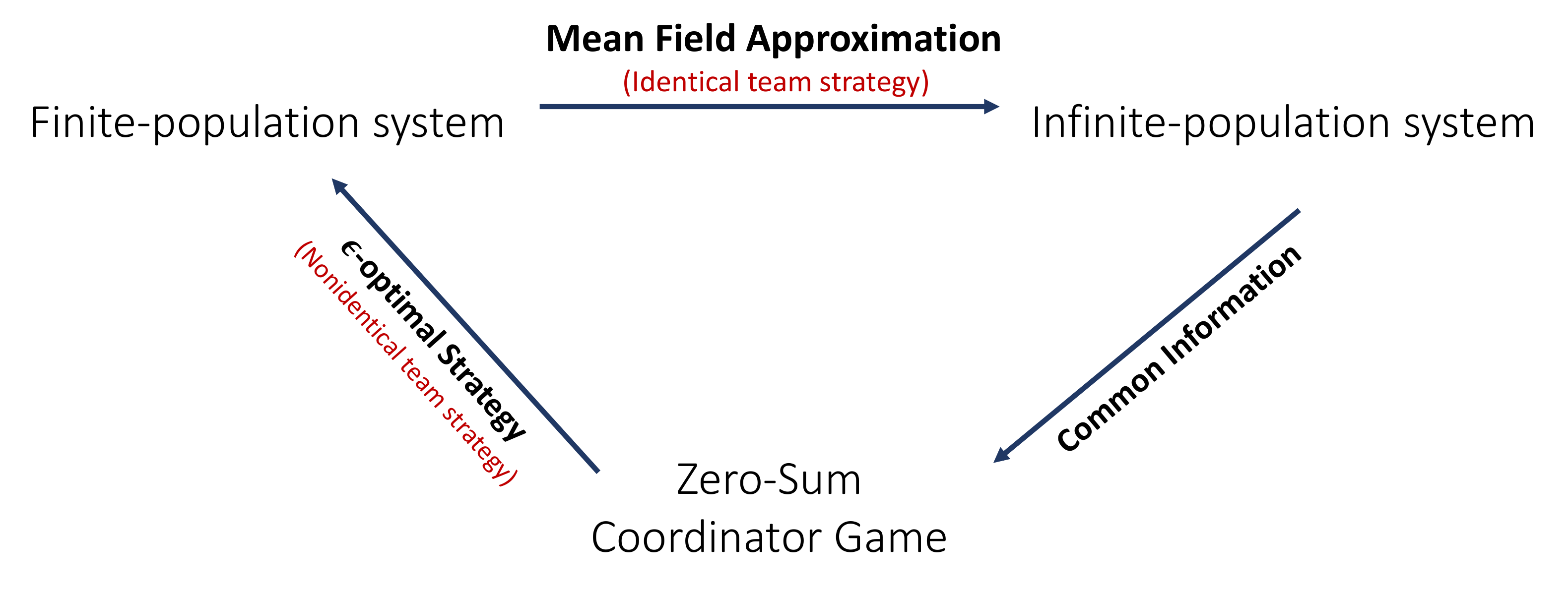}
    \vspace{-0.1in}
    \caption{The road map of the proposed approach.}
    \label{fig:road-map}
\end{figure}

\subsection{Notations}
We use the shorthand notation $[n]$ to denote the set of indices $\{1, 2, \ldots, n\}$. 
The indicator function is denoted as $\indicator{\cdot}{\cdot}$, such that $\indicator{a}{b} = 1$ if $a=b$ and $0$ otherwise.
To distinguish between random variables and their realizations, we use uppercase letters to denote random variables (e.g. $X$ and $\M$) and lowercase letters to denote their realizations (e.g. $x$ and $\mu$).
We use bold letters to denote vectors, e.g. $\bfY = (Y_1, \ldots, Y_n)$.
For a finite set $\E$, we denote the space of all probability measures over $\E$ as $\P(\E)$ and equip it with the total variation norm $\mathrm{d}_\mathrm{TV}$~\citep{chung2001course}.
Additional important notations used in this paper are summarized in Table~\ref{table:notation}.

\begin{table}[h!]
  \begin{center}
    \caption{Important Notation.}
    \label{table:notation}
    \begin{tabular}{l|l|l}
      \toprule 
      \textbf{Notation} & \textbf{Description} & \textbf{Population Size}\\
      \midrule 
      $N_1$ / $N_2$ / $N$ & Number of agents in Blue / Red team / whole system & Finite\\
      $\rho$ & Team size ratio $N_1 / N$ & Finite \\ 
      $X_{i,t}^{N_1} (x_{i,t}^{N_1})$ / $Y_{j,t}^{N_2} (y_{j,t}^{N_2})$ & State of a single Blue/Red agent (realization)& Finite \\
      $\M_t^{N_1}(\mu_t^{N_1})$ / $\N_t^{N_2} (\nu_t^{N_2})$ & Empirical distribution of Blue/Red team (realization)& Finite \\
      $\mu_t^{\rho}$ / $\nu_t^{\rho}$ & Mean-field of Blue/Red team & Infinite \\  
      $\phi_t$ / $\psi_t$ & Identical Blue/Red team policy & Finite \& Infinite\\  
      $\phi$ / $\psi$ & Identical Blue/Red team strategy & Finite \& Infinite\\  
      $\phi^{N_1}_t$ / $\psi^{N_2}_t$ & Non-identical Blue/Red team policy &  Finite\\
      $\phi^{N_1}$ / $\psi^{N_2}$ & Non-identical Blue/Red team strategy &  Finite\\
      $\pi_t$ / $\sigma_t$ & Identical Blue/Red team local policy & Infinite\\
      $\alpha$ / $\beta$ & Blue/Red coordinator strategy & Infinite\\
      $r^\rho_t$ & Game reward & Finite \& Infinite\\
      $f^\rho_t$ / $g^\rho_t$ & Dynamics of Blue/Red agents & Finite \& Infinite\\
      $J^{N}$ & Finite-population game value & Finite \\
      $J^\rho_\cor$ & Coordinator game value & Infinite\\
      $\R^\rho_{\mu,t}$ / $\R^\rho_{\nu,t}$ & Reachablity correspondence for Blue/Red mean-field  \qquad  & Infinite \\
      \bottomrule 
    \end{tabular}
  \end{center}
\end{table}

\section{Problem Formulation}
\label{sec:formulation}

Consider a discrete-time system with two \textit{large} teams of agents that operate over a finite time horizon $T$. 
The Blue team consists of $N_1$ homogeneous agents, and the Red team consists of $N_2$ homogeneous agents.
The total system size is denoted as $N = N_1 + N_2$.
We use $\rho = N_1 / N$ to denote the ratio between the size of the Blue team and the total number of agents (Blue and Red agents combined).
Let $X^{N_1}_{i,t} \in \X$ and $U^{N_1}_{i,t} \in \U$ be the random variables representing the state and action, respectively, taken by Blue agent $i \in [N_1]$ at time $t$. Here, $\X$ and $\U$ represent the \emph{finite} individual state and action space for each Blue agent, independent of $i$ and $t$.
Similarly, we use $Y^{N_2}_{j,t} \in \Y$ and $V^{N_2}_{j,t} \in \V$ to denote the individual state and action of Red agent $j \in [N_2]$.
The joint state and action of the Blue team are denoted as $\bfX^{N_1}_t = \left(X^{N_1}_{1,t},\ldots, X^{N_1}_{N_1,t}\right)$ and $\bfU^{N_1}_t = \left(U^{N_1}_{1,t},\ldots, U^{N_1}_{N_1,t}\right)$, respectively. 
Similarly, the Red team's joint state and action are denoted as $\bfY^{N_2}_t$, and $\bfV^{N_2}_t$.
We use $\bfX^{N_1}_{-i,t} = \left(X^{N_1}_{1,t},\ldots,X^{N_1}_{i-1,t},X^{N_1}_{i+1,t}, \ldots, X^{N_1}_{N_1,t}\right)$
to denote the joint state of the Blue team excluding Blue agent $i$, and $\bfY^{N_2}_{-j,t}$ is defined similarly.

\begin{definition}  [Empirical Distribution]
    \label{def:empirical-dist}
    The \textit{empirical distribution} (\ED{}) for the Blue and Red teams are defined as
    \begin{subequations}\label{eqn:mean-field}
        \begin{align}
            \M^{N_1}_t(x) &= \frac{1}{N_1} \sum_{i=1}^{N_1} \mathds{1}_x(X^{N_1}_{i,t}), ~~~ x \in \X,\\
            \N^{N_2}_t(y) &= \frac{1}{N_2} \sum_{j=1}^{N_2} \mathds{1}_y(Y^{N_2}_{j,t}), ~~~ y \in \Y.
        \end{align}
    \end{subequations}
\end{definition}
Notice that $\M^{N_1}_t(x)$ gives the fraction of Blue agents at state $x$. Hence, the random vector $\M^{N_1}_t = [\M^{N_1}_t(x)]_{x \in \X}$ is a probability measure, i.e., $\M^{N_1}_t \in \P(\X)$. 
Similarly, we have $\N^{N_2}_t \in \P(\Y)$. 
Since finite state and action spaces are considered, we have $\P(\X) = \Delta_{|\X|}$ and $\P(\Y) = \Delta_{|\Y|}$, where $\Delta_k$ is the $k$-dimensional standard simplex.

We introduce two operators $\mathrm{Emp}_\mu: \X^{N_1} \to \P(\X)$ and $\mathrm{Emp}_\nu : \Y^{N_2} \to \P(\X)$ to denote operation performed in \eqref{eqn:mean-field} that relates the joint states and their corresponding \ED{}s:
\begin{align}
    \mu_t^{N_1} =\empMu{\bfx^{N_1}_t}, \qquad
    \nu_t^{N_2} =\empNu{\bfy^{N_2}_t}.
\end{align}

To measure the distance between two distributions (probability measures), we use the total variation.
\begin{definition}
    For a finite set $\E$, the total variation between two probability measures $\mu, \mu' \in \P(\E)$ is given by
    \begin{equation}
        \dtv{\mu, \mu'} = \frac{1}{2} \sum_{e\in \E} \abs{\mu(e) - \mu'(e)} = \frac{1}{2} \norm{\mu - \mu'}_1.
    \end{equation}
\end{definition}

\subsection{Dynamics}

We consider a weakly-coupled dynamics, where the dynamics of each individual agent is coupled with other agents through the \ED{}s (empirical distributions). 
For Blue agent $i$, its stochastic transition is governed by a transition kernel $f^\rho_t: \X \times \X \times \U \times \P(\X) \times \P(\Y) \to [0,1]$ and satisfies
\begin{equation}
    \label{eqn:dynamics-blue}
    \mathbb{P}(X_{i,t+1}^{N_1}=x_{i,t+1}^{N_1} | U_{i,t}^{N_1} = u_{i,t}^{N_1}, \bfX_{t}^{N_1} = \bfx_{t}^{N_1}, \bfY_{t}^{N_2}= \bfy_{t}^{N_2}) = f^\rho_t(x_{i,t+1}^{N_1}\vert x_{i,t}^{N_1}, u_{i,t}^{N_1}, \mu^{N_1}_t, \nu^{N_2}_t),
\end{equation}
where $\mu^{N_1}_t = \empMu{\bfx_t^{N_1}}$ and $\nu^{N_2}_t = \empNu{\bfy_t^{N_2}}$.
Note that the transition kernel has $\rho = N_1 / N$ depends on the ratio between the sizes of the two teams.
Similarly, the transition of Red agent $j$ follows
\begin{equation}
    \label{eqn:dynamics-red}
    \mathbb{P}(Y_{j,t+1}^{N_2}=y_{j,t+1}^{N_2} | V_{j,t}^{N_2} = v_{j,t}^{N_2}, \bfX_{t}^{N_1} = \bfx_{t}^{N_1}, \bfY_{t}^{N_2}= \bfy_{t}^{N_2}) = g^\rho_t(y_{j,t+1}^{N_2}\vert y_{j,t}^{N_2}, v_{j,t}^{N_2}, \mu^{N_1}_t, \nu^{N_2}_t).
\end{equation}

\begin{assumption}
    \label{assmpt:lipschitiz-dynamics}
    The transition kernels $f^\rho_t$ and $g^\rho_t$ are $L_{f_t}$ and $L_{g_t}$-Lipschitz continuous, respectively.
    Formally, for all $\mu, \mu' \in \P(\X)$ and $\nu, \nu' \in \P(\Y)$ and for all $t \in \{0, \ldots, T-1\}$, 
    there exist positive constants $L_{f_t}$ and $L_{g_t}$ such that
    \begin{alignat*}{2}
        &\sum_{x' \in \X} \abs{f^\rho_t(x'|x,u,\mu, \nu) - f^\rho_t(x'|x,u,\mu', \nu')} \leq L_{f_t}\Big(\dtv{\mu, \mu'} + \dtv{\nu, \nu'}\Big) ~~~&&\forall x \in \X, u \in \U, \\
        &\sum_{y' \in \Y} \abs{g^\rho_t(y'|y,v,\mu, \nu) - g^\rho_t (y'|y,v,\mu', \nu')} \leq L_{g_t} \Big(\dtv{\mu, \mu'} + \dtv{\nu,\nu'}\Big)~~~&& \forall y \in \Y, v \in \V.
    \end{alignat*}
\end{assumption}



\subsection{Reward Structure}

Under the team-game framework, agents in the same team share the same reward. 
Similar to the dynamics, we consider a weakly-coupled \emph{team reward} 
\begin{equation}
    \label{eqn:reward}
    r^\rho_t: \P(\X) \times \P(\Y) \to [-R_{\max}, R_{\max}].
\end{equation}

\begin{assumption}
    \label{assmpt:lipschitiz-rewards}
    The reward function is $L_r$-Lipschitz continuous.
    Formally, for all $\mu, \mu' \in \P(\X)$, $\nu, \nu' \in \P(\Y)$, and for all $t \in \{0, \ldots, T\}$, there exists a positive constant $L_r$ such that 
    \begin{equation}
        \label{eqn:lipschitz-rewards}
        \abs{r^\rho_t(\mu, \nu) - r^\rho_t(\mu', \nu')} \leq L_r \Big(\dtv{\mu,\mu'} + \dtv{\nu,\nu'}\Big).
    \end{equation}
\end{assumption}

Under the zero-sum reward structure, we let the Blue team maximize the reward while the Red team minimizes it.

\subsection{Information Structure.}
We assume a mean-field sharing information structure~\citep{arabneydi2015team}. 
At each time step $t$, Blue agent $i$ observes its own state $X_{i,t}^{N_1}$ and the \ED{}s $\M_t^{N_1}$ and $\N_t^{N_2}$.
Similarly, Red agent $j$ observes $Y_{j,t}^{N_2}$, $\M_t^{N_1}$ and $\N_t^{N_2}$.
Viewing the \ED{}s as aggregated states of the teams, 
we consider the following mixed Markov policies:
\begin{equation}
    \label{eqn:IS}
    \begin{aligned}
        \phi_{i,t}&: \U \times \X \times \P(\X) \times \P(\Y) \to [0,1],\\
        \psi_{j,t}&: \V \times \Y \times \P(\X) \times \P(\Y) \to [0,1],
    \end{aligned}
\end{equation}
where $\phi_{i,t}(u|X^{N_1}_{i,t}, \M^{N_1}_t,\N^{N_2}_t)$ represents the probability that Blue agent $i$ selects action $u$ given its state $X_{i,t}^{N_1}$ and the team distributions $\M^{N_1}_t$ and $\N^{N_2}_t$. 
Note that agent's individual state is a private information, while the team \ED{}s are the common information shared among all agents.

An individual strategy is defined as a time sequence $\phi_{i} = \{\phi_{i, t}\}_{t=0}^T$.
A Blue team strategy $\phi^{N_1} = \{\phi_i\}_{i=1}^{N_1}$ is the collection of individual strategies used by each Blue agent. 
We use $\Phi_t$ and $\Phi$ to denote, respectively, the set of individual policies and strategies available to each Blue agent.\footnote{Since Blue agents have the same state and action spaces, they have the same policy space.}
The set of Blue team strategies is then defined as the Cartesian product $\Phi^{N_1} = \times_{i=1}^{N_1} \Phi$.
The notations extend naturally to the Red team.

In summary, an instance of a zero-sum mean-field team game is defined as the tuple: 
\begin{equation*}
    \texttt{ZS-MFTG} = \langle \X, \Y, \U, \V, f^\rho_t, g^\rho_t, r^{\rho}_t, N_1, N_2, T \rangle.
\end{equation*}

\subsection{Optimization Problem.}
\vspace{-0.1in}
Starting from the initial joint states $\bfx_0^{N_1}$ and $\bfy^{N_2}_0$, the performance of any team strategy pair $(\phi^{N_1}, \psi^{N_2}) \in \Phi^{N_1} \times \Psi^{N_2}$ is quantified by the \emph{expected cumulative reward}
\begin{equation}\label{eqn:expected-cum-reward}
    J^{N, \phi^{N_1}, \psi^{N_2}} \big(\bfx_0^{N_1}, \bfy_0^{N_2} \big) = 
    \expct{
    \sum_{t=0}^{T} r^{\rho}_t(\M^{N_1}_t, \N^{N_2}_t) \Big \vert \bfX^{N_1}_0 = \bfx_0^{N_1}, \bfY^{N_2}_0 = \bfy^{N_2}_0}{\phi^{N_1}, \psi^{N_2}},
\end{equation}
where $\M_t^{N_1} = \empMu{\bfX^{N_1}_t}$ and $\N_t^{N_2} = \empNu{\bfY^{N_2}_t}$, and the expectation is with respect to the distribution induced on all system variables by the strategy pair $(\phi^{N_1}, \psi^{N_2})$.

\begin{remark}
    Note that the value function in~\eqref{eqn:expected-cum-reward} takes the joint states as arguments, different from the value function in~\citep{arabneydi2015team} which takes the \ED{} as its argument.
    The difference comes from the non-identical team strategies allowed in our formulation, which requires each agent's state and index information to sample actions and predict the game's evolution. 
    \citeauthor{arabneydi2015team} assume an identical team strategy, and hence the \ED{} is an information state (sufficient statistics) that characterizes the value function of the finite-population game. 
    A counter-example that shows the \ED{} is not an information state under non-identical team strategy is presented in Appendix~\ref{appdx-sec:infomration-example}.
\end{remark}

When the maximizing Blue team considers its worst-case performance, we have the following max-min optimization:
\begin{equation}\label{eqn:origin-optimization}
    \underline{J}^{N*} = \max_{\mathbf{\phi}^{N_1} \in \Phi^{N_1}} ~ \min_{\mathbf{\psi}^{N_2} \in \Psi^{N_2}} ~~ J^{N, \phi^{N_1}, \psi^{N_2}},
\end{equation}

where $\underline{J}^{N*}$ is the lower game value for the finite-population game.
Note that the game value may not always exist (max-min value differs from min-max value)~\citep{elliott1972existence}.
In Section~\ref{sec:existence-game-value}, we present a special case of the ZS-MFTG where the game value exists at the infinite-population limit.
For the general case, we consider the following optimality condition for the Blue team strategy.

\begin{definition}
    \label{def:eps-team-optimal}
    A Blue team strategy $\phi^{N_1 *}$ is $\epsilon$-optimal if
    \begin{equation*}
        \underline{J}^{N*}  \geq \min_{\mathbf{\psi}^{N_2} \in \Psi^{N_2}} J^{N, \phi^{N_1 *}, \psi^{N_2}} \geq \underline{J}^{N*} - \epsilon.
    \end{equation*}
    The strategy $\phi^{N_1 *}$ is optimal if $\epsilon=0$.
\end{definition}

Similarly, the minimizing Red team considers a min-max optimization problem, which leads to the upper game value as follows
\begin{equation}\label{eqn:red-optimization}
    \bar{J}^{N*} = \min_{\mathbf{\psi}^{N_2} \in \Psi^{N_2}} ~\max_{\mathbf{\phi}^{N_1} \in \Phi^{N_1}} ~ J^{N, \phi^{N_1}, \psi^{N_2}}.
\end{equation}

The definitions of optimality and $\epsilon$-optimality extend naturally to the Red team strategies.

The rest of the paper focuses on the performance analysis form the Blue team's perspective (max-min optimization), but the techniques developed are applicable to the Red team's side due to the symmetry of the problem formulation. 

\subsection{Examples}

\paragraph{A two-node example}
Consider a simple team game on a two-node graph in Figure~\ref{fig:demo-example}.
The state spaces are given by $\X\hspace{-0.02in}=\hspace{-0.02in}\{x^1, x^2\}$ and $\Y \hspace{-0.02in}=\hspace{-0.02in} \{y^1, y^2\}$, and the action spaces are $\U \hspace{-0.02in}=\hspace{-0.02in} \{u^1, u^2\}$ and $\V \hspace{-0.02in}=\hspace{-0.02in} \{v^1, v^2\}$.
The Blue action $u^1$ corresponds to staying on the current node and $u^2$ represents moving to the other node. 
The same connotations apply to Red actions $v^1$ and $v^2$.

\begin{figure}[t]
    \centering
    \includegraphics[width=0.3\linewidth]{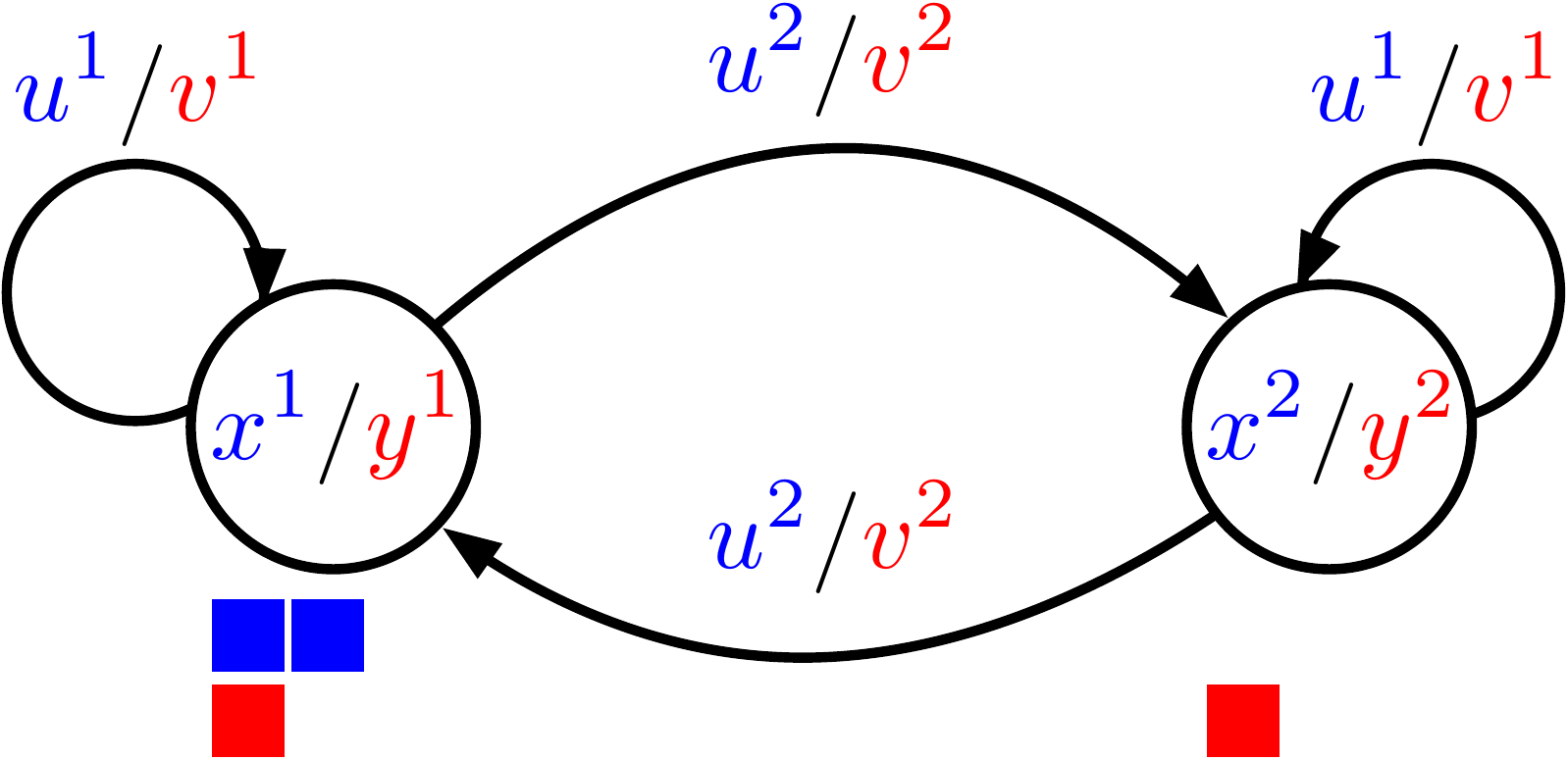}
    \caption{An example of ZS-MFTG over a two-node graph, where $N_1=2$, $N_2=2$ and $\rho=0.5$.}
    \label{fig:demo-example}
    \vspace{-0.2in}
\end{figure}

The maximizing Blue team's objective is to maximize its presence at node~2 (state $x^2$), 
and hence $r_t (\mu, \nu) \hspace{-0.02in} = \mu(x^2)$.
The Blue transition kernel at $x^1$ under $u^2$ is defined as
\begin{align*}
    &f_t(x^1|x^1, u^2, \mu, \nu)  =  0.5 \big(1  -  \big(\rho \mu(x^1)  -  (1 - \rho) \nu(y^1)\big)\big), \\
& f_t(x^2|x^1, u^2, \mu, \nu)  =  0.5 \big(1  +  \big(\rho \mu(x^1)  -  (1  -  \rho) \nu(y^1)\big)\big).
\end{align*}
Under this transition kernel, the probability of a Blue agent transitioning from node 1 to node 2 depends on the Blue team's numerical advantage over the Red team at node~1.

The initial joint states depicted in Figure~\ref{fig:demo-example} are given by  
$\bfx^{2}_0 = [x^1, x^1]$ and $\bfy^{2}_0 = [y^1, y^2]$.
The corresponding \ED{}s are $\mu_0^2 = [1, 0]$, $\nu_0^2=[0.5, 0.5]$, and the running reward is $r_0 = \mu_0^2(x^2)=0$.
Suppose the Blue team applies a team strategy such that $\phi^i_0(u^2|x^1, \mu_0^2, \nu_0^2) = 1$ for both $i\in [2]$ (under which both Blue agents apply $u^2$). 
The probability of an individual Blue agent transitioning to node 2 is 0.625.
Thus, the next Blue \ED{} is a random vector with three possible realizations:
(i) $\M^{2}_1 \hspace{-0.02in}=\hspace{-0.02in} [1, 0]$ with probability 0.14 (both Blue agents remain on node 1);
(ii) $\M^{2}_1 \hspace{-0.02in}=\hspace{-0.02in} [0.5, 0.5]$ with probability 0.47 (one Blue agent moves and the other remains);
and (iii) $\M^2_1 \hspace{-0.02in}=\hspace{-0.02in} [0, 1]$ with probability 0.39 (both Blue agents move).
Suppose the game terminates at $T=1$, then the value under the given Blue strategy $\phi^2$ is given by $J^{4, \phi^{2},\psi^2} (\bfx^2_0, \bfy^2_0) = 0 + (0.14 \cdot 0 + 0.47 \cdot 0.5 + 0.39 \cdot 1) = 0.63$.

\paragraph{Connections to the dynamics in~\citep{huang2006large}}
As a special case of the previously described dynamics and reward structure, we consider the following pairwise-state-coupled dynamics and team reward derived from the single-population model in~\citep{huang2006large}. The transition kernel for the Blue and Red agents are defined as follows: 
\begin{subequations}
\label{eqn:psc-MFTG-dynamics}
\begin{align}
    f_t\big(x^{N_1}_{i,t+1}\big \vert x^{N_1}_{i,t}, u^{N_1}_{i,t}, \bfx^{N_1}_{-i,t}, \bfy^{N_2}_{t}) &= 
    \frac{ \sum_{k=1}^{N_1}f_{1,t}(x^{N_1}_{i,t+1}\big \vert x^{N_1}_{i,t}, u^{N_1}_{i,t}, x^{N_1}_{k,t} )}{N} +
    \frac{ \sum_{k=1}^{N_2}f_{2,t}(x^{N_1}_{i,t+1}\big \vert x^{N_1}_{i,t}, u^{N_1}_{i,t}, y^{N_2}_{k,t} )}{N}\\
    g_t\big(y^{N_2}_{j,t+1}\big \vert y^{N_2}_{j,t}, v^{N_2}_{j,t}, \bfx^{N_1}_{t}, \bfy^{N_2}_{-j, t}) &= 
    \frac{ \sum_{k=1}^{N_1}g_{1,t}(y^{N_2}_{j,t+1}\big \vert y^{N_2}_{j,t}, v^{N_2}_{j,t}, x^{N_1}_{k,t} )}{N} +
    \frac{ \sum_{k=1}^{N_2}g_{2,t}(y^{N_2}_{j,t+1}\big \vert y^{N_2}_{j,t}, v^{N_2}_{j,t}, y^{N_2}_{k,t} )}{N}.
\end{align}
\end{subequations}
The terms $f_{1,t}$, $f_{2,t}$, $g_{1,t}$ and $g_{2,t}$ are transition kernels and model the pair-wise influence between agents. For example, $f_{2,t}$ represents the influence that the state of a Red agent has on the transition of a Blue agent.

The weakly-coupled team rewards are the averaged state-dependent reward
\begin{equation}
\label{eqn:team-rewards-eg}
    \begin{aligned}
        r_t(\bfx_t^{N_1},\bfy_t^{N_2}) &= \frac{1}{N} \sum_{k=1}^{N_1}r_{1,t}(x^{N_1}_{k,t} ) - 
        \frac{1}{N} \sum_{k=1}^{N_2}r_{2,t}(y^{N_2}_{k,t} ).
    \end{aligned}
\end{equation}
Through some algebraic manipulations, the aforementioned pairwise-state-dependent transition kernel and team reward can be transformed into the weakly-coupled form represented by equations \eqref{eqn:dynamics-blue} and \eqref{eqn:reward} as follows:
\begin{align*}
f_t \big(x^{N_1}_{i,t+1}\big \vert x^{N_1}_{i,t}, u^{N_1}_{i,t}, \bfx^{N_1}_{-i,t}, \bfy^{N_2}_{t}\!)
    & =
    \!\rho\!\sum_{x\in \X}\!\mu^{N_1}_t(x) f_{1,t}(x^{N_1}_{i,t+1}\big \vert x^{N_1}_{i,t}, u^{N_1}_{i,t}, x ) \!+\!
    \!(1-\rho)\!\sum_{y\in \Y} \!\nu_t^{N_2}(y)f_{2,t}(x^{N_1}_{i,t+1}\big \vert x^{N_1}_{i,t}, u^{N_1}_{i,t}, y ) \nonumber\\
    &\triangleq f^\rho_t\big(x^{N_1}_{i,t+1}\big \vert x^{N_1}_{i,t}, u^{N_1}_{i,t},  \mu^{N_1}_t, \nu^{N_2}_t),   \nonumber \\
    g_t \big(y^{N_2}_{j,t+1}\big \vert y^{N_2}_{j,t}, v^{N_2}_{j,t}, \bfx^{N_1}_{t}, \bfy^{N_2}_{-j, t}\!)
    & =
    \!\rho\!\sum_{x\in \X}\!\mu^{N_1}_t(x) g_{1,t}(y^{N_2}_{j,t+1}\big \vert y^{N_2}_{j,t}, v^{N_2}_{j,t}, x ) \!+\!
    \!(1-\rho)\!\sum_{y\in \Y} \!\nu_t^{N_2}(y)g_{2,t}(y^{N_2}_{j,t+1}\big \vert y^{N_2}_{j,t}, v^{N_2}_{j,t}, y ) \nonumber\\
    &\triangleq g^\rho_t\big(y^{N_2}_{j,t+1}\big \vert y^{N_2}_{j,t}, v^{N_2}_{j,t},  \mu^{N_1}_t, \nu^{N_2}_t),   \nonumber \\
    r_t(\bfx_t^{N_1},\bfy_t^{N_2}) &= \rho \sum_{x\in \X} \mu^{N_1}_t(x) r_{1,t}(x) - (1-\rho) \sum_{y\in \Y} \nu^{N_2}_t(y) r_{2,t}(y) \triangleq r^\rho_t(\mu_t^{N_1}, \nu_t^{N_2}),
\end{align*}
where $\rho = N_1/N$, $\mu^{N_1}_t = \empMu{\bfx_t^{N_1}}$ and $\nu^{N_2}_t = \empNu{\bfy_t^{N_2}}$ (see Appendix~\ref{appdx:special-case} for the detailed derivations).
In this specific example, it is clear that the team with more agents has a larger influence on the dynamics and the reward.
Hence, it is necessary to include $\rho = N_1/N$ as a parameter for both the dynamics and the reward function.
It can be further verified that the above weakly-coupled transition kernels $f^\rho_t$ and $g^\rho_t$ are both $2$-Lipschitz continuous with respect to the \ED{}s. Additionally, the weakly-coupled reward function has a Lipschitz constant of $L_{r} = 2 \max \{r_{1,\max}, r_{2,\max}\}$, 
where $r_{1,\max} = \max_{x, t} \abs{r_{1,t}(x)}$ and $r_{2,\max} = \max_{y,t} \abs{r_{2,t}(y)}$ (see Propositions~\ref{appdx-prop:psc-MFTG-dynamics-lip} and~\ref{appdx-prop:psc-MFTG-reward-lip} in Appendix~\ref{appdx:special-case}).

\section{Mean-Field Approximation}
\label{sec:mean-field-approx}

The preceding max-min and min-max optimizations are intractable for large-population systems, since the dimension of the joint policy spaces $\Phi^{N_1}$ and $\Psi^{N_2}$ grows exponentially with the number of the agents.
To address this scalability challenge, we first consider the limiting case of the ZS-MFTGs, when the number of agents of both teams goes to infinity.
We further assume that agents from the same infinite-population team employ the same strategy. As a result, the behavioral of the entire team can be represented by a \emph{typical agent}~\citep{huang2006large}.
As we will show later, due to the law of large numbers, the \ED{} of an infinite-population team converges to the state distribution of its corresponding typical agent. 
This limiting distribution is known as the \emph{mean-field}~\citep{huang2006large}.
In the sequel, we formulate the mean-field team game at its infinite-population limit and introduce several additional concepts that are necessary for the performance analysis in the next sections.

\subsection{Identical Team Strategies and Mean-Field Dynamics}
At the infinite population limit, we specifically focus on the Blue and Red teams that employ identical team strategies. Consequently, we first formalize the class of identical team strategies.
\begin{definition} [Identical Blue Team Strategy]
    \label{def:identical-policy}
    The Blue team strategy $\phi^{N_1} = \{\phi_{1},\ldots, \phi_{N_1}\}$ is an identical team strategy, if $\phi_{i_1,t} = \phi_{i_2,t}$ for all $i_1,i_2 \in [N_1]$ and all $t \in \{0, \ldots, T-1\}$. 
\end{definition}

We slightly abuse the notation and use $\phi$ to denote the identical Blue team strategy, under which all Blue agents apply the same individual strategy $\phi$.
Consequently, $\Phi$ is used to denote both the set of Blue individual strategies and the set of identical Blue team strategies.\footnote{When $\Phi$ denotes the set of identical team strategies, it is a subset of the original set of Blue team strategies $\Phi^{N_1}$, i.e., $\Phi \subseteq \Phi^{N_1}$.}
The definitions and notations extend to the identical Red team strategies.

With the notions of identical team strategies, we define the mean-field flow as a deterministic sequence of the state distribution of a typical agent under the given pair of strategies. 
\begin{definition}
    Given identical team strategies $\phi \in \Phi$ and $\psi \in \Psi$, the \MF{}s are defined as the sequence of vectors that adhere to the following coupled deterministic dynamics:
    \begin{subequations}
    \label{eqn:mf-dynamics}
    \begin{align}
        \mu^\rho_{t+1} (x') &= \sum_{x \in \X} \left[\sum_{u \in \U} f^\rho_t(x'|x, u, \mu^\rho_t, \nu^\rho_t )\; \phi_t(u|x,\mu^\rho_t, \nu^\rho_t) \right] \mu^\rho_t(x), \\
        \nu^\rho_{t+1} (y') &= \sum_{y \in \Y} \left[\sum_{v \in \V}g^\rho_t(y'|y, v, \mu^\rho_t, \nu^\rho_t ) \; \psi_t(v|y,\mu^\rho_t, \nu^\rho_t)\right] \nu^\rho_t(y),
    \end{align}
\end{subequations}
where the initial conditions are given by $\mu^\rho_0 = \mu_0$ and $\nu^\rho_0 = \nu_0$.
\end{definition}

The above deterministic mean-field dynamics can be expressed in a compact matrix form as 
\begin{equation}
    \begin{aligned}
        \label{eqn:coordinator-dynamics}
        \mu^\rho_{t+1} &= \mu^\rho_{t} F_t^{\rho}(\mu^\rho_t, \nu^\rho_t,\phi_t), \\
        \nu^\rho_{t+1} &= \nu^\rho_t G^\rho_t(\mu^\rho_t, \nu^\rho_t,\psi_t),
    \end{aligned}
\end{equation}
where $F_t^\rho \in \mathbb{R}^{|\X| \times |\X|}$ is the transition matrix for a typical Blue agent's distribution under the policy $\phi_t$ 
and its entries are given by $\left[F_t^{\rho}(\mu^\rho_t, \nu^\rho_t,\phi_t)\right]_{pq} = \sum_{u \in \U}f^\rho_t(q|p, u, \mu^\rho_t, \nu^\rho_t ) \; \phi_t(u|p, \mu_t^\rho, \nu_t^\rho)$. 
The matrix $G_t^\rho$ is defined similarly.

The following lemma shows that the \emph{deterministic} \MF{} above is an approximation of the (stochastic) finite-population \ED{}, and the approximation error goes to zero when $N_1, N_2 \to \infty$.
Thus, we can regard the mean-field as the empirical distribution of an infinite-population team.

\begin{restatable}{lemma}{mfp}
    \label{lmm:mf-aprx}
    Let $\bfX^{N_1}_t$, $\bfY^{N_2}_t$, $\M^{N_1}_t$ and $\N^{N_2}_t$ be the joint states and the corresponding \ED{}s of a finite-population game.
    Denote the next \ED{}s induced by an identical policy pair $(\phi_t, \psi_t) \in \Phi_t \times \Psi_t$ as $(\M_{t+1}^{N_1}, \N_{t+1}^{N_2})$.
    Then, the following bounds hold:
    \begin{align*}
        \expct{\dtv{\M_{t+1}^{N_1}, \M_{t+1}} \big\vert \; \bfX^{N_1}_t, \bfY^{N_2}_t}{\phi_t} \leq \frac{|\X|}{2}\sqrt{\frac{1}{N_1}},\\
        \expct{\dtv{\N_{t+1}^{N_2}, \N_{t+1}} \big\vert \; \bfX^{N_1}_t, \bfY^{N_2}_t}{\psi_t}  \leq \frac{|\Y|}{2}\sqrt{\frac{1}{N_2}},
    \end{align*}
    where 
    $\M_{t+1} = \M^{N_1}_{t} F_t^{\rho}(\M_{t}^{N_1}, \N_{t}^{N_2},\phi_t)$ and
    $\N_{t+1} = \N^{N_2}_t G^\rho_t(\M_{t}^{N_1}, \N_{t}^{N_2},\psi_t)$ according to~\eqref{eqn:coordinator-dynamics}.
\end{restatable}
\begin{proof}
    See Appendix~\ref{appdx:prop-mf-apprx}
\end{proof}

\subsection{Infinite-Population Optimization}
For the infinite-population game, the cumulative reward induced by the strategy pair $(\phi, \psi) \in \Phi \times \Psi$ is given by 
\begin{equation}
    \label{eqn:mf-optimization-phi-psi}
    J^{\rho, \phi, \psi}(\mu^\rho_0, \nu^\rho_0) = \sum_{t=0}^{T} r_t^\rho(\mu^\rho_t, \nu^\rho_t),
\end{equation}
where the propagation of the mean-fields $\mu^\rho_t$ and $\nu^\rho_t$ are subject to the dynamics~\eqref{eqn:mf-dynamics}.
The worst-case performance of the Blue team is given by the lower game value
\begin{equation}\label{eqn:mf-optimization-lower}
    \underline{J}^{\rho*} (\mu^\rho_0, \nu^\rho_0) = \max_{\phi\in \Phi} \; \min_{\psi \in \Psi} ~ J^{\rho, \phi, \psi}(\mu^\rho_0, \nu^\rho_0).
\end{equation}

The Red team instead uses the upper game value based on the following min-max optimization
\begin{equation}\label{eqn:mf-optimization-upper}
    \bar{J}^{\rho*} (\mu^\rho_0, \nu^\rho_0) =  \min_{\psi \in \Psi} \max_{\phi \in \Phi}~ J^{\rho, \phi, \psi}(\mu^\rho_0, \nu^\rho_0).
\end{equation}

\begin{remark}
    Different from the finite-population value~\eqref{eqn:origin-optimization} which takes joint states as argument, the infinite-population game value~\eqref{eqn:mf-optimization-lower} takes \MF{}s. 
    This is due to the assumed identical team strategy, which no longer requires agents' index information to differentiate them and sample actions. Consequently, the \MF{}s can serve as the information state as in~\citep{arabneydi2015team}.
    The difference comes from the non-identical strategies considered in the finite-population game, which require each agent's state and index information to sample actions and predict the game's evolution.
\end{remark}

\section{Zero-Sum Game Between Coordinators}
\label{sec:coordinator-game}

The mean-field sharing structure in~\eqref{eqn:IS} allows us to reformulate the \emph{infinite}-population competitive team problem~\eqref{eqn:mf-optimization-lower} as an equivalent two-player game from the perspective of two fictitious%
\footnote{The coordinators are fictitious since they are introduced as an auxiliary concept and are not required for the actual implementation of the obtained strategies.}
coordinators.
The coordinators know the common information (\MF{}s) and selects a local policy that maps each agent's local information (individual state) to its actions. 
Through this common-information approach~\citep{nayyar2013decentralized}, we provide a dynamic program that constructs optimal strategies for all agents under the original mean-field sharing information structure. 

\subsection{Equivalent Centralized Problem}

We use $\pi_t : \U \times \X \to [0,1]$ to denote a local Blue policy, which is \emph{open-loop} with respect to the \MF{}s. 
Specifically, $\pi_t (u|x)$ is the probability that a Blue agent selects action $u$ at state $x$ regardless of the current \MF{}s.
The set of open-loop Blue local policies is denoted as $\Pi_t$. 
Similarly, $\sigma_t : \V \times \Y \to [0,1]$ and $\Sigma_t$ denote a Red local policy and its admissible set.
Under the local policy $\pi_t$, the Blue \MF{} propagates as
\begin{equation}
    \label{eqn:mf-dynamics-local}
    \mu^\rho_{t+1} (x') \!= \sum_{x \in \X} \Big[ \sum_{u \in \U}f^\rho_t(x'|x, u, \mu^\rho_t, \nu^\rho_t ) \pi_t(u|x)\Big]\mu^\rho_t(x),
\end{equation}   
and the Red team \MF{} dynamics under Red local policies is defined similarly as
\begin{equation*}
    \nu^\rho_{t+1} (y') \!= \sum_{y \in \Y} \Big[ \sum_{v \in \V}g^\rho_t(y'|y, v, \mu^\rho_t, \nu^\rho_t ) \sigma_t(v|y)\Big]\nu^\rho_t(y).
\end{equation*}  

At each time $t$, a Blue coordinator observes the \MF{}s of both teams (common information) and prescribes a local policy $\pi_t\in \Pi_t$ to all Blue agents within its team. 
The local policy is selected based on:
\begin{equation*}
    \pi_t = \alpha_t \big(\mu^\rho_t,\nu^\rho_t \big),
\end{equation*}
where  $\alpha_t: \P(\X) \times \P(\Y) \to \Pi_t$ is a deterministic Blue \textit{coordination policy},
and $\pi_t(u_t|x_t) \triangleq \alpha_t(\mu^\rho_t,\nu^\rho_t)(u_t|x_t)$ gives the probability that a Blue agent selects action $u_t$ given its current state $x_t$.
Similarly, the Red coordinator observes the \MF{}s and selects a local policy $\sigma_t \in \Sigma_t$ according to
$\sigma_t = \beta_t \big(\mu^\rho_t, \nu^\rho_t\big)$.
We refer to the time sequence $\alpha \hspace{-0.02in}=\hspace{-0.02in}\big(\alpha_1, \ldots, \alpha_{T-1}\big)$ as the Blue team \textit{coordination strategy}
and $\beta=\big(\beta_1, \ldots, \beta_{T-1}\big)$ as the Red team coordination strategy.
The sets of admissible coordination strategies are denoted as $\A$ and $\B$.

\begin{remark}
    \label{rmk:equivalent-policy}
    There is a one-to-one correspondence between the coordination strategies and the identical team strategies.
    For example, given an \emph{identical} Blue team strategy $\phi \in \Phi$, one can define the coordination strategy:
    \begin{equation*}
        \alpha_t (\mu_t, \nu_t) = \pi_t \quad \st \pi_t(u_t|x_t) = \phi_t(u_t|x_t, \mu_t, \nu_t) \quad \forall \mu_t \in \P(\X), \; \nu_t \in \P(\Y),\; x_t \in \X \text{ and } u_t \in \U.
    \end{equation*}
    Similarly, a Blue coordination strategy $\alpha \in \A$ induces an identical team strategy $\phi \in \Phi$ according to the rule
    \begin{equation*}
        \phi_{t}(u_t|x_t, \mu^\rho_t, \nu^\rho_t) = \underbrace{\Big(\alpha_t(\mu^\rho_t,\nu^\rho_t)\Big)}_{\displaystyle{\pi_t}}(u_t|x_t) \quad \forall \mu_t \in \P(\X), \; \nu_t \in \P(\Y),\; x_t \in \X \text{ and } u_t \in \U.
    \end{equation*}

\end{remark}


Plugging in the \textit{deterministic} coordination strategies, the mean-fields dynamics in~\eqref{eqn:mf-dynamics-local} becomes
\begin{subequations}
\label{eqn:cor-mf-dynamics}
\begin{alignat}{2}
    \mu^\rho_{t+1} &= \mu^\rho_t F_t^\rho (\mu^\rho_t, \nu^\rho_t, \underbrace{\alpha_t(\mu^\rho_t,\nu^\rho_t)}_{\displaystyle{\pi_t}}), 
    \\
    \nu^\rho_{t+1} &=\nu^\rho_t G^\rho_t(\mu^\rho_t, \nu^\rho_t, \beta_t(\mu^\rho_t,\nu^\rho_t)).
\end{alignat}
\end{subequations}
For notational simplicity, we will use the shorthand notations $F_t^\rho (\mu^\rho_t, \nu^\rho_t, \alpha_t)$ and $G^\rho_t(\mu^\rho_t, \nu^\rho_t, \beta_t)$.

The original competitive team problem in~\eqref{eqn:mf-optimization-lower} can now be viewed as an equivalent zero-sum game played between the two coordinators,
where the game state is the joint mean-field $(\mu_t^\rho, \nu_t^\rho)$,
and the actions are the local policies $\pi_t$ and $\sigma_t$ selected by the coordinators. Figure~\ref{fig:schematic} provides a schematic of the equivalent centralized system.

Formally, the zero-sum coordinator game can be defined via a tuple 
\begin{equation}
    \label{eqn:coordinator-game-tuple}
    \texttt{ZS-CG} = \langle \P(\X), \P(Y), \Pi_t, \Sigma_t, F^\rho_t, G^\rho_t, r^\rho_t, \rho, T\rangle.
\end{equation}
In particular, the continuous game state space is $\P(\X) \times \P(Y)$, and the continuous action spaces are $\Pi_t$ and $\Sigma_t$ for the Blue and Red coordinators, respectively. 
The deterministic dynamics of the game is given in~\eqref{eqn:cor-mf-dynamics}. 
Finally, the reward structure $r^\rho_t$ is given in~\eqref{eqn:reward}, and the horizon $T$ is the same as in the original game. 

\begin{figure}[t]
    \vspace{-0.2in}
    \centering
    \includegraphics[width = 0.6\linewidth]{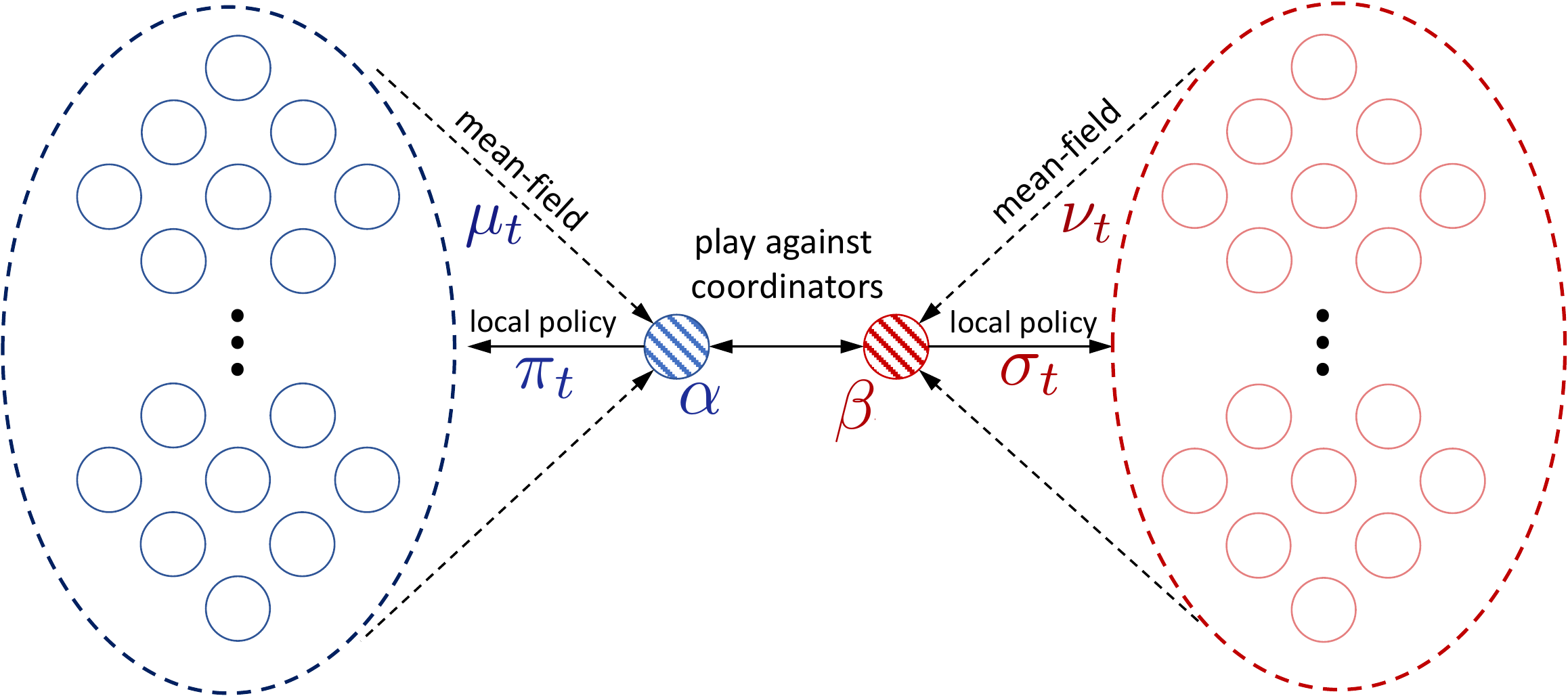}
    \caption{A schematic of the proposed common-information approach for the mean-field zero-sum team games.}
    \label{fig:schematic}
    \vspace{-0.1in}
\end{figure}

Given two coordination strategies $\alpha \in \A$ and $\beta \in \B$, the cumulative rewards of the coordinator game is defined as 
\begin{equation}
    \label{eqn:mf-optimization-alpha-beta}
    J_\cor^{\rho, \alpha, \beta}(\mu^\rho_0, \nu^\rho_0) = \sum_{t=0}^{T} r_t^\rho(\mu^\rho_t, \nu^\rho_t),
\end{equation}
where the propagation of the mean-fields $\mu^\rho_t$ and $\nu^\rho_t$ are subject to the dynamics~\eqref{eqn:cor-mf-dynamics}.
The worst-case performance of the Blue team is given by the lower coordinator game value
\begin{equation}\label{eqn:cor-mf-optimization-lower}
    \underline{J}_\cor^{\rho*} (\mu^\rho_0, \nu^\rho_0) = \max_{\alpha \in \A} \; \min_{\beta \in \B} ~ J_\cor^{\rho, \alpha, \beta}(\mu^\rho_0, \nu^\rho_0),
\end{equation}
and the upper value for the Red team is based on the following min-max optimization
\begin{equation}\label{eqn:cor-mf-optimization-upper}
    \bar{J}_\cor^{\rho*} (\mu^\rho_0, \nu^\rho_0) =  \min_{\beta \in \beta} \max_{\alpha \in \A}~ J_\cor^{\rho, \alpha, \beta}(\mu^\rho_0, \nu^\rho_0).
\end{equation}


In the next section, we examine the properties of the max-min (lower) and min-max (upper) game values of the coordinator game.

\subsection{Value Functions of the Coordinator Game}
\label{sec:optimal-value}

Similar to the standard two-player zero-sum games, we use a backward induction scheme to find the lower and upper values of the coordinator game. 
In general, the max-min value need not be equal to the min-max value~\citep{elliott1972existence}.
Consequently, from the Blue coordinator's perspective, we consider the lower value (max-min), which provides the highest guaranteed performance for the maximizing Blue team in a worst-case scenario.

The terminal lower value at time $T$ is defined as 
\begin{equation}
    \label{eqn:J-ab-T}
    \lowervalue_{\cor,T}^{\rho *}(\mu^\rho_T, \nu^\rho_T) = r^\rho_T(\mu^\rho_T, \nu^\rho_T).
\end{equation}
For all previous time steps $t=0,\ldots, T-1$, the two coordinators optimize their cumulative reward function by choosing their actions (i.e., local policies) $\pi_t$ and $\sigma_t$.
Consequently, for all $t=0,\ldots,T-1,$ we have
\begin{equation}
    \label{eqn:j-ab-t-blue}
    \lowervalue_{\cor,t}^{\rho*}(\mu^\rho_t, \nu^\rho_t) = r^\rho_t(\mu^\rho_t, \nu^\rho_t) +
    \max_{\pi_t \in \Pi_t}~ \min_{\sigma_t \in \Sigma_t}
    \lowervalue_{\cor,t+1}^{\rho*}\big(\mu^\rho_t F^\rho_t(\mu^\rho_t, \nu^\rho_t, \pi_t), \nu^\rho_t G^\rho_t(\mu^\rho_t, \nu^\rho_t, \sigma_t)\big).
\end{equation}

With the optimal value function, the optimal Blue team coordination policies can then be easily constructed via
\begin{align}
    \label{eqn:optimal-blue-policy}
    \alpha^*_t(\mu^\rho_t, \nu^\rho_t) &\in \argmax_{\pi_t \in \Pi_t}~ \min_{\sigma_t \in \Sigma_t}
    \lowervalue_{\cor,t+1}^{\rho*}\big(\mu^\rho_t F_t^{\rho}(\mu^\rho_t, \nu^\rho_t,\pi_t), \nu_t^\rho G^\rho_t(\mu^\rho_t, \nu^\rho_t, \sigma_t)\big).
\end{align}

Similarly, for the Red team coordinator, the upper values are computed as 
\begin{subequations}
    \label{eqn:j-ab-t-red}
    \begin{align}
    &\uppervalue_{\cor,T}^{\rho *}(\mu^\rho_T, \nu^\rho_T) = r^\rho_T(\mu^\rho_T, \nu^\rho_T),
    \\
    &\uppervalue_{\cor,t}^{\rho*}(\mu^\rho_t, \nu^\rho_t) = r_t(\mu^\rho_t, \nu^\rho_t) +
    \min_{\sigma_t \in \Sigma_t} \max_{\pi_t \in \Pi_t} 
    \uppervalue_{\cor,t+1}^{\rho*}\big(\mu^\rho_t F^\rho_t(\mu^\rho_t, \nu^\rho_t, \pi_t), \nu^\rho_t G^\rho_t(\mu^\rho_t, \nu^\rho_t, \sigma_t)\big),~~ t=0,\ldots,T-1,
    \end{align}
\end{subequations}
and the optimal Red team coordination policy is given by
\begin{align}
    \label{eqn:optimal-red-policy}
    \beta^*_t(\mu^\rho_t, \nu^\rho_t) &\in  \argmin_{\sigma_t \in \Sigma_t}~\max_{\pi_t \in \Pi_t}
    \uppervalue_{\cor,t+1}^{\rho*}\big(\mu^\rho_t F_t^{\rho}(\mu^\rho_t, \nu^\rho_t,\pi_t), G^\rho_t(\mu^\rho_t, \nu^\rho_t, \sigma_t)\big).
\end{align}

Note that the optimal Blue team coordination strategy induces an identical Blue team strategy that satisfies the mean-field sharing information structure and can be implemented in the finite-population game (Remark~\ref{rmk:equivalent-policy}).

\subsection{Reachable Sets} \label{sec:reachable-set}
At the infinite-population limit, the \MF{} dynamics in~\eqref{eqn:coordinator-dynamics} is deterministic, 
and thus selecting the local policies $\pi_t$ and $\sigma_t$ at time $t$ is equivalent to selecting the desirable \MF{}s at the next time step. 
Consequently, we examine the set of \MF{}s that can be reached from the current \MF{}s.

\begin{definition}
    \label{def:reachable-set}
    The Blue reachable set, starting from the mean-fields $\mu^\rho_t$ and $\nu^\rho_t$, is defined as the set comprising all the possible next Blue team mean-fields  $\mu^\rho_{t+1}$ that can be achieved by employing a local policy $\pi_t \in \Pi_t$. Formally, 
    \begin{equation}
        \label{eqn:blue-reachable-set-def}
        \RMSet{t}{\rho} = \left\{\mu^\rho_{t+1} \;\vert\; \exists \pi_t \in \Pi_t \st \mu^\rho_{t+1} = \mu^\rho_t F^{\rho}_t(\mu^\rho_t,\nu^\rho_t, \pi_t)\right\}.
    \end{equation}
    Similarly, the Red reachable set is defined as 
    \begin{equation}
        \label{eqn:red-reachable-set-def}
         \RNSet{t}{\rho} = \left\{\nu^\rho_{t+1} \; \vert \; \exists \sigma_t \in \Sigma_t \st \nu^\rho_{t+1} = \nu^\rho_t G^{\rho}_t(\mu^\rho_t,\nu^\rho_t, \sigma_t)\right\}.
    \end{equation}
\end{definition}

We will regard the reachable sets as set-valued functions (correspondences)~\citep{freeman2008robust}. 
In this case, we write $\R^\rho_{\mu,t} : \P(\X) \times \P(\Y) \rightsquigarrow \P(\X)$,
and similarly $\R^\rho_{\nu,t} : \P(\X) \times \P(\Y) \rightsquigarrow \P(\Y)$.

The following lemma justifies using the reachable sets constructed based on the local policies to analyze the reachability of identical team policies.
\begin{lemma}
    For all $\mu_t^\rho \in \P(\X)$ and $\nu_t^\rho \in \P(\Y)$, we have that 
    \begin{equation}
        \left\{\mu^\rho_{t+1} \;\vert\; \exists \phi_t \in \Phi_t \st \mu^\rho_{t+1} = \mu^\rho_t F^{\rho}_t(\mu^\rho_t,\nu^\rho_t, \phi_t)\right\} = \RMSet{t}{\rho}.
    \end{equation}
\end{lemma}
\begin{proof}
    This lemma is a direct consequence of Remark~\ref{rmk:equivalent-policy}.
    
\end{proof}

\begin{remark}
    Note that the reachable sets are constructed based on identical team strategies, since under the coordinator game formulation, all agents in the same team follow the same local policies prescribed by their coordinator.
\end{remark}

\subsection{Equivalent Form of Value Functions} 
\label{subsec:equivalent-opt}

We can now change the optimization domains in~\eqref{eqn:J-ab-T} and~\eqref{eqn:j-ab-t-blue} from the policy spaces to the corresponding reachable sets. 
One can easily see that the following lower value propagation scheme is equivalent to the one in~\eqref{eqn:J-ab-T} and~\eqref{eqn:j-ab-t-blue}.
\begin{equation}
    \label{eqn:lower-j-ab-t-rset}
    \begin{aligned}
        \lowervalue_{\cor,T}^{\rho*}(\mu^\rho_T, \nu^\rho_T) &= r^\rho_T(\mu^\rho_T, \nu^\rho_T)\\
        \lowervalue_{\cor,t}^{\rho*}(\mu^\rho_t, \nu^\rho_t) &= r^\rho_t(\mu^\rho_t, \nu^\rho_t) +
        \max_{\mu_{t+1}^\rho \in \mathcal{R}_{\mu,t}^\rho(\mu_t^\rho,\nu_t^\rho)}~ \min_{\nu_{t+1}^\rho \in \mathcal{R}_{\nu,t}^\rho(\mu_t^\rho,\nu_t^\rho)}
        \lowervalue_{\cor,t+1}^{\rho*}(\mu^\rho_{t+1}, \nu^\rho_{t+1}), \quad t= 0,\ldots, T-1.
    \end{aligned}
\end{equation}

Given the current mean-fields $(\mu_t^\rho, \nu_t^\rho)$, the optimal next mean-field to achieve is then given by 
\begin{equation*}
    \mu_{t+1}^{\rho*} \in \argmax_{\mu_{t+1}^\rho \in \mathcal{R}_{\mu,t}^\rho(\mu_t^\rho,\nu_t^\rho)}~ \min_{\nu_{t+1}^\rho \in \mathcal{R}_{\nu,t}^\rho(\mu_t^\rho,\nu_t^\rho)}
        \lowervalue_{\cor,t+1}^{\rho*}(\mu^\rho_{t+1}, \nu^\rho_{t+1}).
\end{equation*}
The Blue coordination policy that leads to $\mu_{t+1}^{\rho*}$ can be constructed via a simple linear program leveraging the linearity of the mean-field dynamics with respect to the policy.
See~Appendix~\ref{appdx-sec:LP} for details.

In the sequel, we primarily work with the reachability-based optimization in~\eqref{eqn:lower-j-ab-t-rset}.
There are two advantages to this approach:
First, the reachable sets generally have a lower dimension than the coordinator action spaces,
\footnote{The Blue reachable set is a subset of $\P(\X)$, while the Blue coordinator action space is given by $\Pi_t = (\P(\U))^{|\X|}$.}, 
which is desirable for numerical algorithms;
Second, the reachability-based optimization allows us to 
compare the ``reachability" induced by non-identical and identical team strategies (Theorem~\ref{thm:R-set-rich} in Section~\ref{sec:performance}) and then study the performance loss due to the identical strategy assumption.

\subsection{Existence of the Coordinator Game Value}
\label{sec:existence-game-value}
In general, the coordinator game value may not exist, i.e, the max-min and min-max values differ (see Numerical Example 1 in Section~\ref{sec:example}). 
However, we can show that the existence of coordinator game value for a special class of mean-field team games, where the dynamics of each individual agent is independent of other agents' states (both its teammates and its opponents).
It is left as a future research direction to obtain more general conditions that ensure the existence of the game value.

\begin{restatable}{definition}{IDD}
    \label{def:independent-dynamics}
        We say that the weakly-coupled dynamics are independent if the following holds for all $\mu_t \in \P(\X)$, $\nu_t \in \P(\Y)$ and $t\in \{0, \ldots, T-1\}$,
        \begin{equation}
            \label{eqn:independent-dynamics}
            \begin{aligned}
                f^\rho_t(x_{t+1} \vert x_t, u_t, \mu_t, \nu_t) &= \bar{f}_t(x_{t+1} \vert x_t, u_t),
                \\
                g^\rho_t(y_{t+1}\vert y_{t}, v_{t}, \mu_t, \nu_t) &= \bar{g}_t(y_{t+1}\vert y_{t}, v_{t}),
            \end{aligned}
        \end{equation}
        for some transition kernels $\bar{f}_t$ and $\bar{g}_t$. 
\end{restatable}

In other words, the dynamics of an agent only depends on that agent's current state and action and is fully decoupled from all other agents. However, we still allow reward coupled via the \MF{}s.
The following theorem ensures the existence of the game value under independent dynamics. 

\begin{restatable}{theorem}{EGV}
    \label{thm:game-value}
    Suppose the reward function $r^\rho_t$ is concave-convex for all $t \in \{0, \ldots, T\}$ and the system dynamics is independent.
    Then, 
    the game value exists. 
\end{restatable}

\begin{proof}
    One can show that the optimal value under independent dynamics is concave-convex and apply the minimax theorem. 
    The detailed proof is given in Appendix~\ref{appdx-sec:game-value}.
\end{proof}

\section{Main Results}  \label{sec:performance}
Recall that the optimal Blue team coordination strategy $\alpha^*$ is constructed for the infinite-population game assuming that both teams employ identical team strategies.
This section establishes the performance guarantees for $\alpha^*$ in the finite-population games where both teams are allowed to deploy non-identical strategies.

\subsection{Approximation Error}

As $\alpha^*$ is solved at the infinite-population limit, it is essential to understand how well the infinite-population game approximates the original finite-population problem.
In this subsection, we show that that the reachable set constructed using identical strategies and the mean-field dynamics is rich enough to approximate any empirical distributions induced by \emph{non-identical} team strategies in finite-population games. 
We start with constructing local policies that mimics the behaviors of non-identical team strategies.

\begin{restatable}{lemma}{mfa}
    \label{lmm:mf-apprx-team-policy}
    Let $\bfX^{N_1}_t$, $\bfY^{N_2}_t$, $\M^{N_1}_t$ and $\N^{N_2}_t$ be the joint states and the corresponding \ED{}s of a finite-population game.
    Given any Blue team policy $\phi^{N_1}_t \in \Phi^{N_1}_t$ (potentially non-identical),
    define the following local policy
    \begin{equation}
        \label{eqn:apprx-local-policy}
        \pi_{\apprx, t} (u|x) = \left\{
        \begin{array}{ll}
             \frac{\sum_{i=1}^N \indicator{x}{X_{i,t}^{N_1}} \phi_{i,t}(u|x,\M^{N_1}_t, \N^{N_2}_t)}{N_1 \M^{N_1}_t(x)} \quad 
             &  \text{if } ~~ \M_t^{N_1}(x) >0, \\
             {1}/{|\U|} \quad 
             &  \text{if } ~~ \M_t^{N_1}(x) =0.
        \end{array}
        \right .
    \end{equation}
    Further, define the next mean-field induced by $\pi_{\apprx, t}$ as 
    \begin{equation}
        \label{eqn:blue-apprx-mf}
        \M_{\apprx, t+1} = \M^{N_1}_t F^\rho_t(\M^{N_1}_t, \N^{N_2}_t, \pi_{\apprx, t}).
    \end{equation}
    Then, the expected distance between the next Blue \ED{} $\M^{N_1}_{t+1}$ induced by the team policy $\phi^{N_1}_t$ and the mean-field $\M_{\apprx, t+1}$ satisfies
    \begin{equation}
        \label{eqn:mf-apprx-joint-state}
        \mathbb{E}\left[ \dtv{
    \M^{N_1}_{t+1}, \M_{\apprx, t+1}}\vert \bfX_t^{N_1}, \bfY_t^{N_2}\right] \leq \frac{|\X|}{2} \sqrt{\frac{1}{N_1}}.
    \end{equation}
\end{restatable}

\begin{proof}
    See Appendix~\ref{appdx:mf-reachability}.
\end{proof}

The local policy $\pi_{\apprx, t}$ mimics the population behavior by setting its action distribution at state $x$ as the average of the policies used by the Blue agents at state $x$.
The purpose of the second case in~\eqref{eqn:apprx-local-policy} is to ensure that the constructed local policy is well-defined at states where no Blue agent is present.
The mean-field induced by $\pi_{\apprx,t}$ is within the reachable set and is close to the next \ED{} induced by the non-identical team policy $\phi^{N_1}_t$ in expectation.
This idea is visualized in~\figref{fig:apprx}.

\begin{figure}[b]
    \centering
    \includegraphics[width=0.4\linewidth]{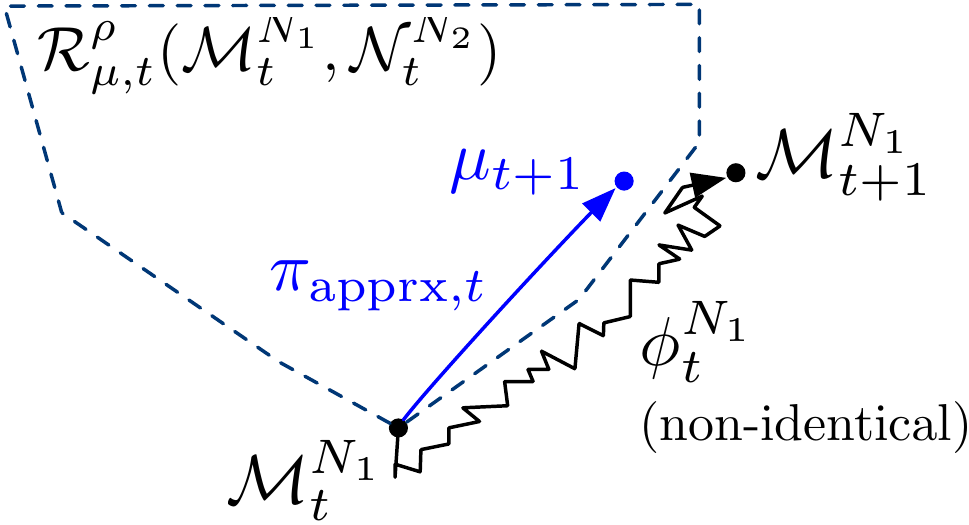}
    \vspace{-0.1in}
    \caption{An illustration of Lemma~\ref{lmm:mf-apprx-team-policy}.}
    \label{fig:apprx}
\end{figure}

Lemma~\ref{lmm:mf-apprx-team-policy} directly leads to the following theorem regarding the richness of the reachable sets, as the mean-field induced by $\pi_{\apprx,t}$ is within the reachable set.

\begin{theorem}
    \label{thm:R-set-rich}
    Let $\bfX^{N_1}_t$, $\bfY^{N_2}_t$, $\M^{N_1}_t$, and $\N^{N_2}_t$ be the joint states and the corresponding \ED{}s at time~$t$.
    Denote the next Blue \ED{} induced by some Blue team policy $\phi^{N_1}_t \in \Phi^{N_1}_t$ as $\M^{N_1}_{t+1}$. 
    Then, there exists a mean-field $\mu_{t+1} \in \mathcal{R}^\rho_{\mu,t}(\M^{N_1}_t, \N^{N_2}_t)$ that satisfies
    \begin{equation}
        \label{eqn:apprx-exist-mu-apprx}
        \expct{\dtv{\M^{N_1}_{t+1}, \mu_{t+1}} \big \vert \bfX^{N_1}_t, \bfY^{N_2}_t}{\phi_t^{N_1}} \leq \frac{|\X|}{2} \sqrt{\frac{1}{N_1}}.
    \end{equation}
\end{theorem}


\begin{remark}
    The construction of $\pi_{\apprx,t}$ in~\eqref{eqn:apprx-local-policy} requires knowing 
    each Blue agent's state, which may seem to violate the mean-field sharing information structure in~\eqref{eqn:IS}. 
    However, $\pi_{\apprx,t}$ only serves as an auxiliary concept to prove the existence of a mean-field in the reachable set that satisfies~\eqref{eqn:apprx-exist-mu-apprx}.
    The existence result in Theorem~\ref{thm:R-set-rich} is all we need to provide performance guarantees.
    In fact, $\pi_{\apprx,t}$ will not be used to construct the optimal policies.
\end{remark}

\begin{remark}
    All results in this subsection extend to the analysis from Red team's side.
\end{remark}

\subsection{Lipschitz Continuity of the Value Functions}

Next, we examine the continuity of the optimal value function in~\eqref{eqn:lower-j-ab-t-rset} with respect to the mean-field arguments, which is essential for the performance guarantees. 
Clearly, the continuity of the value function depends on the continuity of the two reachability correspondences $\R^\rho_{\mu,t}$ and $\R^\rho_{\nu,t}$.
To properly study the continuity of the reachability correspondences, we use the Hausdorff distance to measure the distance between two sets.

\begin{definition}[Hausdorff distance]
    For a normed space $(\X, \norm{\cdot})$, the Hausdorff distance between the sets $A, B \subseteq \X$ is defined as
    \begin{equation}
        \distH{A,B} = \max\left\{\sup_{a \in A} \inf_{b \in B} \norm{a-b}, \sup_{b \in B} \inf_{a \in A} \norm{a-b}\right\}.
    \end{equation}
\end{definition}

Based on the Hausdorff distance, we have the following notion of Lipschitz continuity for correspondences~\citep{freeman2008robust}.
\begin{definition}
    A correspondence $\Gamma: \X \rightsquigarrow \Z$ is $L_\Gamma$-Lipschitz continuous under the Hausdorff distance if, for all $x, x' \in \X$, it satisfies that
    \begin{equation}
        \distH{\Gamma(x), \Gamma(x')} \leq L_\Gamma \norm{x-x'}.
    \end{equation}
\end{definition}

The following lemma presents a Lipschitz continuity result for the reachability correspondences.

\begin{restatable}{lemma}{RLC}
    \label{lmm:R-L-cont}
    The reachability correspondences $\mathcal{R}_{\mu,t}$ and $\mathcal{R}_{\nu,t}$ in~\eqref{eqn:blue-reachable-set-def} and~\eqref{eqn:red-reachable-set-def} satisfy the following inequalities for all $\mu_t, \mu'_t \in \P(\X)$ and $\nu_t, \nu'_t \in \P(\Y)$,
    \begin{equation}
        \label{eqn:reachability-continuity-1}
        \begin{aligned}
            \distH{\R^\rho_{\mu,t}(\mu_t,\nu_t), \R^\rho_{\mu,t}(\mu'_t,\nu'_t)} \leq L_{R_{\mu,t}} \big(\dtv{\mu_t, \mu'_t} + \dtv{\nu_t, \nu'_t}\big), \\
            \distH{\R^\rho_{\nu,t}(\mu_t,\nu_t), \R^\rho_{\nu,t}(\mu'_t,\nu'_t)} \leq L_{R_{\nu,t}} \big(\dtv{\mu_t, \mu'_t} + \dtv{\nu_t, \nu'_t}\big),
        \end{aligned}
    \end{equation}
    where $L_{R_{\mu,t}}  = 1+ \frac{1}{2} L_{f_t}$, $L_{R_{\nu,t}} = 1+ \frac{1}{2} L_{g_t}$, and $L_{f_t}$ and $L_{g_t}$ are the Lipschitz constants in Assumption~\ref{assmpt:lipschitiz-dynamics}.
\end{restatable}
\begin{proof}
    The proof is postponed to Appendix~\ref{appdx-sec:Continuity-RSet}.
\end{proof}

Leveraging the continuity of the reachability correspondences, the following theorem establishes the Lipschitz continuity of the optimal value functions. 

\begin{restatable}{theorem}{vlc}
    \label{thm:value-L-cont}
    The optimal lower value function $\lowervalue^{\rho*}_{\cor, t}$ and the optimal upper value function $\uppervalue^{\rho*}_{\cor, t}$ are both Lipschitz continuous. 
    Formally, for all $\mu_{t}^\rho,\mu^{\rho\prime}_t \in \P(\X)$ and $\nu_{t}^\rho,\nu^{\rho\prime}_t \in \P(\Y)$, the following inequalities hold
    \begin{equation}
        \label{eqn:L-cont-cor-value}
        \begin{aligned}
            \abs{\lowervalue^{\rho*}_{\cor, t}(\mu^\rho_{t}, \nu^\rho_t) - \lowervalue^{\rho*}_{\cor, t}(\mu^{\rho\prime}_{t}, \nu^{\rho\prime}_t)} &\leq 
        L_{J,t}\big(\dtv{\mu^\rho_t, \mu_t^{\rho\prime}}+\dtv{\mu^\rho_t, \nu_t^{\rho\prime}}\big),
        \\
        \abs{\uppervalue^{\rho*}_{\cor, t}(\mu^\rho_{t}, \nu^\rho_t) - \uppervalue^{\rho*}_{\cor, t}(\mu^{\rho\prime}_{t}, \nu^{\rho\prime}_t)} &\leq 
        L_{J,t}\big(\dtv{\mu^\rho_t, \mu_t^{\rho\prime}}+\dtv{\mu^\rho_t, \nu_t^{\rho\prime}}\big),
        \end{aligned}
    \end{equation}
    where the Lipschitz constant $L_{J,t}$ is given by
    \begin{equation}
        \label{eqn:value-L-constant}
        L_{J,t} = L_r \big( 1+ \sum_{k=t}^{T-1} \prod_{\tau =t}^{k} (L_{\R^\rho_{\mu,\tau}} + L_{\R^\rho_{\nu,\tau}})\big).
    \end{equation}
\end{restatable}

\begin{proof}
    Observe that the lower value in~\eqref{eqn:lower-j-ab-t-rset} takes the  form:
    $f(x,y) = \max_{p\in \Gamma(x,y)} \min_{q\in \Theta(x,y)} g(p,q)$,
    which can be viewed as an extension of the marginal function~\citep{freeman2008robust} to the max-min case.
    We present a continuity result for this type of marginal function in Lemma~\ref{lmm:min-max-marginal-L-cont} in Appendix~\ref{appdx-sec:Continuity-Prelim}.
    Based on Lemma~\ref{lmm:min-max-marginal-L-cont}, we can prove the Lipschitz continuity result through an inductive argument since the terminal rewards are assumed to be Lipschitz.
    A detailed proof is given in Appendix~\ref{appdx-sec:Continuity-RSet-thm-2}.
\end{proof}

\subsection{Performance Guarantees}

In the previous section, we obtained an optimal Blue coordination strategy $\alpha^*$ for the coordinator game.
The coordinator game is constructed under the mean-field setting, i.e., both teams have an infinite number of agents and all agents in each team apply the same strategy.
We analyze the performance guarantees for $\alpha^*$ in the \emph{finite-population} game and compare the worst-case performance of this coordinator-induced Blue team strategy to the original max-min optimization in~\eqref{eqn:origin-optimization} where agents have the flexibility to apply non-identical strategies within each team.

To this end, let the Red team apply non-identical team strategies $\psi^{N_2} = (\psi^{N_2}_1, \ldots \psi^{N_2}_{N_2}) \in \Psi^{N_2}$. 
Recall that a Blue coordination strategy induces an identical Blue team strategy (see Remark~\ref{rmk:equivalent-policy}).
We denote the finite-population performance under the identical Blue team strategy induced by $\alpha^*$ and the Red team strategy $\psi^{N_2}$ as 
\begin{equation*}
    J_{0}^{N, \alpha^*, \psi^{N_2}}(\bfx^{N_1}_0, \bfy^{N_2}_0) = \mathbb{E}_{\alpha^*, \psi^{N_2}}\Big[\sum_{t =0}^T r^\rho_t(\M^{N_1}_t, \N^{N_2}_t) \vert \bfX^{N_1}_0 = \bfx^{N_1}_0, \bfY^{N_2}_0 = \bfy^{N_2}_0\Big].
\end{equation*}
Through dynamic programming, we can compute the above-induced value through a backward propagation scheme 
\begin{subequations}
    \begin{align}
        J_{T}^{N, \alpha^*, \psi^{N_2}}(\bfx^{N_1}_T, \bfy^{N_2}_T) 
        &= r^\rho_T(\mu^{N_1}_T, \nu^{N_2}_T) \\
        J_{t}^{N, \alpha^*, \psi^{N_2}}(\bfx^{N_1}_t, \bfy^{N_2}_t) 
        &= r^\rho_t(\mu^{N_1}_t, \nu^{N_2}_t) + \mathbb{E}_{\alpha^*, \psi^{N_2}}\Big[J_{t+1}^{N, \alpha^*, \psi^{N_2}}(\bfX^{N_1}_{t+1}, \bfY^{N_2}_{t+1})  \big \vert  \bfX^{N_1}_{t}= \bfx^{N_1}_{t}, \bfY^{N_2}_{t}=\bfy^{N_2}_{t}\;  \Big], \\
        & \qquad \qquad\qquad \qquad \qquad \qquad\qquad \qquad \qquad\qquad\qquad \qquad\qquad\qquad\qquad \forall t = 0, \ldots, T-1, \nonumber
    \end{align}
\end{subequations}
where $\mu_t^{N_1} = \empMu{\bfx^{N_1}_t}$ and $\nu_t^{N_2} = \empNu{\bfy^{N_2}_t}$ are the \ED{}s corresponding to the joint states $\bfx^{N_1}_t$ and $\bfy^{N_2}_t$, respectively.

We have the following main result regarding the performance guarantees for the optimal Blue coordination strategy. 

\begin{theorem}
    \label{thm:performance-guarantees}
    The optimal Blue coordination strategy $\alpha^*$ obtained from~\eqref{eqn:optimal-blue-policy} induces an $\epsilon$-optimal Blue team strategy. 
    Formally, for all $\bfx^{N_1} \in \X^{N_1}, \bfy^{N_2} \in \Y^{N_2}$,
    \begin{equation}
        \label{eqn:performance-bounds}
        \underline{J}^{N*}(\bfx^{N_1}, \bfy^{N_2}) \geq \min_{\psi^{N_2} \in \Psi^{N_2}} J^{N,\alpha^*,\psi^{N_2}} (\bfx^{N_1}, \bfy^{N_2}) \geq \underline{J}^{N*}(\bfx^{N_1}, \bfy^{N_2}) - \residue 
    \end{equation}
    where $\subN = \min \{N_1, N_2\}$.
\end{theorem}

\begin{proof}
    We first prove the first inequality in~\eqref{eqn:performance-bounds}.
    \begin{align*}
        \underline{J}^{N*}(\bfx^{N_1}, \bfy^{N_2}) &= \max_{\phi^{N_1} \in \Phi^{N_1}} \min_{\psi^{N_2} \in \Psi^{N_2}} J^{N,\phi^{N_1},\psi^{N_2}}(\bfx^{N_1}, \bfy^{N_2}) \\
        &\stackrel{\text{(i)}}{\geq} \max_{\phi \in \Phi} \min_{\psi^{N_2} \in \Psi^{N_2}} J^{N,\phi,\psi^{N_2}}(\bfx^{N_1}, \bfy^{N_2})
        \stackrel{\text{(ii)}}{=}\max_{\alpha \in \A} \min_{\psi^{N_2} \in \Psi^{N_2}} J^{N,\alpha,\psi^{N_2}}(\bfx^{N_1}, \bfy^{N_2})\\
        &\geq \min_{\psi^{N_2} \in \Psi^{N_2}} J^{N,\alpha^*,\psi^{N_2}}(\bfx^{N_1}, \bfy^{N_2}),
    \end{align*}
    where inequality (i) is a result of $\Phi \subseteq \Phi^{N_2}$, and equality (ii) is due to the one-to-one correspondence between coordination strategies and identical team strategies (see Remark~\ref{rmk:equivalent-policy}).
    
    For the second inequality in~\eqref{eqn:performance-bounds}, we break it down into two lemmas:
    First, Lemma~\ref{lmm:coordination-esp-nash} states that $\min_{\psi^{N_2}\in \psi^{N_2}} J^{N, \alpha^*, \psi^{N_2}} \geq \lowervalue_{\cor}^{\rho*} - \epsilon_1$;
    and Lemma~\ref{lmm:cor-apprx-err} shows that $\lowervalue_{\cor}^{\rho*} \geq \lowervalue^{N*} - \epsilon_2$;
    finally, it is shown that both error terms are of order $\residue$.
    Combining the two lemmas, we obtain the desired result.
\end{proof}

\begin{remark}
    Recall that $\alpha^*$ is solved at the \emph{infinite-population} limit under the restriction that both teams apply \emph{identical} team strategies. 
    Theorem~\ref{thm:performance-guarantees} states that
    the \emph{identical} Blue team strategy induced by $\alpha^*$ is still $\epsilon$-optimal, even if 
    (i) it is deployed in a \emph{finite-population} game and 
    (ii) the opponent team employs \emph{non-identical} strategies to exploit. 
\end{remark}

\begin{remark}
    Continuity Assumptions~\ref{assmpt:lipschitiz-dynamics} and~\ref{assmpt:lipschitiz-rewards} are necessary to translate the infinite-population performance back to the finite-population game. 
    See Appendix~\ref{appdx-sec:discontinuity} for a discontinuous example where the infinite-population game value is significantly different from that of the finite-population problem.
\end{remark}

\vspace{+0.2in}

\begin{lemma}
    \label{lmm:coordination-esp-nash}
    For all joint states $\bfx^{N_1} \in \X^{N_1}$ and $\bfy^{N_2} \in \Y^{N_2}$, the optimal Blue coordination strategy $\alpha^*$ in~\eqref{eqn:optimal-blue-policy} guarantees
    \begin{equation}
        \label{eqn:deviation-performance-bound}
        \min_{\psi^{N_2}\in \psi^{N_2}} J^{N, \alpha^*, \psi^{N_2}}(\bfx^{N_1}, \bfy^{N_2}) \geq \lowervalue_{\cor}^{\rho*}(\mu^{N_1}, \nu^{N_2}) - \residue,
    \end{equation}
    where $\mu^{N_1}=\empMu{\bfx^{N_1}}$ and $\nu^{N_2} = \empNu{\bfy^{N_2}}$ are the corresponding \ED{}s.
\end{lemma}

\begin{proof}
    The proof is constructed based on induction.
    Fix an arbitrary Red team strategy $\psi^{N_2} \in \Psi^{N_2}$.
    
    \textit{Base case:}
    At the terminal timestep $T$, since there is no decision to be made, both value functions are equal to the terminal reward and are thus the same.
    %
    %
    Formally, for all $\bfx^{N_1}_T \in \X^{N_1}$ and $\bfy^{N_2}_T \in \Y^{N_2}$,
    \begin{equation*}
        J_{T}^{N, \alpha^*, \psi^{N_2}}(\bfx^{N_1}_T, \bfy^{N_2}_T) = \lowervalue_{\cor,T}^{\rho*}(\mu^{N_1}_T, \nu^{N_1}_T)
        = r^\rho_T(\mu^{N_1}_T, \nu^{N_2}_T),
    \end{equation*} 
    where $\mu^{N_1}_T = \empMu{\bfx^{N_1}_T}$ and $\nu^{N_2}_T = \empNu{\bfy^{N_2}_T}$.
    For simplicity, we do not emphasize the correspondence between the joint states and the \ED{}s for the rest of the proof, as it is clear from the context.
    
    \textit{Inductive hypothesis: }
    Assume that at $t+1$, the following holds for all joint states  $\bfx^{N_1}_{t+1} \in \X^{N_1}$ and $\bfy^{N_2}_{t+1} \in \Y^{N_2}$:
    
    \begin{equation}
        J_{t+1}^{N, \alpha^*, \psi^{N_2}} (\bfx^{N_1}_{t+1}, \bfy^{N_2}_{t+1}) \geq 
        \lowervalue_{\cor, t+1}^{\rho*}(\mu^{N_1}_{t+1}, \nu^{N_2}_{t+1}) - \residue.
    \end{equation}
    
    \textit{Induction: }
    At timestep $t$, consider an arbitrary pair of joint states $\bfx^{N_1}_{t} \in \X^{N_1}$ and $\bfy^{N_2}_{t} \in \Y^{N_2}$, and their corresponding \ED{}s $\mu^{N_1}_t$ and $\nu^{N_2}_t$. 
    Define $\mu_{t+1}^{*}$ as
    \begin{equation*}
        \mu_{t+1}^* = \mu^{N_1}_{t} F(\mu^{N_1}_{t}, \nu^{N_2}_{t}, \alpha_t^*).
    \end{equation*}
    Note that, from the optimality of $\alpha_t^*$, we have
    \begin{equation}
        \label{eqn:mu-star-def}
        \mu_{t+1}^* \in \argmax_{\mu_{t+1} \in \mathcal{R}_{\mu,t}^\rho(\mu^{N_1}_{t}, \nu^{N_2}_{t})}~ \min_{\nu_{t+1} \in \mathcal{R}_{\nu,t}^\rho(\mu^{N_1}_{t}, \nu^{N_2}_{t})}
        \lowervalue_{\cor,t+1}^{\rho*}(\mu_{t+1}, \nu_{t+1}).
    \end{equation}

    Furthermore, from Theorem~\ref{thm:R-set-rich}, there exists a $\nu_{\apprx, t+1} \in \R_{\nu,t}^{\rho}(\mu^{N_1}_{t}, \nu^{N_2}_{t})$ for the Red team policy $\psi_t^{N_2}$ such that
    \begin{equation}
        \label{eqn:red-apprx-nu}
        \expct{\dtv{\N^{N_2}_{t+1}, \nu_{\apprx,t+1}}\big \vert \bfX_t^{N_1} =\bfx^{N_1}_t, \bfY_t^{N_2}=\bfy^{N_2}_t }{\psi_t^{N_2}} \leq \frac{|\Y|}{2}\sqrt{\frac{1}{N_2}}.
    \end{equation}

    For notational simplicity, we drop the conditions $\bfX_t^{N_1} =\bfx^{N_1}_t$ and $\bfY_t^{N_2}=\bfy^{N_2}_t$ in the following derivations. 
    Then, for all joint states $\bfx^{N_1}_t \in \X^{N_1}$ and $\bfy^{N_2}_t \in \Y^{N_2}$, we have
    \begin{flalign}
        & J_{t}^{N, \alpha^*, \psi^{N_2}}(\bfx^{N_1}_t, \bfy^{N_2}_t)
        =r^\rho_t(\mu^{N_1}_t, \nu^{N_2}_t) + \mathbb{E}_{\alpha^*, \psi^{N_2}}\Big[J_{t+1}^{N, \alpha^*, \psi^{N_2}}(\bfX^{N_1}_{t+1}, \bfY^{N_2}_{t+1}) \Big]\\
        &\stackrel{\text{(i)}}{\geq}  r^\rho_t(\mu^{N_1}_t, \nu^{N_2}_t) + \mathbb{E}_{\alpha^*, \psi^{N_2}}\Big[\lowervalue_{\cor, t+1}^{\rho*}(\M^{N_1}_{t+1}, \N^{N_2}_{t+1}) \Big] -\residue
        \label{eqn:main-2}\\
        &= r^\rho_t(\mu^{N_1}_t, \nu^{N_2}_t) -\residue \\
        & \qquad + \mathbb{E}_{\alpha^*, \psi^{N_2}}\Big[\lowervalue_{\cor, t+1}^{\rho*}(\M^{N_1}_{t+1}, \N^{N_2}_{t+1}) 
        - \lowervalue_{\cor, t+1}^{\rho*}(\mu_{t+1}^*, \nu_{\apprx,t+1})  + \lowervalue_{\cor, t+1}^{\rho*}(\mu_{t+1}^*, \nu_{\apprx,t+1}) \Big] \nonumber\\
        &\stackrel{\text{(ii)}}{\geq}  r^\rho_t(\mu^{N_1}_t, \nu^{N_2}_t) +
        \lowervalue_{\cor, t+1}^{\rho*}(\mu_{t+1}^*, \nu_{\apprx,t+1}) -\residue\label{eqn:main-3} \\
        &\qquad \qquad \qquad \qquad \qquad 
        - L_{J, t+1}  \underbrace{\mathbb{E}_{\alpha^*}\Big[\dtv{\M_{t+1}^{N_1},{\mu}^*_{t+1}}\Big]}_{=\mathcal{O}(\frac{1}{\sqrt{N_1}}) \text{ due to Lemma~\ref{lmm:mf-aprx}}}  
        - L_{J, t+1} \underbrace{\mathbb{E}_{ \psi^{N_2}}\Big[\dtv{ \N^{N_2}_{t+1},\nu_{\apprx,t+1}}\Big]} _{=\mathcal{O}(\frac{1}{\sqrt{N_2}}) \text{ due to \eqref{eqn:red-apprx-nu}}}
        \nonumber \\
        & \stackrel{\text{(iii)}}{=}  r^\rho_t(\mu^{N_1}_t, \nu^{N_2}_t) +
        \lowervalue_{\cor, t+1}^{\rho*}(\mu_{t+1}^*, \nu_{\apprx,t+1}) -\residue \label{eqn:main-5}\\
        &\stackrel{\text{(iv)}}{\geq}  r^\rho_t(\mu^{N_1}_t, \nu^{N_2}_t) + \min_{\nu_{t+1}\in \mathcal{R}_{\nu,t}^\rho(\mu^{N_1}_t, \nu^{N_2}_t)}
        \lowervalue_{\cor,t+1}^{\rho*}(\mu^*_{t+1}, \nu_{t+1}) -\residue\label{eqn:main-6}\\
        &\stackrel{\text{(v)}}{=}  r^\rho_t(\mu^{N_1}_t, \nu^{N_2}_t) +
        \max_{\mu_{t+1} \in \mathcal{R}_{\mu,t}^\rho(\mu^{N_1}_t, \nu^{N_2}_t)}~ \min_{\nu_{t+1}\in \mathcal{R}_{\nu,t}^\rho(\mu^{N_1}_t, \nu^{N_2}_t)}
        \lowervalue_{\cor,t+1}^{\rho*}(\mu_{t+1}, \nu_{t+1}) -\residue \label{eqn:main-7}\\
        &= \lowervalue_{\cor,t}^{\rho*}(\mu^{N_1}_t, \nu^{N_2}_t) -\residue. \label{eqn:main-8}
    \end{flalign}
    For inequality (i), we used the inductive hypothesis; 
    for inequality (ii), we utilized the Lipschitz continuity of the coordinator value function in Theorem~\ref{thm:value-L-cont}; 
    equality (iii) is due to  Lemma~\ref{lmm:mf-aprx} and Theorem~\ref{thm:R-set-rich}, and the fact that the two error terms with Lipschitz constant $L_{J, t+1}$ are bounded by 
    $\residue$;
    inequality (iv) is due to the fact that $\nu_{\apprx,t+1}$ is in the reachable set;
    equality (v) comes from the definition of $\mu^*_{t+1}$ in~\eqref{eqn:mu-star-def},
    and the final equality in~\eqref{eqn:main-8} is simply the definition of $\lowervalue^{\rho*}_{\cor,t}$, which completes the induction.

    Since the Red team strategy $\psi^{N_2} \in \Psi^{N_2}$ is arbitrary, we have that, for all joint states $\bfx^{N_1} \in \X^{N_1}$ and $\bfy^{N_2} \in \Y^{N_2}$,
    \begin{align*}
         \min_{\psi^{N_2}\in \psi^{N_2}} J^{N, \alpha^*, \psi^{N_2}}&(\bfx^{N_1} , \bfy^{N_2}) = \min_{\psi^{N_2}\in \psi^{N_2}} J^{N, \alpha^*, \psi^{N_2}}_{0}(\bfx^{N_1}, \bfy^{N_2}) \\
         &\geq \lowervalue_{\cor,0}^{\rho*}(\mu, \nu) - \residue
         =\lowervalue_{\cor}^{\rho*}(\mu^{N_1}, \nu^{N_2}) - \residue.
    \end{align*}
    
\end{proof}

\begin{lemma}
    \label{lmm:cor-apprx-err}
    The following inequality holds for all joint states $\bfx^{N_1} \in \X^{N_1}$ and $\bfy^{N_2} \in \Y^{N_2}$, 
    \begin{equation}
        \label{eqn:cor-apprx-err-bound}
        \lowervalue_{\cor}^{\rho*}(\mu^{N_1} , \nu^{N_1}) \geq \lowervalue^{N*}(\bfx^{N_1} , \bfy^{N_2} )- \residue,
    \end{equation}
    where $\mu^{N_1}=\empMu{\bfx^{N_1}}$ and $\nu^{N_2} = \empNu{\bfy^{N_2}}$.
\end{lemma}

\begin{proof}
    We prove the lemma through an inductive argument. 
    
    \textit{Base case:}
    At the terminal timestep $T$, the two value functions are the same. Thus, we have, for all joint states $\bfx^{N_1}_T \in \bfX^{N_1}$ and $\bfy^{N_2}_T \in \bfY^{N_2}$, that
    \begin{equation*}
        \lowervalue_{\cor,T}^{\rho*}(\mu^{N_1}_T , \nu^{N_1}_T)
         = \underline{J}_{T}^{N*}(\bfx^{N_1}_T , \bfy^{N_2}_T) = r^\rho_t (\mu^{N_1}_T , \nu^{N_1}_T).
    \end{equation*}

    \textit{Inductive hypothesis:}
    Assume that, at time step $t+1$, the following holds for all $\bfx^{N_1}_{t+1} \in \bfX^{N_1}$ and $\bfy^{N_2}_{t+1} \in \bfY^{N_2}$,
    \begin{equation}
        \lowervalue_{\cor, t+1}^{\rho*}(\mu^{N_1}_{t+1} , \nu^{N_1}_{t+1}) \geq 
        \underline{J}_{t+1}^{N*} (\bfx^{N_1}_{t+1} , \bfy^{N_2}_{t+1})- \residue.
    \end{equation}

    \textit{Induction:}
    Consider arbitrary $\bfx^{N_1}_{t} \in \bfX^{N_1}$ and $\bfy^{N_2}_{t} \in \bfY^{N_2}$.
    For each Blue team policy $\phi_t^{N_1} \in \Phi_t^{N_1}$, Theorem~\ref{thm:R-set-rich} allows us to define $\mu_{\apprx, t+1}^{\phi^{N_1}_t} \in \R_{\mu}(\mu_t^{N_1}, \nu_t^{N_2})$ such that
    \begin{equation}
        \label{eqn:blue-apprx-mu}
        \expct{\dtv{\M^{N_1}_{t+1}, \mu^{\phi^{N_1}_t}_{\apprx,t+1}}\big \vert \bfX_t^{N_1} =\bfx^{N_1}_t, \bfY_t^{N_2}=\bfy^{N_2}_t }{\phi_t^{N_1}} \leq \frac{|\X|}{2}\sqrt{\frac{1}{N_1}}.
    \end{equation}

    For an identical Red team policy $\psi_t \in \Psi_t$, denote $\nu^{\psi_t}_{t+1} = \nu^{N_2}_t G^\rho_t(\mu^{N_1}_t, \nu^{N_2}_t, \psi_t)$. Then, we have
    \begin{align*}
        \underline{J}_{t}^{N*} &(\bfx^{N_1}_{t} , \bfy^{N_2}_{t}) 
        = \max_{\phi^{N_1}_t \in \Phi^{N_1}_t} \min_{\psi^{N_2}_t \in \Psi^{N_2}_t} r^\rho_{t} (\mu^{N_1}_t, \nu^{N_2}_t) + \mathbb{E}_{\phi^{N_1}_t,\psi_t^{N_2}} \Big[\underline{J}_{t+1}^{N*} (\bfX^{N_1}_{t+1} , \bfY^{N_2}_{t+1}) \Big]\\
        & \stackrel{\text{(i)}}{\leq}   r^\rho_{t} (\mu^{N_1}_t, \nu^{N_2}_t) + \max_{\phi^{N_1}_t \in \Phi^{N_1}_t} \min_{\psi_t^{N_2} \in \Psi^{N_2}_t}\mathbb{E}_{\phi^{N_1}_t,\psi_t^{N_2}} \Big[\lowervalue_{\cor, t+1}^{\rho*}(\M^{N_1}_{t+1}, \N^{N_2}_{t+1}) \Big] +\residue \\
        & \stackrel{\text{(ii)}}{\leq}  r^\rho_{t} (\mu^{N_1}_t, \nu^{N_2}_t) + \max_{\phi^{N_1}_t \in \Phi^{N_1}_t} \; \min_{\psi_t \in \Psi_t} \; \mathbb{E}_{\phi^{N_1}_t,\psi_t} \Big[\lowervalue_{\cor, t+1}^{\rho*}(\M^{N_1}_{t+1}, \N^{N_2}_{t+1}) \Big] +\residue\\
        &=  r^\rho_{t} (\mu^{N_1}_t, \nu^{N_2}_t) + \residue 
        + \max_{\phi^{N_1}_t \in \Phi^{N_1}_t} \; \min_{\psi_t \in \Psi_t} \; \mathbb{E}_{\phi^{N_1}_t,\psi_t} \Big[\lowervalue_{\cor, t+1}^{\rho*}(\M^{N_1}_{t+1}, \N^{N_2}_{t+1}) \\
        & \qquad \qquad \qquad \qquad \qquad \qquad \qquad \qquad \qquad 
        - \lowervalue_{\cor, t+1}^{\rho*}(\mu_{\apprx, t+1}^{\phi^{N_1}_t}, \nu_{t+1}^{\psi_t}) +  \lowervalue_{\cor, t+1}^{\rho*}(\mu_{\apprx, t+1}^{\phi^{N_1}_t}, \nu_{t+1}^{\psi_t})\Big]\\
        &\stackrel{\text{(iii)}}{\leq} r^\rho_{t} (\mu^{N_1}_t, \nu^{N_2}_t) +\max_{\phi^{N_1}_t \in \Phi^{N_1}_t} \; \min_{\psi_t \in \Psi_t} \lowervalue_{\cor, t+1}^{\rho*}(\mu_{\apprx, t+1}^{\phi^{N_1}_t}, \nu_{t+1}^{\psi_t})         +\residue\\
        & \qquad \qquad \qquad \qquad \qquad \qquad\qquad \qquad
        + L_{J, t+1}  \mathbb{E}_{\phi^{N_1}_t,\psi_t}\Big[\dtv{\M_{t+1}^{N_1}, \mu_{\apprx, t+1}^{\phi^{N_1}_t}} + \dtv{\N_{t+1}^{N_2}, \nu_{t+1}^{\psi_t}}\Big] \\
        &\stackrel{\text{(iv)}}{\leq} r^\rho_{t} (\mu^{N_1}_t, \nu^{N_2}_t) + \max_{\phi^{N_1}_t \in \Phi^{N_1}_t} \; \min_{\psi_t \in \Psi_t} \lowervalue_{\cor, t+1}^{\rho*}(\mu_{\apprx, t+1}^{\phi^{N_1}_t}, \nu_{t+1}^{\psi_t}) +\residue \\
        &\stackrel{\text{(v)}}{=} r^\rho_{t} (\mu^{N_1}_t, \nu^{N_2}_t) + \max_{\phi^{N_1}_t \in \Phi^{N_1}_t} \min_{\nu_{t+1} \in \R_{\nu,t}^\rho (\mu^{N_1}_t, \nu^{N_2}_t)} \lowervalue_{\cor, t+1}^{\rho*}(\mu_{\apprx, t+1}^{\phi^{N_1}_t}, \nu_{t+1})+ \residue
        \\
        &\stackrel{\text{(vi)}}{\leq} r^\rho_{t} (\mu _{t}, \nu_{t}) + \max_{\mu_{t+1} \in \R_{\mu,t}^\rho (\mu^{N_1}_t, \nu^{N_2}_t)} \min_{\nu_{t+1} \in \R_{\nu,t}^\rho (\mu^{N_1}_t, \nu^{N_2}_t)} \lowervalue_{\cor, t+1}^{\rho*}(\mu_{t+1}, \nu_{t+1}) + \residue
        \\
        &= \lowervalue_{\cor, t}^{\rho*}(\mu_{t}, \nu_{t}) +\residue.
    \end{align*}
    For inequality (i), we used the inductive hypothesis; 
    for inequality (ii), we reduced the optimization domain of the Red team to the identical policy space;
    inequality (iii) is a result of the Lipschitz continuity;
    for inequality (iv), we use Lemma~\ref{lmm:mf-aprx} and Theorem~\ref{thm:R-set-rich};
    equality (v) holds, since for all $\mu_{t+1}$ in the reachable set, there is an identical Red team strategy that induces it, and vice versa.  Consequently, switching from optimizing with the identical Red team policies to the reachable set does not change the minimization domain;
    inequality(vi) is due to the fact that $\mu_{\apprx, t+1}^{\phi^{N_1}_t}$ is always in the reachable set by construction.
\end{proof}

\section{Numerical Example}
\label{sec:example}
In this section, we provide two numerical examples. 
The first one illustrates a scenario where the lower and upper coordinator game values are different. 
The second one verifies the performance guarantees in the finite-population game.
For both examples, the states spaces for the Blue and Red team are $\X=\{x^1,x^2\}$ and $\Y=\{y^1,y^2\}$, and the action spaces are $\U = \{u^1, u^2\}$ and $\V = \{v^1, v^2\}$. 
The two-state state spaces allow the \emph{joint} mean-fields $(\mu_t, \nu_t)$ to be fully characterized solely by $\mu_t(x^1)$ and $ \nu_t(y^1)$.

The code that implements the two examples can be found at \hyperlink{https://github.com/scottyueguan/MFTG}{https://github.com/scottyueguan/MFTG}.

\subsection{Numerical Example 1}
\vspace{-0.05in}

We consider a following two-stage example with a population ratio $\rho=0.6$. The \emph{time-invariant} transition kernel for the Blue agents is defined as follows:
\begin{small}
\begin{alignat*}{2}
&f^\rho(x^1|x^1, u^1, \mu, \nu) = 0.5 \big(1 + \big(\rho \mu(x^1) - (1-\rho) \nu(y^1)\big)\big), ~~
&& f^\rho(x^2|x^1, u^1, \mu, \nu) = 0.5 \big(1 - \big(\rho \mu(x^1) - (1-\rho) \nu(y^1)\big)\big),\\
&f^\rho(x^1|x^1, u^2, \mu, \nu) = 0.5 \big(1 - 0.3\big(\rho \mu(x^1) - (1-\rho) \nu(y^1)\big)\big), ~~
&& f^\rho(x^2|x^1, u^2, \mu, \nu) = 0.5 \big(1 + 0.3\big(\rho \mu(x^1) - (1-\rho) \nu(y^1)\big)\big),\\
&f^\rho(x^1|x^2, u^1, \mu, \nu) = 0.5 \big(1 - \big(\rho \mu(x^2) - (1-\rho) \nu(y^2)\big)\big), ~~
&&f^\rho(x^2|x^2, u^1, \mu, \nu) = 0.5 \big(1 + \big(\rho \mu(x^2) - (1-\rho) \nu(y^2)\big)\big), \\
&f^\rho(x^1|x^2, u^2, \mu, \nu) = 0.5 \big(1 + 0.3\big(\rho \mu(x^2) - (1-\rho) \nu(y^2)\big)\big), ~~
&&f^\rho(x^2|x^2, u^2, \mu, \nu) = 0.5 \big(1 -0.3\big(\rho \mu(x^2) - (1-\rho) \nu(y^2)\big)\big).
\end{alignat*}
\end{small}
The time-invariant Red transition kernel is similarly defined as
\begin{small}
\begin{alignat*}{2}
&g^\rho(y^1|y^1, v^1, \mu, \nu) = 0.5 \big(1 + \big((1-\rho) \nu(y^1)-\rho \mu(x^1) \big)\big), ~~
&& g^\rho(y^2|y^1, v^1, \mu, \nu) = 0.5 \big(1 - \big((1-\rho) \nu(y^1)-\rho \mu(x^1) \big)\big),\\
&g^\rho(y^1|y^1, v^2, \mu, \nu) = 0.5 \big(1 - 0.3\big((1-\rho) \nu(y^1)-\rho \mu(x^1) \big)\big), ~~
&& g^\rho(y^2|y^1, v^2, \mu, \nu) = 0.5 \big(1 + 0.3\big((1-\rho) \nu(y^1)-\rho \mu(x^1) \big)\big),\\
&g^\rho(y^1|y^2, v^1, \mu, \nu) = 0.5 \big(1 - \big((1-\rho) \nu(y^2)-\rho \mu(x^2) \big)\big), ~~
&& g^\rho(y^2|y^2, v^1, \mu, \nu) = 0.5 \big(1 + \big((1-\rho) \nu(y^2)-\rho \mu(x^2) \big)\big),\\
&g^\rho(y^1|y^2, v^2, \mu, \nu) = 0.5 \big(1 + 0.3\big((1-\rho) \nu(y^2)-\rho \mu(x^2) \big)\big), ~~
&& g^\rho(y^2|y^2, v^2, \mu, \nu) = 0.5 \big(1 -0.3\big((1-\rho) \nu(y^2)-\rho \mu(x^2) \big)\big).
\end{alignat*}
\end{small}

\vspace{-0.2in}
For simplicity, we only consider non-zero reward at the terminal time step $T=2$. Specifically,
\begin{align*}
    &r^\rho_0(\mu, \nu)= r^\rho_1(\mu, \nu) = 0 \quad \forall \mu \in \P(\X), \nu \in \P(\Y), \\
    &r^\rho_2(\mu, \nu) = \mu(x^2).
\end{align*}
In other words, the objective for the Blue team is to maximize its population fraction on state $x^2$ at $t=2$. While the Red distribution does not play a role in the reward structure, the Red agents need to strategically distribute themselves to influence the Blue agent's transitions in order to minimize the terminal reward. 

\begin{figure}[b!]
    \vspace{-0.1in}
    \centering
    \includegraphics[width=\linewidth]{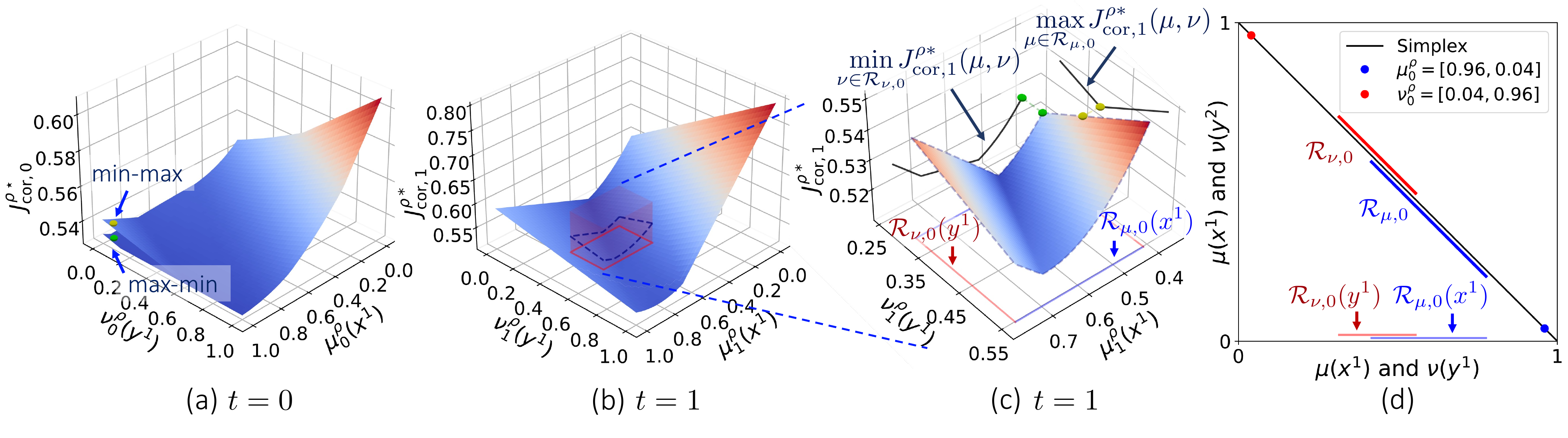}
    \vspace{-0.3in}
    \caption{Subplots (a)-(c) present the game value computed via discretization. 
    The x- and y-axes correspond to $\mu^\rho_t(x^1)$ and $\nu^\rho_t(y^1)$, respectively.
    Subplot (d) illustrates the reachable sets starting from $\mu_0=[0.96, 0.04]$ and $\nu_0 = [0.04, 0.96]$.
    }
    \label{fig:simple-example-values}
\end{figure}

The game values for the coordinator game are computed through discretization, where we uniformly mesh the two-dimensional simplices $\P(\X)$ and $\P(\Y)$ into 500 bins%
\footnote{One can easily provide an error bound on the difference between the discretized value and the true optimal value using the Lipschitz constants. See~\citep{guan:2023chasing} for an example.}.
The numerically solved values are presented in~\figref{fig:simple-example-values}. 

While the game value $J^{\rho*}_{\cor,1}$ exists at time $t=1$, it is not convex-concave, and hence the best-response correspondences in the coordinator game is not convex-valued. 
Consequently, the upper (max-min) and lower (min-max) game values at the previous step $t\hspace{-0.02in}=\hspace{-0.02in}0$ differs, as observed in subplot (a).
Specifically, at $\mu^\rho_0=[0.96, 0.04]$ and $\nu^\rho_0=[0.04, 0.96]$, we have the lower value $\underline{J}^{\rho*}_{\cor,0} = 0.5298$ and the upper value $\bar{J}^{\rho*}_{\cor,0} = 0.5384$, which are visualized as the green and yellow points.
This discrepancy in the game values implies the absence of a Nash equilibrium in this coordinator game. 

Based on~\eqref{eqn:lower-j-ab-t-rset}, the optimization domains for computing $J^{\rho*}_{\cor,0}$ are $\mathcal{R}_{\mu,0}(\mu_0^\rho,\nu_0^\rho)$ for the maximization and $\mathcal{R}_{\nu,0}(\mu_0^\rho,\nu_0^\rho)$ for the minimization, both of which are plotted in (d) and visualized as the box in subplots (b) and (c). 
Subplot (c) presents a zoom-in for the optimization $\max_{\mathcal{R}_{\mu,0}} \min_{\mathcal{R}_{\nu,0}} J^{\rho*}_{\cor,1}$ and its min-max counterpart.
The marginal functions are also plotted, from which the max-min (green point) and min-max values (yellow point) at $t=0$ can be directly obtained.

\begin{figure}[t]
    \centering
    \vspace{-0.2in}
    \includegraphics[width=0.9\linewidth]{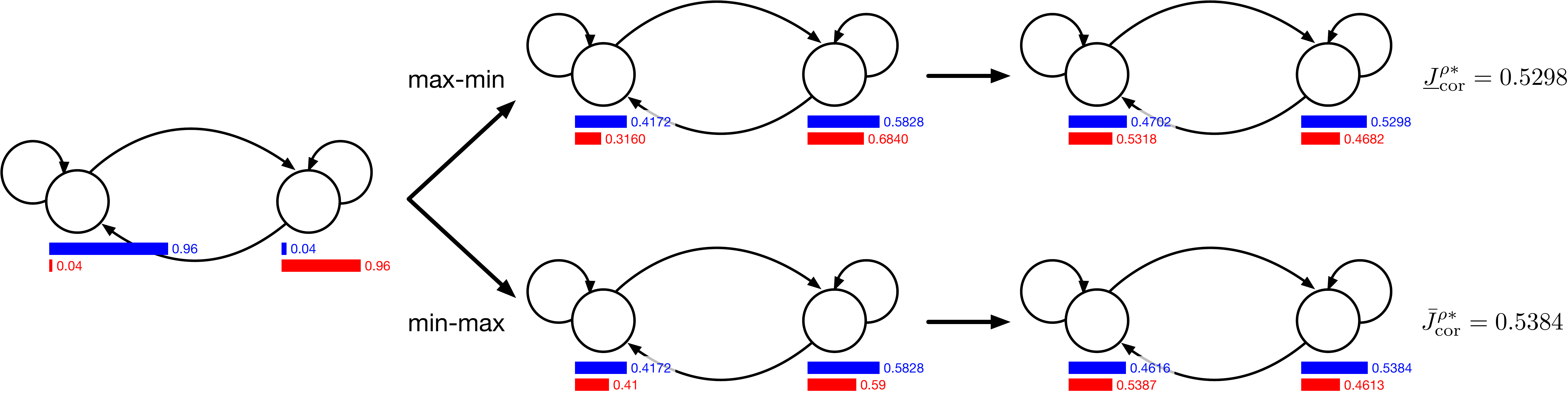}
    \caption{Max-min and min-max trajectory of the coordinator game.}
    \label{fig:example-trajectories}
    \vspace{-0.15in}
\end{figure}

Figure~\ref{fig:example-trajectories} presents the max-min and min-max trajectory of the coordinator game. 
Note that the max-min solution corresponds to the highest worst case performance for the Blue coordinator, where the Blue coordinator announces its first move from $[0.96, 0.04]$ to $[0.4172, 0.5828]$ and then the Red coordinator exploits the announced move. 
If instead, the Red coordinator announces its first move toward $[0.3160, 0.6840]$ in the max-min trajectory, the Blue coordinator can exploit that move by achieving $[0.776, 0.224]$ at $t=1$ and ultimately receive a higher reward of 0.5442.

\subsection{Numerical Example 2}

It is generally challenging to verify the suboptimality bound in Theorem~\ref{thm:performance-guarantees}, since computing the true optimal performance of a finite-population team game is intractable. 
However, for the following simple example, we can construct the optimal team strategies even for large but still finite-population teams.

Consider a ZS-MFTG instance with a terminal time $T=2$ and population ratio $\rho=0.375$.
The (minimizing) Red team's objective is to maximize its presence at state $y^1$ at $t=2$, which translates to 
\begin{align*}
    &r^\rho_0(\mu, \nu)= r^\rho_1(\mu, \nu) = 0 \quad \forall \mu \in \P(\X), \nu \in \P(\Y), \\
    &r^\rho_2(\mu, \nu) = - \nu(y^1)
\end{align*}

The Blue transition kernels are time-invariant, deterministic and independent of the \MF{}s, and are defined as 
\begin{equation}
    \label{eqn:example2-blue-dynamics} 
    \begin{alignedat}{3}
        &  f_t^\rho(x^1|x^1, u^1, \mu, \nu) = 1.0, \quad
        && f_t^\rho(x^2|x^1, u^1, \mu, \nu) = 0.0, \qquad
        && \forall \; \mu \in \P(\X), \nu \in \P(\Y), t\in \{0,1\}, \\
        &  f_t^\rho(x^1|x^1, u^2, \mu, \nu) = 0.0, \quad
        && f_t^\rho(x^2|x^1, u^2, \mu, \nu) = 1.0, \qquad
        && \forall \; \mu \in \P(\X), \nu \in \P(\Y),t\in \{0,1\}, \\
        &  f_t^\rho(x^1|x^2, u^1, \mu, \nu) = 0.0, \quad
        && f_t^\rho(x^2|x^2, u^1, \mu, \nu) = 1.0, \qquad
        && \forall \; \mu \in \P(\X), \nu \in \P(\Y),t\in \{0,1\}, \\
        &  f_t^\rho(x^1|x^2, u^2, \mu, \nu) = 1.0, \quad
        && f_t^\rho(x^2|x^2, u^2, \mu, \nu) = 0.0, \qquad
        && \forall \; \mu \in \P(\X), \nu \in \P(\Y), t\in \{0,1\}.
    \end{alignedat}
\end{equation}
Under the above transition kernel, a Blue agent can freely move tween the two nodes.
Specifically, action $u^1$ instructs a Blue agent to remain at its current state, while action $u^2$ directs it to transition to the other state.
One can verify that the above Blue transition kernel ensures that the Blue reachable set covers the whole simplex regardless of the initial team distributions, i.e., $\mathcal{R}_{\mu, t}(\mu, \nu) = \P(\X)$ for all $\mu \in \P(\X)$, $\nu \in \P(\Y)$ and $t\in \{0,1\}$.

The Red transition kernels are time dependent and are defined as 
\begin{alignat}{3}
    &   g_0^\rho(y^1|y^1, v, \mu, \nu) = 1.0, \quad
    &&  g_0^\rho(y^2|y^1, v, \mu, \nu) = 0.0,  ~~ 
    && \forall \; v \in \V, \mu \in \P(\X), \nu \in \P(\Y),
    \label{eqn:example2-red-dynamics-t0-1} \\
    &   g_0^\rho(y^1|y^2, v, \mu, \nu) = 0.0, \quad
    &&  g_0^\rho(y^2|y^2, v, \mu, \nu) = 1.0, ~~ 
    && \forall \; v \in \V, \mu \in \P(\X), \nu \in \P(\Y), 
    \label{eqn:example2-red-dynamics-t0-2}\\
    &   g_1^\rho(y^1|y^1, v, \mu, \nu) = 1.0, \quad
    &&  g_1^\rho(y^2|y^1, v, \mu, \nu) = 0.0, ~~ 
    && \forall \; v \in \V, \mu \in \P(\X), \nu \in \P(\Y),
    \label{eqn:example2-red-dynamics-t1-2}\\
    &   g_1^\rho(y^1|y^2, v^1, \mu, \nu) = 0.0, \quad
    &&  g_1^\rho(y^2|y^2, v^1, \mu, \nu) = 1.0, ~~  
    && \qquad \quad \forall \; \mu \in \P(\X), \nu \in \P(\Y),
    \label{eqn:example2-red-dynamics-t1-1}
\end{alignat}
and the transitions from state $y^2$ using action $v^2$ at $t=1$ are given by
\begin{equation}
    \label{eqn:example2-red-dynamics-t1-3}
    \begin{alignedat}{2}
        & g_1^\rho(y^2|y^2, v^2, \mu, \nu) = 1-\min\Big\{ 5 \Big((\mu(x^1)-\frac{1}{\sqrt{2}})^2 + (\mu(x^2)-(1-\frac{1}{\sqrt{2}}))^2\Big), 1\Big\}, ~~~~
        &&\forall \nu \in \P(\Y),\\
        & g_1^\rho(y^2|y^2, v^2, \mu, \nu) =\min\Big\{ 5 \Big((\mu(x^1)-\frac{1}{\sqrt{2}})^2 + (\mu(x^2)-(1-\frac{1}{\sqrt{2}}))^2\Big), 1\Big\}, ~~~~
        &&\forall \nu \in \P(\Y).
    \end{alignedat}
\end{equation}

The Red transitions in \eqref{eqn:example2-red-dynamics-t0-1} and~\eqref{eqn:example2-red-dynamics-t0-2} implies that all Red agents are frozen and can not change their states with any action at $t=0$.
Similarly, all Red agents on state $y^1$ are frozen at $t=1$ as depicted in~\eqref{eqn:example2-red-dynamics-t1-2}.
The Red agents at state $y^2$ at $t=1$ can choose to either deterministically stay at state $y^2$ using action $v^1$ as in~\eqref{eqn:example2-red-dynamics-t1-1} or try to move toward state $y^1$ using action $v^2$. 
As in~\eqref{eqn:example2-red-dynamics-t1-3}, the probability of transitioning from $y^2$ to $y^1$ is controlled by the Blue team's distribution $\mu$ at $t=1$. 
If the Blue team can perfectly match the target distribution $[1/\sqrt{2}, 1- 1/\sqrt{2}]$, then no Red agent can transition from $y^2$ to $y^1$. 
Otherwise, the more the Blue team deviates from the target distribution, the more likely a Red agent can transition from $y^2$ to $y^1$.

\vspace{-0.1in}

\paragraph{Infinite-population case.}
Since all Red agents are frozen at $t=0$ and the Red sub-population at state 
$y^2$ is again frozen at $t=1$, only the actions of Red agents at state $y^1$ at $t=1$ have an impact on the game outcome. 
As a result, the above setup leads to a dominant optimal Red team strategy:
all Red agents at $y^2$ use action~$v^2$ at $t=1$ and try to transition to state $y^1$ for the sake of maximizing the team's presence at $y^1$, regardless of the Blue team's distribution.
Note that this is the \emph{dominant} optimal strategy against all Blue team strategies in both finite and infinite-population scenarios. 
Specifically, 
\begin{equation}
    \psi^{N_2 *}_{j, 1}(v^2|y^2, \mu, \nu) = \psi^{\rho *}_1(v^2|y^2, \mu, \nu) = 1.0, ~~
    \forall \mu \in \P(\X), \nu \in \P(\Y), j \in [N_2].
\end{equation}

On the other hand, the Blue team should try to match the target distribution $[1/\sqrt{2}, 1-1/ \sqrt{2}]$ to minimize the portion of Red agents transitioning from $y^2$ to $y^1$ at $t=1$.
Since the Blue reachable set at $t=0$ covers the whole simplex, the target distribution can always be achieved at $t=1$ starting from any initial Blue team distribution in the \emph{infinite-population} case. 
One Blue coordination strategy that achieves the given target distribution is
\begin{equation}
\label{eqn:example2-coord-strate}
    \alpha_0 (\mu, \nu) = \left \{
    \begin{array}{cc}
        \pi_0^1, \qquad & \text{if } \mu_0(x^1) < \frac{1}{\sqrt{2}}, \\
        \pi_0^2, \qquad & \text{if } \mu_0(x^1) \geq \frac{1}{\sqrt{2}},
    \end{array}
    \right .
\end{equation}
where,
\begin{alignat*}{2}
    &\pi_0^1(u^1|x^1) =1,                               \qquad 
    && \pi_0^1(u^2|x^1) =0, \\
    &\pi_0^1(u^1|x^2) =\frac{1-1/\sqrt{2}}{\mu_0(x^2)}, \qquad 
    && \pi_0^1(u^2|x^2) =\frac{1/\sqrt{2}-\mu_0(x^1)}{\mu_0(x^2)},\\
    &\pi_0^2(u^1|x^1) =\frac{1/\sqrt{2}}{\mu_0(x^1)},   \qquad && \pi_0^2(u^2|x^1) =\frac{1-1/\sqrt{2}-\mu_0(x^2)}{\mu_0(x^1)}, \\
    &\pi_0^2(u^1|x^2) =1, \qquad && \pi_0^2(u^2|x^2) =0.
\end{alignat*}

The following Figure~\figref{fig:example-2-value} presents the coordinator game value, which is solved via discretization. 
In this example, the upper and lower game coincides, and thus the game value exists.
The black line in the middle subplot marks the game value with $\mu^\rho_1(x^1) = 1/\sqrt{2}$, when the Blue team perfectly achieves the target distribution. 
The values on the black line dominate and give the highest performance the Blue team can achieve given any $\nu_1^\rho(y^1)$, which aligns with our analysis which states that the Blue team should match the target distribution. 

\begin{figure}[b]
\vspace{-0.2in}
    \centering
    \includegraphics[width=0.85\linewidth]{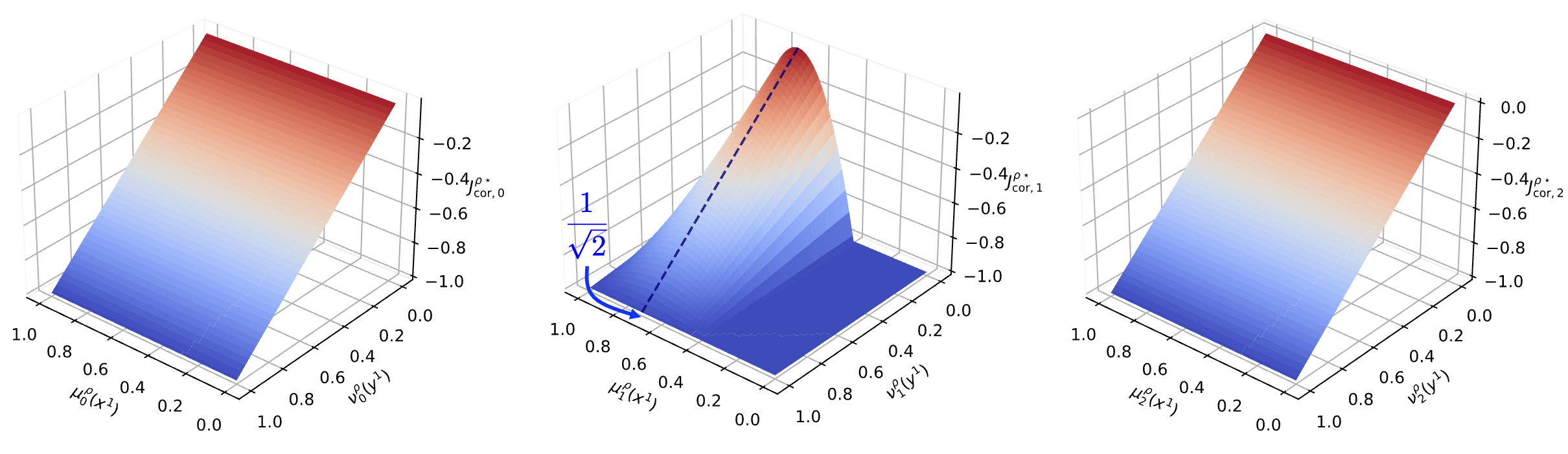}
    \vspace{-0.2in}
    \caption{The values of the coordinator game at time $t= 0,1$ and 2. 
    The blue line in the middle figure marks the value at $\mu^\rho_1(x^1) = 1/\sqrt{2}$, when the Blue team perfectly achieves the target distribution.
    }
    \label{fig:example-2-value}
\end{figure}

In the infinite-population coordinator game, the Blue team can always achieve the target distribution regardless of its initial distribution and thus completely block the migration of Red agents from $y^2$ to $y^1$. 
Consequently, only the Red agents that are initialized directly at $y^1$ can be counted toward the reward if the Blue team plays rationally, 
and thus the game value is simply given by $J^{\rho*}_0= - \nu^\rho_0(y^1)$, which matches the game value plotted in~\figref{fig:example-2-value}.

\paragraph{Finite-population case.}
The Red team's optimal strategy remains the same as the infinite-population case. 
However, the Blue team cannot achieve the \emph{irrational} target distribution with finite number of agents. 
While the Blue team can still match the target distribution \emph{in expectation} using a (stochastic) identical team strategy, the following analysis shows that a non-identical deterministic Blue team strategy achieves a better performance.

Consider a Blue team with three agents and all Blue agents are on node 1, i.e., $\mu^{3}_0 = [1,0]$.
The optimal Blue coordination strategy prescribes that all Blue agents pick $u^1$ (``stay") with probability $1/\sqrt{2}$ and $u^2$ (``move to $x^2$") with probability $(1\!-\!1/\sqrt{2})$ to reach the target distribution \emph{in expectation}.
Such action selection leads to the following four possible outcomes of the next Blue team \ED{} $\mu^3_1$:
$\mathbb{P}([1, 0]) = 0.354$,
$\mathbb{P}([2/3, 1/3]) = 0.439$,
$\mathbb{P}([1/3, 2/3]) = 0.182$,
and $\mathbb{P}([0, 1]) = 0.025$.
In expectation, these empirical distributions lead to a transition probability of 0.518 for a Red team agent moving from $y^2$ to $y^1$. 
Consequently, we have the worst-case performance of the optimal Blue coordinator strategy as 
$\min_{\psi^{N_2}} J^{3,\alpha^*,\psi^{N_2}} = -\nu_0(y^1) - 0.518\nu_0(y^2)$.

Next, consider the non-identical deterministic Blue team strategy, under which Blue agents 1 and 2 apply action $u^1$  and Blue agent 3 applies $u^2$.
This Blue team strategy deterministically leads to $\M_1^3=[2/3, 1/3]$ at $t=1$, and the resultant Red team transition probability from $y^2$ to $y^1$ is 0.016.
Clearly, the non-identical Blue team strategy significantly outperforms the identical mixed team strategy in this \emph{three-agent} case. 
Furthermore, this Blue team strategy is optimal over the entire non-identical Blue team strategy set, resulting in a finite-population optimal game value $J^{3*} = -\nu_0(y^1) - 0.016\nu_0(y^2)$. 

We repeat the above computation for multiple Blue team size $N_1$ and plot the suboptimality gap as the blue line in \figref{fig:deviation},
which verifies the $\mathcal{O}(1/\sqrt{\underline{N}})$ decrease rate predicted by Theorem~\ref{thm:performance-guarantees}.

Specifically, the optimal value for the finite-population game is computed based on the following:
\begin{itemize}
    \item At $t=0$, the Blue team uses a non-identical deterministic team strategy and deterministically achieves the optimal \ED{} $\mu^{N_1}_1$, which minimize the two-norm distance to the target distribution. 
    Such optimal Red \ED{} is constructed as $[\frac{n}{N_1}, \frac{N_1 - n}{N1}]$ where 
        $n = \argmin_{0 \leq n\leq N_1 } (\frac{n}{N_1} - \frac{1}{\sqrt{2}})^2$.
    \item 
    At $t=1$, all Red agents apply $v^2$, 
    and the distribution of the next Red \ED{}s $\nu^{N_2}_2$ can be computed based on the optimal $\mu^{N_1}_1$ achieved at $t=1$.
    \item 
    The expected game value is then computed based on the distribution of the Red \ED{} $\nu_2^{N_2}$.
\end{itemize}

\begin{figure}[b]
    \centering
    \includegraphics[width=0.55\linewidth]{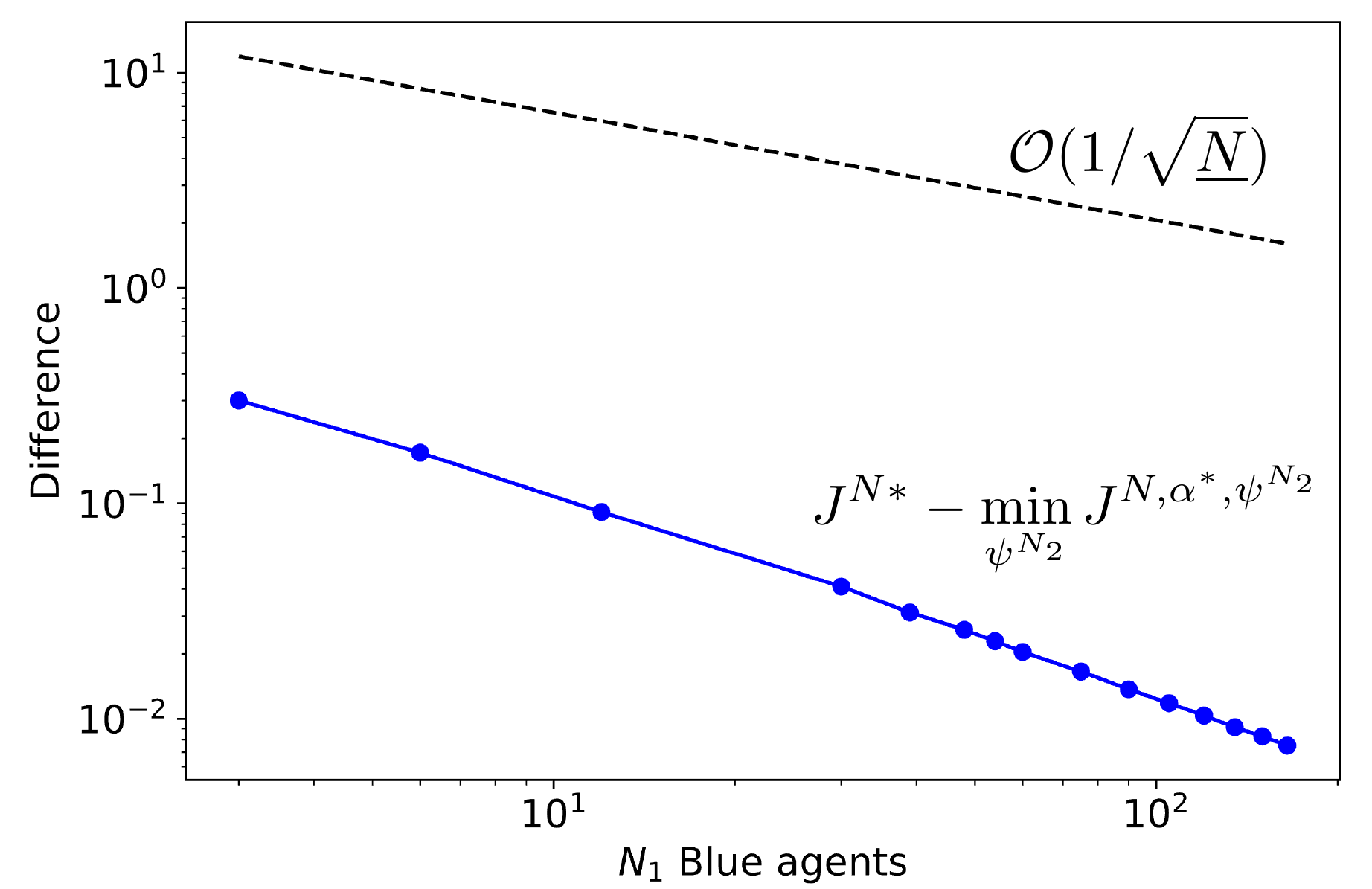}
    \caption{Difference between game values with initial distributions $\mu_0 = [1,0]$ and $\nu_0 = [0.6, 0.4]$.
    }
    \label{fig:deviation}
\end{figure}

\section{Conclusion}

In this work, we introduced a discrete zero-sum mean-field team game to model the behaviors of competing large-population teams. 
We developed a dynamic programming approach that approximately solves this large-population game at its infinite-population limit where only identical team strategies are considered. 
Our analysis demonstrated that the identical strategies constructed are $\epsilon$-optimal within the general class of non-identical team strategies when deployed in the original finite-population game.
The derived performance guarantees are verified through numerical examples.
Future work will investigate the LQG setup of this problem and explore machine-learning techniques to address more complex zero-sum mean-field team problems.
Additionally, we aim to generalize our results to the infinite-horizon discounted case and problems with heterogeneous sub-populations.

\section*{Acknowledgments:} 
The authors acknowledge fruitful discussions with Dr. Ali Pakniyat.

\clearpage
\begin{appendices}

\section{Special Cases and Examples}
\label{appdx:special-case}

\subsection{Pairwise State-Coupled Mean-Field Team Games}
In Section~\ref{sec:formulation}, we introduced the pairwise state-coupled dynamics and the averaged state-dependent reward as a special case for the proposed mean-field team game. 
Specifically, the pairwise state-coupled transition kernels are defined as follows
\begin{equation}
\tag{\ref{eqn:psc-MFTG-dynamics}}
\begin{aligned}
    f_t\big(x^{N_1}_{i,t+1}\big \vert x^{N_1}_{i,t}, u^{N_1}_{i,t}, \bfx^{N_1}_{-i,t}, \bfy^{N_2}_{t}) &= 
    \frac{ \sum_{k=1}^{N_1}f_{1,t}(x^{N_1}_{i,t+1}\big \vert x^{N_1}_{i,t}, u^{N_1}_{i,t}, x^{N_1}_{k,t} )}{N} +
    \frac{ \sum_{k=1}^{N_2}f_{2,t}(x^{N_1}_{i,t+1}\big \vert x^{N_1}_{i,t}, u^{N_1}_{i,t}, y^{N_2}_{k,t} )}{N}\\
    g_t\big(y^{N_2}_{j,t+1}\big \vert y^{N_2}_{j,t}, v^{N_2}_{j,t}, \bfx^{N_1}_{t}, \bfy^{N_2}_{-j, t}) &= 
    \frac{ \sum_{k=1}^{N_1}g_{1,t}(y^{N_2}_{j,t+1}\big \vert y^{N_2}_{j,t}, v^{N_2}_{j,t}, x^{N_1}_{k,t} )}{N} +
    \frac{ \sum_{k=1}^{N_2}g_{2,t}(y^{N_2}_{j,t+1}\big \vert y^{N_2}_{j,t}, v^{N_2}_{j,t}, y^{N_2}_{k,t} )}{N}.
\end{aligned}
\end{equation}

Through some algebraic manipulations, the Blue transition kernel can be re-written in the weakly-coupled form in~\eqref{eqn:dynamics-blue}:
\begin{align*}
&f_t \big(x^{N_1}_{i,t+1}\big \vert x^{N_1}_{i,t}, u^{N_1}_{i,t}, \bfx^{N_1}_{-i,t}, \bfy^{N_2}_{t}) \\
&= \frac{N_1}{N}\frac{ \sum_{k=1}^{N_1}f_{1,t}(x^{N_1}_{i,t+1}\big \vert x^{N_1}_{i,t}, u^{N_1}_{i,t}, x^{N_1}_{k,t} )}{N_1} +
    \frac{N_2}{N}\frac{ \sum_{k=1}^{N_2}f_{2,t}(x^{N_1}_{i,t+1}\big \vert x^{N_1}_{i,t}, u^{N_1}_{i,t}, y^{N_2}_{k,t} )}{N_2} \nonumber\\  
    &=  
    \rho\frac{ \sum_{k=1}^{N_1}\sum_{x\in\X} \indicator{x}{x^{N_1}_{k,t}}f_{1,t}(x^{N_1}_{i,t+1}\big \vert x^{N_1}_{i,t}, u^{N_1}_{i,t}, x^{N_1}_{k,t} )}{N_1} +
    (1-\rho)\frac{ \sum_{k=1}^{N_2}\sum_{y\in\Y} \indicator{y}{y^{N_2}_{k,t}}f_{2,t}(x^{N_1}_{i,t+1}\big \vert x^{N_1}_{i,t}, u^{N_1}_{i,t}, y^{N_2}_{k,t} )}{N_2} \nonumber \\  
    &= 
    \rho\frac{\sum_{x\in\X}  \sum_{k=1}^{N_1} \indicator{x}{x^{N_1}_{k,t}}f_{1,t}(x^{N_1}_{i,t+1}\big \vert x^{N_1}_{i,t}, u^{N_1}_{i,t}, x )}{N_1} +
    (1-\rho)\frac{ \sum_{y\in\Y} \sum_{k=1}^{N_2} \indicator{y}{y^{N_2}_{k,t}}f_{2,t}(x^{N_1}_{i,t+1}\big \vert x^{N_1}_{i,t}, u^{N_1}_{i,t}, y )}{N_2} \nonumber \\
    & =
    \rho\sum_{x\in \X}\mu^{N_1}_t(x) f_{1,t}(x^{N_1}_{i,t+1}\big \vert x^{N_1}_{i,t}, u^{N_1}_{i,t}, x ) +
    (1-\rho)\sum_{y\in \Y} \nu_t^{N_2}(y)f_{2,t}(x^{N_1}_{i,t+1}\big \vert x^{N_1}_{i,t}, u^{N_1}_{i,t}, y ) \nonumber\\
    &\triangleq f^\rho_t\big(x^{N_1}_{i,t+1}\big \vert x^{N_1}_{i,t}, u^{N_1}_{i,t},  \mu^{N_1}_t, \nu^{N_2}_t).   \nonumber
\end{align*}
The Red transition kernel can be similarly transformed into a weakly-coupled form.

The averaged state-dependent reward is defined as 
\begin{equation}
    \tag{\ref{eqn:team-rewards-eg}}
    \begin{aligned}
        r_t(\bfx_t^{N_1},\bfy_t^{N_2}) &= \frac{1}{N} \sum_{k=1}^{N_1}r_{1,t}(x^{N_1}_{k,t} ) - 
        \frac{1}{N} \sum_{k=1}^{N_2}r_{2,t}(y^{N_2}_{k,t} ).
    \end{aligned}
\end{equation}
One can also transform the above reward structure to a weakly-coupled form as in~\eqref{eqn:reward}
\begin{align*}
    r_t(\bfx_t^{N_1},\bfy_t^{N_2})
    &= \frac{N_1}{N} \frac{\sum_{k=1}^{N_1}r_{1,t}(x^{N_1}_{k,t}) }{N_1} + \frac{N_2}{N} \frac{\sum_{k=1}^{N_2}r_{2,t}(y^{N_2}_{k,t} )}{N_2}\\
    &= \frac{N_1}{N} \frac{\sum_{k=1}^{N_1} 
    \sum_{x\in\X} \indicator{x}{x^{N_1}_{k,t}}r_{1,t}(x^{N_1}_{k,t}) }{N_1} + \frac{N_2}{N} \frac{\sum_{k=1}^{N_2} \sum_{y \in \Y} \indicator{y}{y^{N_2}_{k,t}}r_{2,t}(y^{N_2}_{k,t} )}{N_2} \\
     &= \rho \sum_{x\in \X} \mu^{N_1}_t(x) r_{1,t}(x) - (1-\rho) \sum_{y\in \Y} \nu^{N_2}_t(y) r_{2,t}(y) \\
     &\triangleq r^\rho_t(\mu_t^{N_1}, \nu_t^{N_2}).
\end{align*}

\begin{proposition}
    \label{appdx-prop:psc-MFTG-dynamics-lip}
    The transition kernels in~\eqref{eqn:psc-MFTG-dynamics} satisfy
    \begin{alignat*}{2}
        &\sum_{x' \in \X} \abs{f^\rho_t(x'|x,u,\mu, \nu) - f^\rho_t(x'|x,u,\mu', \nu')} \leq 2 \Big(\dtv{\mu, \mu'} + \dtv{\nu, \nu'}\Big) ~~~&&\forall x \in \X, u \in \U, \\
        &\sum_{y' \in \Y} \abs{g^\rho_t(y'|y,v,\mu, \nu) - g^\rho_t (y'|y,v,\mu', \nu')} \leq 2 \Big(\dtv{\mu, \mu'} + \dtv{\nu,\nu'}\Big)~~~&& \forall y \in \Y, v \in \V.
    \end{alignat*}
\end{proposition}

\begin{proof}
    Due to symmetry, we only show the result for the Blue transition kernel. 
    From the definition, we have, for all $x \in \X$ and $u\in \U$, that 
    \begin{align*}
        \sum_{x' \in \X} \big \vert f^\rho_t(&x'|x,u,\mu, \nu) - f^\rho_t(x'|x,u,\mu', \nu')\big \vert  = \sum_{x' \in \X} \Big \vert \rho \sum_{z \in \X} \mu(z) f_{1,t}(x'|x,u,z) + (1-\rho) \sum_{y \in \Y} \nu(y) f_{2,t}(x'|x,u,y)\\
        & \qquad \qquad \qquad \qquad\qquad \quad\qquad \quad\qquad ~~~ 
        -  \rho \sum_{z \in \X} \mu(z) f_{1,t}(x'|x,u,z) - (1-\rho) \sum_{y \in \Y} \nu(y) f_{2,t}(x'|x,u,y) \Big \vert \\
        &\leq \rho  \sum_{x'\in \X} \sum_{z \in \X} f_{1,t}(x'|x,u,z) \abs{\mu(z) - \mu'(z)} + 
        (1-\rho) \sum_{x'\in \X}  \sum_{y \in \Y} f_{2,t}(x'|x,u,y) \abs{\nu(y) - \nu'(y)}\\
        &=  \rho \sum_{z \in \X} \abs{\mu(z) - \mu'(z)} +  (1-\rho) \sum_{y \in \Y} \abs{\nu(y) - \nu'(y)} = 2\rho \; \dtv{\mu, \mu'} + 2(1-\rho) \;\dtv{\nu, \nu'} \\
        &\leq 2 \Big(\dtv{\mu, \mu'} + \dtv{\nu, \nu'}\Big) .
    \end{align*}
\end{proof}

\begin{proposition}
    \label{appdx-prop:psc-MFTG-reward-lip}
    For all $\mu, \mu' \in \P(\X)$, $\nu, \nu' \in \P(\Y)$ and $t \in \{0, \ldots, T\}$, the reward structure in~\eqref{eqn:team-rewards-eg} satisfies 
    \begin{equation*}
        \abs{r^\rho_t(\mu, \nu) - r^\rho_t(\mu', \nu')} \leq (2 \max \{r_{1,\max}, r_{2,\max}\}) \Big(\dtv{\mu,\mu'} + \dtv{\nu,\nu'}\Big),
    \end{equation*}
    where $r_{1,\max} = \max_{x, t} \abs{r_{1,t}(x)}$ and $r_{2,\max} = \max_{y,t} \abs{r_{2,t}(y)}$.
\end{proposition}

\begin{proof}
    Note that
    \begin{align*}
        \abs{r^\rho_t(\mu,\nu)-r^\rho_t(\mu',\nu')} &= \abs{\sum_{x} r_{1,t}(x) (\mu(x)- \mu'(x)) - \sum_{y} r_{2,t}(y) (\nu(y)- \nu'(y))}\\
        &\leq \abs{\sum_{x} r_{1,t}(x) (\mu(x)- \mu'(x))} + \abs{\sum_{y} r_{2,t}(y) (\nu(y)- \nu'(y))}\\
        &\leq \sum_{x} \abs{r_{1,t}(x)} \abs{\mu(x)- \mu'(x)} + \sum_{y} \abs{r_{2,t}(y)}\abs{\nu(y)- \nu'(y)}\\
        &\leq r_{1,\max} \sum_{x} \abs{\mu(x)- \mu'(x)} + r_{2,\max}\sum_{y} \abs{\nu(y)- \nu'(y)}\\
        &\leq \max\{r_{1,\max}, r_{2,\max}\} \big(2\dtv{\mu, \mu'} + 2\dtv{\nu, \nu'}\big).
    \end{align*}
\end{proof}

\subsection{Discontinuous Reward Example}    \label{appdx-sec:discontinuity}

Through this example, we demonstrate the necessity of Assumption~\ref{assmpt:lipschitiz-rewards}, 
i.e., the reward function needs to be continuous with respect to the mean-fields as in~\eqref{eqn:lipschitz-rewards}.
Consider a MFTG with terminal time $T=1$ over the following deterministic system. 

\begin{figure}[b]
    \centering
    \includegraphics[width=0.5\linewidth]{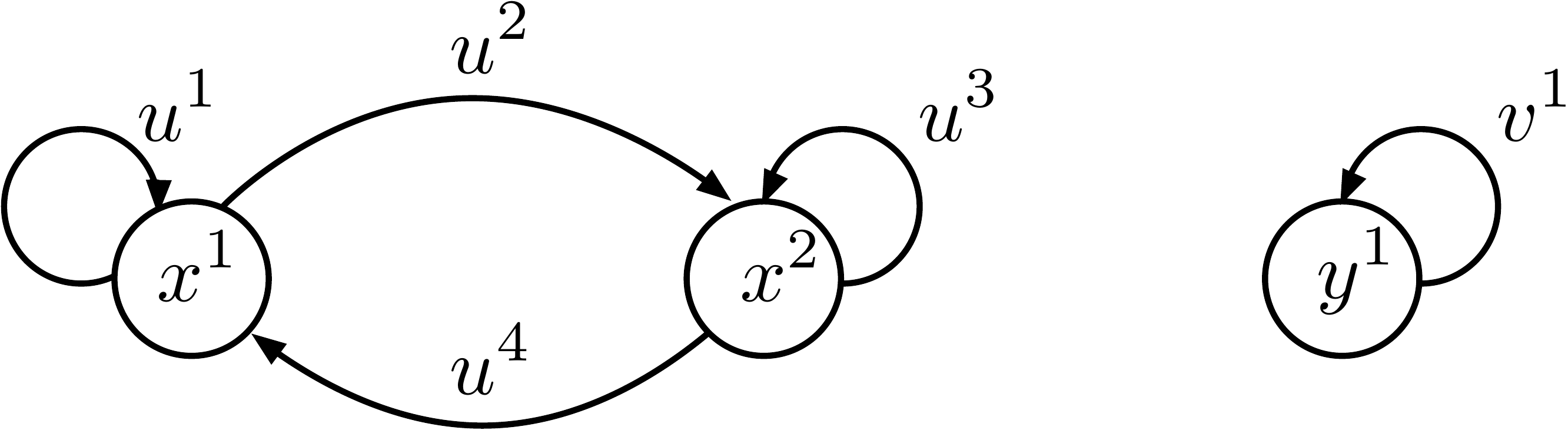}
    \caption{A deterministic system with two Blue states (left) and one Red state (right).}
    \label{fig:example-reward-cont}
\end{figure}

The Blue agent's transition kernel are given by 
\begin{alignat*}{2}
    f^\rho_0(x^1|x^1, u^1, \mu, \nu) = 1, \qquad f^\rho_0(x^2|x^1, u^1, \mu, \nu) = 0,\\
    f^\rho_0(x^1|x^1, u^2, \mu, \nu) = 0, \qquad f^\rho_0(x^2|x^1, u^1, \mu, \nu) = 1,\\
    f^\rho_0(x^1|x^2, u^3, \mu, \nu) = 0, \qquad f^\rho_0(x^2|x^2, u^3, \mu, \nu) = 1,\\
    f^\rho_0(x^1|x^2, u^4, \mu, \nu) = 1, \qquad f^\rho_0(x^2|x^2, u^4, \mu, \nu) = 0,\\
\end{alignat*}
which holds for all distributions $\mu \in \P(\X)$ and $\nu \in \P(\Y)$.

The Red agent's transition kernel is given by
\begin{alignat*}{2}
    g^\rho_0(y^1|y^1, v^1, \mu, \nu) = 1, \quad \forall \mu \in \P(\X), ~\nu \in \P(\Y).
\end{alignat*}

Given $\mu \in \P(\X)$ and $\nu \in \P(\Y)$, the \emph{discontinuous} reward is given by
\begin{align*}
    r_0^\rho(\mu, \nu) &= 0, \\
    r_1^\rho(\mu, \nu) &= \mathds{1}\Big(\mu = (1 / \sqrt{3}, 1-1 / \sqrt{3})\Big).
\end{align*}

As the dynamics are decoupled and the reward does not depend on the Red distribution, the game is essentially a single-team problem. We choose this degenerate example for its simplicity.
Since the reward at time $t=0$ is always zero, the objective of the Blue team is to achieve the desired distribution $(1 / \sqrt{3}, 1-1 / \sqrt{3})$ at time $t=1$ to receive one unit of reward.
Since $t=0$ is the only time step that the Blue agents select action, a Blue coordination strategy only consists of a Blue coordination policy at time $t=0$. 
In this case, an optimal Blue coordination strategy can be
\begin{equation*}
    \alpha^*_0 (\mu, \nu) = \left \{
    \begin{array}{cc}
        \pi_0^1, \qquad & \text{if } \mu_0(x^1) < \frac{1}{\sqrt{3}} \\
        \pi_0^2, \qquad & \text{if } \mu_0(x^1) \geq \frac{1}{\sqrt{3}}
    \end{array}
    \right .
    ,
\end{equation*}
where,
\begin{alignat*}{2}
    &\pi_0^1(u^1|x^1) =1,                               \qquad && \pi_0^1(u^2|x^1) =0, \\
    &\pi_0^1(u^3|x^2) =\frac{1-1/\sqrt{3}}{\mu_0(x^2)}, \qquad && \pi_0^1(u^4|x^2) =\frac{1/\sqrt{3}-\mu_0(x^1)}{\mu_0(x^2)},\\
    &\pi_0^2(u^1|x^1) =\frac{1/\sqrt{3}}{\mu_0(x^1)},   \qquad && \pi_0^2(u^2|x^1) =\frac{1-1/\sqrt{3}-\mu_0(x^2)}{\mu_0(x^1)}, \\
    &\pi_0^2(u^3|x^2) =1, \qquad && \pi_0^2(u^4|x^2) =0.
\end{alignat*}
The proposed local policy is well-defined, since for the case where $\mu_0(x^1) < 1/\sqrt{3}$, it is implied that $\mu_0(x^2) >0$, and similarly, for the case where $\mu_0(x^1) > 1/\sqrt{3}$, we automatically have $\mu_0(x^1) > 0$.

One can verify that the above coordination strategy will achieve the mean-field $\mu_1 = (\sqrt{3}, 1-1 / \sqrt{3})$ from all $\mu_0 \in \P(\X)$ under the mean-field dynamics~\eqref{eqn:mf-dynamics-local}.
Consequently, we have
\begin{equation*}
    J^{\rho*}_{\cor} (\mu_0, \nu_0) = 1, \quad \forall \mu_0\in \P(\X), \nu_0 \in \P(\Y).
\end{equation*}

However, we know that for all $N_1$, the \ED{} must be rational, and thus $\M^{N_1} \neq (\sqrt{3}, 1-1 / \sqrt{3})$ for all $N_1$ almost surely, which leads to 
\begin{equation*}
    J^{N*} (\bfx_0^{N_1}, \bfy^{N_2}_0) = 0, \quad \forall \bfx_0^{N_1} \in \X^{N_1}, \bfy^{N_2}_0 \in \Y^{N_2},
\end{equation*}
which implies that the performance of the coordination strategy achieved at the infinite-population limit fails to translate back to the finite-population game, and hence Theorem~\ref{thm:performance-guarantees} no longer holds.

One can construct similar examples to illustrate the necessity of the transition kernels being continuous with respect to the mean-fields.

\subsection{Counter-Example for Information State}
\label{appdx-sec:infomration-example}

For simplicity, we use a single-team example to illustrate why the \ED{} cannot be an information state when different agents are allowed to apply different strategies, even under the weakly-coupled dynamics~\eqref{eqn:dynamics-blue} and rewards~\eqref{eqn:reward}.
Consider the following system with state spaces $\X = \{x^1, x^2\}$ and $\Y = \varnothing$, and action spaces $\U = \{u^1, u^2\}$ and $\V = \varnothing$.

\begin{figure}
    \centering
    \includegraphics[width=0.3\linewidth]{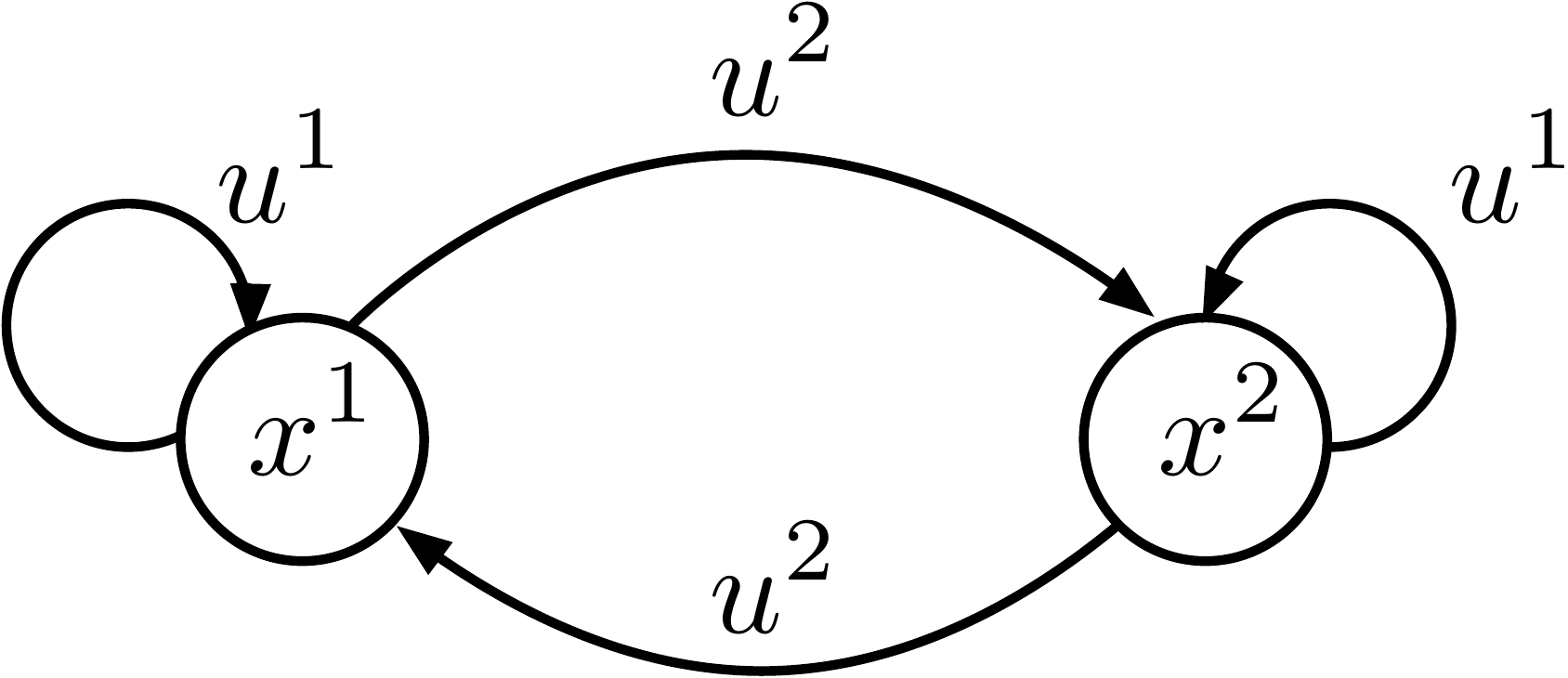}
    \caption{A counter-example showing \ED{} cannot serve as an information state.}
    \label{fig:enter-label}
\end{figure}

The transition kernel is time-invariant and is given by
\begin{alignat*}{2}
    & f_t(x^1|x^1, u^1, \mu) = 1.0, \qquad   && f_t(x^2|x^1, u^2, \mu) = 1.0, \qquad
    \raisebox{-.7\normalbaselineskip}[0pt][0pt]{$\forall \; \mu \in \P(\X), \; t \in \{0,1\}$.}
    \\
    & f_t(x^2|x^2, u^1, \mu) = 1.0, \qquad   && f_t(x^1|x^2, u^2, \mu) = 1.0, \\
\end{alignat*}

In words, each agent's transition is deterministic and decoupled from the other agents. 
At each state, the agent can select action $u^1$ to stay at its current state or use action $u^2$ to move to the other state.

We consider a one-stage scenario with the reward functions
\begin{align*}
    r_0(\mu) &= 0, \qquad \forall \mu \in \P(\X), \\
    r_1(\mu) &= \mu(x^1).
\end{align*}

To illustrate that the \ED{} cannot serve as an information state, consider a two-agent team and the following non-identical team strategy:
\begin{alignat*}{2}
    & \phi^2_{1,0}(u^2|x^1, \mu) = 1, \qquad &&\phi^2_{1,0}(u^1|x^2, \mu) = 1, ~~ \forall \mu \in \P(\X),\\
    & \phi^2_{2,0}(u^1|x^1, \mu) = 1, \qquad &&\phi^2_{2,0}(u^1|x^2, \mu) = 1, ~~ \forall \mu \in \P(\X).
\end{alignat*}
In words, agent 1 selects $u^2$ on state $x^1$ and selects $u^1$ on state $x^2$, while agent 2 selects $u^1$ on both states $x^1$ and $x^2$.

We now consider two initial configurations with the same \ED{}.
\begin{itemize}
    \item 
    $\X^2_{1,0} = x^1$ and $\X^2_{2,0} = x^2$, i.e., agent 1 is initially at~$x^1$ and agent~2 at~$x^2$.
    The \ED{} is thus $\M^2_{0}=[0.5, 0.5]$.
    If the two agents follow the above non-identical team strategy, then agent 1 uses action $u^2$ and transits to $x^2$, while agent 2 chooses action $u^1$ and stays on $x^2$, leading to $\M^2_1=[0, 1]$.
    The resultant reward value is then 0. 
    \item 
    $\X^2_{1,0} = x^2$ and $\X^2_{2,0} = x^1$, i.e., agent 1 is initially at~$x^2$ and agent~2 at~$x^1$.
    The \ED{} is again $\M^2_{0}=[0.5, 0.5]$.
    If the two agents follow the above non-identical team strategy, then both agents select action $u^1$ and stay at their current states, leading to $\M^2_1 = [0.5, 0.5]$ and a reward value of 0.5. 
\end{itemize}

From this example, we have shown that the values can be different under the same team strategy given different initial conditions that correspond to the same \ED{} since the \ED{} does not differentiate the ordering of the agents.
Clearly, the \ED{} alone is not enough to characterize the future evolution of the game, nor the value function. 
Consequently, the \ED{} is not an information state and we need the joint state information in the finite-population game to properly construct the value functions when different agents apply different strategies.

\section{Proof of the Mean-Field Approximation Results}

\subsection{Modified \texorpdfstring{$\ell^2$}{l2} Weak Law of Large Numbers}
The following lemma is a modified version of the $\ell^2$ weak law of large numbers~\citep{chung2001course} adapted to accommodate the notion of \ED{}s used in this work.
This lemma will be used extensively in the later proofs. 

\begin{lemma}
    \label{lmm:l2-apprx}
    Let $X_1, \ldots, X_N$ be $N$ independent random variables (vectors) taking values in a finite set $\X$.
    Suppose $X_i$ is distributed according to $p_i \in \mathcal{P}(\X)$, such that $\mathbb{P}(X_i = x) = p_i(x)$.
    Define the \ED{} as
    \begin{equation*}
        \M^N(x) = \frac{1}{N}\sum_{i=1}^N \indicator{x}{X_i}.
    \end{equation*}
    and let $\mu(x) = \frac{1}{N }\sum_{i=1}^N p_i(x)$.
    Then, the following inequality holds
        \begin{equation}
        \mathbb{E}\left[\dtv{\M^N, \mu}\right]  \leq \frac{1}{2} \sqrt{\frac{|\X|}{N}}.    
    \end{equation}
\end{lemma}

\begin{proof}
    Consider the random variables $X_{x, i} = \indicator{x}{X_i}$, which are independent across $i$ with mean $p_i(x)$.
    Since $X_{x,i}$ only takes a value of 0 and 1, we have $\mathbb{E}[X_{x,i}^2] = \mathbb{E}[X_{x,i}] = p_i(x)$, and $\mathrm{Var}(X_{x, i}) = p_i(x) - p_i(x)^2$.
    It follows that
    \begin{align*}
        \expectation{\lVert \M^N - \mu \rVert_2^2 } &= \expectation{\sum_{x\in \X} \abs{\M^N(x) - \mu(x)}^2}
        = \mathbb{E}\Big[\sum_{x \in \X}\Big\vert\frac{1}{N}\sum_{i=1}^N (X_{x,i}-p_i(x))\Big \vert ^2\Big]\\
        &= \frac{1}{N^2} \sum_{x\in \X} \mathrm{Var}\Big(\sum_{i=1}^N X_{x,i}\Big)
        \stackrel{(i)}{=} \frac{1}{N^2} \sum_{x \in \X}  \sum_{i=1}^N\mathrm{Var}(X_{x,i})\\
        &= \frac{1}{N^2} \sum_{x \in \X} \sum_{i=1}^N  (p_i(x) - p_i(x)^2)
        \leq \frac{1}{N^2} \sum_{i=1}^N \Big(\sum_{x \in \X} p_i(x)\Big) \\
        &\leq \frac{1}{N},
    \end{align*}
    where $(i)$ leverages the fact that $X_{x,i}$ are independent across $i$. 
    By Jensen's inequality, we have $\expectation{\lVert \M^N - \mu \rVert_2} \leq \frac{1}{\sqrt{N}}$.
    Furthermore, due to equivalence of norms in the finite dimensional space $\X$, we have $\norm{\M^N - \mu}_1 \leq \sqrt{|\X|} \norm{{\M^N - \mu}}_2$ almost surely, where $|\X|$ is the cardinality of the set $\X$. 
    It then follows that 
    \begin{align*}
         \mathbb{E}[\dtv{\M^N, \mu}] = \frac{1}{2} \expectation{\norm{\M^N - \mu}_1} \leq \frac{1}{2} \sqrt{\frac{|\X|}{N}}.
    \end{align*}

\end{proof}

\begin{corollary}
    \label{cor:l2-weak-law-iid}
    Let $X_1, \ldots, X_N$ be $N$ independent and identically distributed random variables (vectors) taking values in a finite set $\X$.
    Suppose $X_i$ is distributed according to $p \in \mathcal{P}(\X)$, such that $\mathbb{P}(\x_i = x) = p(x)$.
    Define the \ED{} as
    \begin{equation*}
        \M^N(x) = \frac{1}{N}\sum_{i=1}^N \indicator{x}{X_i}.
    \end{equation*}
    Then,
    \begin{equation}
        \mathbb{E}\left[\dtv{\M^N, p}\right] =\frac{1}{2} \sqrt{\frac{|\X|}{N}}.    
    \end{equation}
\end{corollary}

\subsection{Proof of Lemma~\ref{lmm:mf-aprx}}
\label{appdx:prop-mf-apprx}

The following lemma is used to support the proof of Lemma~\ref{lmm:mf-aprx}.
\begin{lemma}
    \label{appdx-lmm:sum}
    For all $a_1, \ldots, a_n \in \mathbb{R}_{\geq 0}$, the following inequality holds:
    \begin{equation*}
        \sum_{i=1}^n a_i \leq \sqrt{n} \left( \sqrt{\sum_{i=1}^n a_i^2} \right) .
    \end{equation*}
\end{lemma}
\begin{proof}
    Cauchy-Schwarz inequality directly implies
    \begin{equation*}
        \sum_{i=1}^n a_i = \sum_{i=1}^n |a_i| \cdot 1 \leq \sqrt{\sum_{i=1}^n a_i^2} \sqrt{\sum_{i=1}^n 1^2} = \sqrt{n} \sqrt{\sum_{i=1}^n a_i^2}.
    \end{equation*}
\end{proof}

Now, we restate Lemma~\ref{lmm:mf-aprx} and present its proof.
\mfp*

\begin{proof}
    Due to symmetry, we only prove the result for the Blue team.
    Let  $N_{1,t}^x = N_1 \M^{N_1}_t(x)$ 
    be the number of Blue agents that are at state $x$ at time $t$ in the finite-population system. 
    With a slight abuse of notation, we use $[N_{1, t}^x]$ to denote the set of the Blue agents currently at state $x$.
    We first focus on the sub-population of Blue agents that are at state $x$.
    Since all Blue agents in $[N_{1, t}^x]$ apply $\phi_t$ and randomize independently, 
    the next states for these $N_{1,t}^x$ agents are independent and identically distributed conditioned on the joint states $\bfX^{N_1}_t$ and $\bfY^{N_2}_t$.
    For an agent~$i$ in this sub-population $[N_{1,t}^x]$ the conditional distribution of its next state $X^{N_1}_{i,t+1}$ is given by
    \begin{equation}
        \label{appdx-eqn-001}
        \mathbb{P}\left[X_{i,t+1}^{N_1} = x' \big\vert X_{i,t}^{N_1}, \bfX^{N_1}_t, \bfY^{N_2}_t\right] = \sum_{u} f_t^\rho(x'|x, u, \M^{N_1}_t, \N^{N_1}_t) \phi_t(u|x, \M^{N_1}_t, \N^{N_1}_t) \triangleq \M_{t+1}^{x}(x'),
    \end{equation}
    where $\M^{N_1}_t = \empMu{\bfX^{N_1}_t}$ and $\N^{N_2}_t = \empNu{\bfY^{N_2}_t}$.
    Define the \ED{} of these $N_{1,t}^x$ Blue agents at time $t+1$ as
    \begin{equation*}
        \M_{t+1}^{N_{1,t}^x}(x') = \frac{1}{N_{1,t}^x} \sum_{i \in [N_{1, t}^x]} \indicator{x'}{X_{i,t+1}^{N_1}}.
    \end{equation*}
    Then, we have the following inequality due to Corollary~\ref{cor:l2-weak-law-iid}:
        \begin{equation}
            \label{appdx-eqn:apprx-err-001}
            \mathbb{E}_{\phi_t} \left[\dtv{\M^{N_{1,t}^{x}}_{t+1}, \M^{x}_{t+1}}\Big \vert \bfX^{N_1}_t, \bfY^{N_2}_t\right] \leq \frac{1}{2} \sqrt{\frac{|\X|}{N_{1,t}^x}}.
        \end{equation}
    
    Note that we can compute the \ED{} of the whole Blue team  based on the \ED{} of sub-populations of Blue agents on different states via
    \begin{equation}
        \label{appdx-eqn:001}
        \begin{aligned}
        \M_{t+1}^{N_1}(x') &= \frac{1}{N_1}\sum_{i =0}^{N_1} \indicator{x'}{X_{i,t+1}^{N_1}} = \frac{1}{N_1} \sum_{x\in \X} \sum_{i\in [N_{1, t}^x]}\indicator{x'}{X_{i,t+1}^{N_1}} = \sum_{x\in\X} \frac{N_{1,t}^x}{N_1} \M_{t+1}^{N_{1,t}^x}(x').
        \end{aligned}
    \end{equation}

    Similarly, we can relate the propagated mean-field of the whole team with the ones of each sub-population via
    \begin{equation}
    \label{appdx-eqn:002}
        \begin{aligned}
            \M^\rho_{t+1} (x') &= \sum_{x \in \X} \left[\sum_{u \in \U} f^\rho_t(x'|x, u, \M^{N_1}_t, \N^{N_2}_t )\; \phi_t(u|x,\M^{N_1}_t, \N^{N_2}_t) \right] \M^{N_1}_t(x) \\
            &=\sum_{x \in \X} \M^x_{t+1}(x') \M_t^{N_1}(x) 
            = \sum_{x \in \X} \frac{N_{1,t}^{x}}{N_1} \M^x_{t+1}(x').
        \end{aligned}
    \end{equation}

    Combining~\eqref{appdx-eqn:apprx-err-001}, \eqref{appdx-eqn:001}, and~\eqref{appdx-eqn:002}, we have
    \begin{align*}
        \expct{\dtv{\M^{N_1}_{t+1}, \M^{\rho}_{t+1}}\big\vert \bfX^{N_1}_t, \bfY^{N_2}_t}{\phi_t}
        &=\expct{\frac{1}{2}\sum_{x'}\abs{\M_{t+1}^{N_1}(x')- \M^\rho_{t+1}(x')} \vert\; \bfX^{N_1}_t, \bfY^{N_2}_t}{\phi_t} \\
        &= \expct{\frac{1}{2}\sum_{x'}\abs{\sum_{x} \frac{N_{1,t}^{x}}{N_1} \M_{t+1}^{N^x_{1,t}}(x')- \sum_{x} \frac{N_{1,t}^{x}}{N_1} \M^x_{t+1}(x')}\Big\vert\; \bfX^{N_1}_t, \bfY^{N_2}_t}{\phi_t}\\
        &\leq \sum_{x} \frac{N_{1,t}^{x}}{N_1} \expct{\frac{1}{2}\sum_{x'}\abs{\M_{t+1}^{N^x_{1,t}}(x')- \M^x_{t+1}(x')} \big\vert\; \bfX^{N_1}_t, \bfY^{N_2}_t}{\phi_t}
        \\
        &= \sum_{x} \frac{N_{1,t}^{x}}{N_1}\expct{\dtv{\M_{t+1}^{N^x_{1,t}}, \M^x_{t+1}}\big\vert \bfX^{N_1}_t, \bfY^{N_2}_t}{\phi_t} \\
        &\leq \sum_{x} \frac{N_{1,t}^{x}}{2N_1}\sqrt{\frac{|\X|}{N_{1,t}^x}}
        \leq \frac{|\X|}{2} \sqrt{\frac{1}{N_1}},
    \end{align*}
    where the last inequality is a result of Lemma~\ref{appdx-lmm:sum}.%
    \footnote{Let $a_x = \sqrt{N_{1,t}^{x}}$ for each $x \in \X$, it follows that $a_x \geq 0$ and $\sum_{x} (a_x)^2 = N_1$ almost surely, which allows the application of Lemma~\ref{appdx-lmm:sum}.}
\end{proof}

\subsection{Proof of Lemma~\ref{lmm:mf-apprx-team-policy}}
\label{appdx:mf-reachability}
\mfa*

\begin{proof}
Note that if $\M^{N_1}_t(x) = 0$, then the state $x$ has no contribution to the next mean-field propagated via~\eqref{eqn:blue-apprx-mf}. 
The second case in~\eqref{eqn:apprx-local-policy} ensures that the constructed policy is well-defined.
Without loss of generality, we assume that $\M^{N_1}_t(x) > 0$  for all $x \in \X$.

Let $N_{1, t}^x = N_1 \M^{N_1}_t(x)$ be the number of Blue agents that are on state $x$ at time $t$ in the finite-population system. 
With a slight abuse of notation, we use $[N_{1, t}^x]$ to denote the set of the Blue agents on state $x$.
Then, for agent $i \in [N_{1, t}^x]$ applying policy $\phi_{i,t}$, its conditional distribution of its next state $X^{N_1}_{i,t+1}$ is given by
\begin{equation*}
    p_{i}^x (x') = \mathbb{P}(X^{N_1}_{i, t+1} = x' | \bfX^{N_1}_t, \bfY^{N_2}_t) = \sum_{u \in \U} f^\rho_t(x'|u, \M^{N_1}_t, \N^{N_2}_t) \phi_{i, t}(u|x, \M^{N_1}_t, \N^{N_2}_t)
\end{equation*}
Define the \ED{} of these $N^x_{1,t}$ Blue agents at time $t+1$ as 
\begin{equation}
    \M^{N_{1,t}^x}_{t+1}(x') = \frac{1}{N^{x}_{1,t}} \sum_{i \in [N^{x}_{1,t}]}\indicator{x'}{X_{i,t+1}^{N_1}}.
\end{equation}

On the other hand, the mean-field propagated from state $x$ under local policy $\pi_{\apprx,t}$ can be expressed as 
\begin{equation}
    \begin{aligned}
        \M^{x}_{\apprx, t+1} (x') &= \sum_{u} f_t^\rho(x'|x,u, \M^{N_1}_t, \N^{N_2}_t) \pi_{\apprx, t}(u|x) \\
        &=\sum_{u} f_t^\rho(x'|x,u, \M^{N_1}_t, \N^{N_2}_t) \frac{\sum_{i \in [N^x_{1,t}]}\phi_{i,t}(u|x, \M^{N_1}_t, \N^{N_2}_t)}{N^x_{1,t}}\\
        &=\frac{\sum_{i \in [N^x_{1,t}]}p^x_i(x')}{N^x_{1,t}}.
    \end{aligned}  
\end{equation}

From Lemma~\ref{lmm:l2-apprx}, we have that 
\begin{equation}
    \expct{\dtv{\M^{N_{1,t}^x}_{t+1}, \M_{\apprx, t+1}^x}\big\vert \bfX^{N_1}_t, \bfY^{N_2}_t}{} \leq \frac{1}{2} \sqrt{\frac{|\X|}{N_{1,t}^x}}.
\end{equation}

The rest of the analysis is similar to the proof of Lemma~\ref{lmm:mf-aprx}.
Notice that 
\begin{equation*}
    \begin{aligned}
        \M_{t+1}^{N_1}(x') &=  \sum_{x\in\X} \frac{N_{1,t}^x}{N_1} \M_{t+1}^{N_{1,t}^x}(x'), \\
        \M_{\apprx, t+1}(x') &=  \sum_{x\in\X} \frac{N_{1,t}^x}{N_1} \M_{\apprx, t+1}^{x}(x').
    \end{aligned}
\end{equation*}
It then follows
\begin{align*}
        \expct{\dtv{\M_{t+1}^{N_1}, \M_{\apprx, t+1}}\big\vert \bfX^{N_1}_t, \bfY^{N_2}_t}{} 
        &\leq \sum_{x} \frac{N_{1,t}^x}{N_1}\expectation{\dtv{\M_{t+1}^{N_{1,t}^x}, \M_{\apprx, t+1}^{x}}\big\vert \bfX^{N_1}_t, \bfY^{N_2}_t} \\
        & \leq \frac{|\X|}{2} \sqrt{\frac{1}{N_1}},
    \end{align*}
    where the last inequality results from Lemma~\ref{appdx-lmm:sum}.

\end{proof}

\section{Proof of the Continuity Results}
\label{appdx-sec:Continuity}

In this section, we present the proofs for Lemma~\ref{lmm:R-L-cont} and Theorem~\ref{thm:value-L-cont}. We start with several preliminary results.

\subsection{Preliminary results} 
\label{appdx-sec:Continuity-Prelim}

\begin{lemma}
    \label{appdx-lmm:max-diff-set}
    Let $\xi: \A \times \B \to \mathbb{R}$ be a Lipschitz continuous function in its second argument, such that $\abs{\xi(a, b) - \xi(a,b')} \leq L_\xi \norm{b-b'}$ for all $a \in \A$ and $b,b'\in \B$.
   Let compact sets $Q, Q' \subseteq \B$ such that $\distH{Q, Q'} \leq \epsilon$.
   Then, for all $a \in \A$,
    \begin{subequations}
    \begin{align}
        \abs{\min_{b\in Q} \xi(a,b) - \min_{b' \in Q'} \xi(a,b')} &\leq L_\xi \epsilon. \label{eqn:lemma10-1}
        \\
        \abs{\max_{b\in Q} \xi(a,b) - \max_{b' \in Q'} \xi(a,b')} &\leq L_\xi \epsilon. \label{eqn:lemma10-2}
    \end{align}
    \end{subequations}
\end{lemma}

\begin{proof}
    We start with the proof for~\eqref{eqn:lemma10-1}. 
    Fix $a\in \A$ and consider the minimizer $b^* \in \argmin_{b \in Q} \xi(a,b)$. 
    From the assumption on the distance between $Q$ and $Q'$, there exists a $b^{* \prime} \in Q'$ such that 
    $\norm{b^{*\prime} - b^*} \leq \epsilon$.
    Since $\xi$ is Lipschitz with respect to $b$, it follows that 
    \begin{equation}
        \abs{\xi(a, b^*) - \xi(a, b^{*\prime})} \leq L_\xi \epsilon.
    \end{equation}
    Then, 
    \begin{align*}
        \min_{b\in Q} \xi(a,b) = \xi(a, b^*) \geq \xi(a, b^{*\prime}) - L_\xi \epsilon \geq \min_{b'\in Q'} \xi(a,b') - L_\xi \epsilon.
    \end{align*}
    Similarly, one can show that $\min_{b'\in Q'} \xi(a,b') \geq \min_{b\in Q} \xi(a,b) -  L_\xi \epsilon$, and thus we have~\eqref{eqn:lemma10-1}.

    The result~\eqref{eqn:lemma10-1} can be easily extended to the maximization case in~\eqref{eqn:lemma10-2} by making use of the fact that $\max_{b\in Q} \xi(a,b) = - \min_{b\in Q} -\xi(a,b)$.
\end{proof}

\begin{lemma}
    \label{appdx-lmm:min-lip}
    Let $\zeta:\A \times \B \to \mathbb{R}$ be a Lipschitz continuous function in its first argument, such that $\abs{\zeta(a,b) - \zeta(a',b)}\leq L_\zeta \norm{a-a'}$ for all $a,a' \in \A$ and $b\in \B$.
    Then,  for all compact sets $Q \subseteq \B$,
    \begin{equation*}
        \abs{\min_{b \in Q} \zeta(a,b) - \min_{b'\in Q} \zeta(a',b')} \leq L_\zeta \norm{a-a'}.
    \end{equation*}
\end{lemma}

\begin{proof}
    Since $Q$ is compact and $\zeta$ is continuous, the two minimization problems are well-defined.
    Next, let $b^* \in \argmin_{b \in Q} \zeta(a,b)$, it follows that
    \begin{equation*}
        \min_{b'\in Q} \zeta(a',b') - \min_{b \in Q} \zeta(a,b) = \min_{b'\in Q} \zeta(a',b') - \zeta(a,b^*) \leq \zeta(a',b^*) - \zeta(a,b^*) \stackrel{\mathrm{(i)}}{\leq} L_\zeta\norm{a-a'},
    \end{equation*}
    where (i) is due to the Lipschitz assumption.
    Similarly, one can show the other direction $\min_{b \in Q} \zeta(a,b) -\min_{b'\in Q} \zeta(a',b') \leq L_\zeta\norm{a-a'}$.
\end{proof}

\begin{lemma}
    \label{lmm:min-max-marginal-L-cont}
    Consider the compact-valued correspondences $\Gamma: \X \times \Y \rightsquigarrow \X$ and $\Theta: \X \times \Y \rightsquigarrow \Y$, which are Lipschitz continuous with Lipschitz constants $L_\Gamma$ and $L_\Theta$, respectively.
    Given a $L_g$-Lipschitz continuous
    real-valued function $g: \X \times \Y \to \mathbb{R}$, the max-min marginal function
    \begin{equation}
        \label{eqn:max-min-marginal-def}
        f(x,y) = \max_{p\in \Gamma(x,y)} \min_{q\in \Theta(x,y)} g(p,q)
    \end{equation}
    is Lipschitz continuous with Lipschitz constant $L_g (L_\Gamma + L_\Theta)$. 
\end{lemma}

\begin{proof}
    Let $h(x, y, p) = \min_{q\in \Theta(x,y)}g(p,q) $. 
    Since $g(p,q)$ is continuous and $\Theta(x,y)$ is compact, the minimization is well-defined for each $p$.
    Fix $p \in \X$, and consider $x,x' \in \X$ and $ y,y'\in \Y$. 
    The Lipschitz continuity of $\Theta$ implies that
    \begin{equation}
        \distH{\Theta(x,y), \Theta(x',y')} \leq L_\Theta(\norm{x-x'}+ \norm{y-y'}).
        \label{eqn:marginal-3}
    \end{equation}
    For simplicity, let $\epsilon = (\norm{x-x'}+ \norm{y-y'})$.
    Leveraging Lemma~\ref{appdx-lmm:max-diff-set}, we have%
    \footnote{
    Relating to Lemma~\ref{appdx-lmm:max-diff-set}:
    Set $g$ as function $\xi$; 
    treat the argument $p$ as $a$ and the optimization variable $q$ as $b$;
    regard the sets $\Theta(x, y)$ and $\Theta(x', y')$ as $Q$ and $Q'$ respectively.}
    \begin{equation}
        \abs{h(x,y,p) - h(x',y',p)} \leq L_g L_\Theta \epsilon. \label{eqn:marginal-0}
    \end{equation}
    
    Next, fix $(x, y) \in \X \times \Y$ and consider $p, p' \in \Y$. It follows from Lemma~\ref{appdx-lmm:min-lip} that%
    \footnote{Relating to Lemma~\ref{appdx-lmm:min-lip}:
    Set $g$ as function $\zeta$;
    treat the arguments $p$ and $q$ as $a$ and $b$, respectively;
    regard the set $\Theta(x, y)$ as $Q$.
    }
    \begin{align}
        \label{eqn:marginal-2}
        \abs{h(x,y,p) - h(x,y,p')} = \abs{\min_{q \in \Theta(x,y)} g(p,q) - \min_{q' \in \Theta(x,y)}g(p', q')} \leq L_g \norm{p-p'}.
    \end{align}
    Finally, consider $x,x' \in \X$ and $ y,y'\in \Y$, it follows that
    %
    \begin{align}
        \big\vert f(x,y) &- f(x',y')\big\vert   = \abs{\max_{p \in \Gamma(x,y)} h(x,y,p)- \max_{p' \in \Gamma(x',y')} h(x',y',p')}\\
        &\leq \abs{\max_{p \in \Gamma(x,y)} h(x,y,p)- \max_{p \in \Gamma(x,y)} h(x',y',p)} + 
        \abs{\max_{p \in \Gamma(x,y)} h(x',y',p)- \max_{p' \in \Gamma(x',y')} h(x',y',p')} \label{eqn:marginal-1}\\
        &\leq L_g L_\Theta \epsilon + L_\Gamma L_g \epsilon
        = L_g(L_\Theta + L_\Gamma) \epsilon,
    \end{align}
    where the first term in~\eqref{eqn:marginal-1} is bounded using Lemma~\ref{appdx-lmm:min-lip} and the Lipschitz constant derived in~\eqref{eqn:marginal-0} %
    \footnote{
    Relating to Lemma~\ref{appdx-lmm:min-lip}:
    Set $h$ as function $\zeta$;
    treat the arguments $(x,y)$ as $a$ and the argument $p$ as $b$;
    regard the set $\Gamma(x, y)$ as the optimization domain $Q$.
    },
    and the second term in~\eqref{eqn:marginal-1} is bounded using Lemma~\ref{appdx-lmm:max-diff-set} and the Lipschitz constant in~\eqref{eqn:marginal-2}
    \footnote{
    Relating to Lemma~\ref{appdx-lmm:max-diff-set}:
    Set $h$ as function $\xi$; 
    treat the argument $(x', y')$ as $a$ and the optimization variable $p$ as $b$;
    regard the sets $\Gamma(x, y)$ and $\Gamma(x', y')$ as $Q$ and $Q'$ respectively.
    }. 
\end{proof}

Next, we present some results regarding the continuity of the reachability correspondences. 
We start by defining the pure local policies.

\begin{definition}[Pure local policy]
    \label{appdx-def:pure-policy}
    A Blue local policy $\pi_t \in \Pi_t$ is said to be pure if $\pi_t(u|x) \in \{0,1\}$ for all $u \in \U$ and $x \in \X$.
    We use $\hat{\Pi}_t = \{\hat{\pi}^k\}_{k=1}^{|\U|^{|\X|}}$ to denote the set of pure Blue policies, where $\hat{\pi}^k$ denotes the $k$-th pure Blue local policy. 
    The pure Red local policies are defined similarly.
\end{definition}

\begin{proposition}
    \label{appdx-lmm:Rset-convex-hull}
    The Blue reachable set is characterized as 
    \begin{equation}
        \label{appdx-eqn:Rset-characterization}
        \R^\rho_{\mu,t}(\mu_t, \nu_t) = \Co{\{\mu_t F_t^\rho(\mu_t,\nu_t, \hat{\pi}^k\}_{k=1}^{|\U|^{|\X|}}},
    \end{equation}
    where $\hat{\pi}^k$ are pure Blue local policies, and $\mathrm{Co}$ denotes the convex hull. 
\end{proposition}

\begin{proof}    
    We first show that the mean-field propagation rule is linear with respect to the policy. 
    Note that the set of admissible policies $\Pi_t$ can be characterized as the convex hull of the pure policies, i.e., $\Pi_t = \Co{\{\hat{\pi}^k\}_k}$.
    Consider arbitrary current mean-fields $\mu_t$ and $\nu_t$,
    and consider the policy $\pi_t = \lambda \pi_t^1 + (1-\lambda) \pi_t^2$ for some $\lambda \in (0,1)$ and $\pi_t^1, \pi_t^2 \in \Pi_t$.
    Then, the next mean-field induced by $\pi_t$ is given by
    \begin{align*}
        \mu_{t+1}(x') &= \sum_{x} \left[\sum_u f^\rho_t(x'|x, \mu_t, \nu_t) \pi_t(u|x)\right]\mu_t(x) \\
        &= \lambda \sum_{x} \left[\sum_u f^\rho_t(x'|x, \mu_t, \nu_t) \pi^1_t(u|x)\right]\mu_t(x) + (1-\lambda) \sum_{x} \left[\sum_u f^\rho_t(x'|x, \mu_t, \nu_t) \pi^2_t(u|x)\right]\mu_t(x),
    \end{align*}
    which implies that
    \begin{equation}
        \label{appd-eqn:linear-mf-dynamics}
        \mu_{t+1} = \mu_t F^\rho_t(\mu_t, \nu_t, \pi_t) = \lambda \mu_t F^\rho_t(\mu_t, \nu_t, \pi^1_t) + (1-\lambda) \mu_t F^\rho_t(\mu_t, \nu_t, \pi^2_t).
    \end{equation}

    For the $\subseteq$ direction in~\eqref{appdx-eqn:Rset-characterization}, consider $\mu_{t+1} \in \R^\rho_{\mu,t}(\mu_t, \nu_t)$. Then there exists a policy $\pi_t$ such that $\mu_{t+1} = \mu_{t}F^\rho_t(\mu_t, \nu_t, \pi_t)$.
    Since the admissible policy set is the convex hull of pure policies, we have $\pi_t = \sum_k \lambda_k \hat{\pi}^k$, where $\lambda_k \geq 0$ and $\sum_k \lambda_k = 1$.
    It follows directly from~\eqref{appd-eqn:linear-mf-dynamics} that
    \begin{equation}
        \mu_{t+1} = \sum_{k} \lambda_k \mu_{t}F^\rho_t(\mu_t, \nu_t, \hat{\pi}^k) \in \Co{\{\mu_t F_t^\rho(\mu_t,\nu_t, \hat{\pi}^k\}_{k=1}^{|\U|^{|\X|}}}.
    \end{equation}
    
    For the $\supseteq$ direction in~\eqref{appdx-eqn:Rset-characterization}, consider a point $\mu_{t+1} \in \Co{\{\mu_t F_t^\rho(\mu_t,\nu_t, \hat{\pi}^k)\}_{k=1}^{|\U|^{|\X|}}}$. 
    By definition and the linearity in~\eqref{appd-eqn:linear-mf-dynamics}, we have 
    \[\mu_{t+1} = \sum_{k} \lambda_k  \mu_t F_t^\rho(\mu_t,\nu_t, \hat{\pi}^k)
    =\mu_t F_t^\rho(\mu_t,\nu_t, \pi_t),\]
    where $\pi_t = \sum_k {\lambda_k \hat{\pi}^k} \in \Pi_t$, which implies $\mu_{t+1} \in \RMSet{t}{\rho}$    
\end{proof}

\begin{lemma}
    \label{appdx-lmm:vertex-perturb}
    Consider a pure local policy $\hat{\pi}_t \in \hat{\Pi}_t$ and arbitrary mean-fields $\mu_t, \mu'_t \in \P(\X)$ and $\nu_t, \nu'_t\in \P(\Y)$. The following bound holds: 
    \begin{equation*}
        \dtv{\mu_t F^\rho_t(\mu_t, \nu_t, \hat{\pi}_t), \mu_t' F^\rho_t(\mu'_t, \nu'_t, \hat{\pi}_t)} \leq \big(1+ \frac{1}{2} L_{f_t}\big) \big(\dtv{\mu_t, \mu'_t} + \dtv{\nu_t, \nu'_t}\big).
    \end{equation*}

\end{lemma}
\begin{proof}
    Using the triangle inequality, we have that
    \begin{align*}
        \dtv{\mu_t F^\rho_t(\mu_t, \nu_t, \hat{\pi}_t), \mu_t' F^\rho_t(\mu'_t, \nu'_t, \hat{\pi}_t)} &\leq \underbrace{\dtv{\mu_t F^\rho_t(\mu_t, \nu_t, \hat{\pi}_t), \mu_t' F^\rho_t(\mu_t, \nu_t, \hat{\pi}_t)}}_{A} \\
        & \qquad \qquad \qquad \qquad \qquad \qquad + \underbrace{\dtv{\mu'_t F^\rho_t(\mu_t, \nu_t, \hat{\pi}_t), \mu_t' F^\rho_t(\mu'_t, \nu'_t, \hat{\pi}_t)}}_{B}. 
    \end{align*}
    
    Since $F^\rho_t$ is a stochastic matrix, it is a non-expansive operator under the 1-norm, and hence
    $A \leq \dtv{\mu_t, \mu'_t}$. 
    Next, we bound the term $B$ using Assumption~\ref{assmpt:lipschitiz-dynamics}. 
    \begin{align*}
        B 
        &= \frac{1}{2} \sum_{x'} \abs{\sum_x \mu'_t(x) \Big(\sum_{u}\hat{\pi}_t(u|x) 
        \big(f_t^\rho(x'|x, \mu_t, \nu_t, u) - f_t^\rho(x'|x, \mu'_t, \nu'_t, u)\big) \Big)}\\
        &\leq  \frac{1}{2} \sum_{x',x,u} \mu_t'(x) \hat{\pi}_t(u|x) \big\vert f_t^\rho(x'|x, \mu_t, \nu_t, u) - f_t^\rho(x'|x, \mu'_t, \nu'_t, u)\big \vert \\
        &= \frac{1}{2} \sum_{x,u} \mu_t'(x) \hat{\pi}_t(u|x) \Big( \sum_{x'} \big\vert f_t^\rho(x'|x, \mu_t, \nu_t, u) - f_t^\rho(x'|x, \mu'_t, \nu'_t, u)\big \vert \Big)\\
        & \leq \frac{1}{2} \sum_{x,u} \mu_t'(x) \hat{\pi}_t(u|x) L_{f_t} \Big( \dtv{\mu_t, \mu'_t} + \dtv{\nu_t, \nu_t'}\Big) = \frac{1}{2} L_{f_t} \Big( \dtv{\mu_t, \mu'_t} + \dtv{\nu_t, \nu_t'} \Big)
     \end{align*}

    Combining the bounds, we have 
    \begin{equation*}
        \dtv{\mu_t F^\rho_t(\mu_t, \nu_t, \hat{\pi}_t), \mu_t' F^\rho_t(\mu'_t, \nu'_t, \hat{\pi}_t)} \leq \big(1+ \frac{1}{2} L_{f_t}\big) \dtv{\mu_t, \mu'_t} + \frac{1}{2}L_{f_t} \dtv{\nu_t, \nu'_t}.
    \end{equation*}
    
\end{proof}

\begin{lemma}
    \label{appdx-lmm:vertex-movement}
    Let two sets of points $\{x_i\}_{i=1}^N$ and $\{y_i\}_{i=1}^N$ that satisfy  $\norm{x_i - y_i} \leq \epsilon$ for all $i$. 
    Then,
    \begin{equation*}
        \distH{\Co{\{x_i\}_{i=1}^N}, \Co{\{y_i\}_{i=1}^N}} \leq \epsilon.
    \end{equation*}
\end{lemma}

\begin{proof}
    Consider a point $x \in \Co{\{x_i\}_{i=1}^N}$. 
    By definition, there exists a set of non-negative coefficients $\{\lambda_i\}_{i=1}^N$, such that $x = \sum_{i} \lambda_i x_i$ and $\sum_i \lambda_i = 1$.
    Using the same set of coefficients, define $y = \sum_i \lambda_i y_i  \in \Co{\{y_i\}_{i=1}^N}$.
    Then, 
    \begin{equation*}
        \norm{x-y} \leq \sum_{i} \lambda_i \norm{x_i - y_i} \leq \epsilon.
    \end{equation*}
    Thus, for a fixed $x \in \Co{\{x_i\}_{i=1}^N}$, we have
    $\inf_{y \in \Co{\{y_i\}_{i=1}^N}} \norm{x-y} \leq \epsilon.$
    Since $x \in \Co{\{x_i\}_{i=1}^N}$ is arbitrary, the above inequality implies 
    \[\sup_{x\in \Co{\{x_i\}_{i=1}^N}} \inf_{y \in \Co{\{y_i\}_{i=1}^N}} \norm{x-y} \leq \epsilon.\]
    Through a similar argument, we can conclude that $\sup_{y\in \Co{\{y_i\}_{i=1}^N}} \inf_{x \in \Co{\{x_i\}_{i=1}^N}} \norm{x-y} \leq \epsilon,$ which completes the proof.
\end{proof}

\subsection{Proof of Lemma~\ref{lmm:R-L-cont}}
\label{appdx-sec:Continuity-RSet}
\RLC*

\begin{proof}
    Due to symmetry, we only present the proof for the Blue reachability correspondence. 
    From Lemma~\ref{appdx-lmm:vertex-perturb}, we have that, for all pure local policies $\hat{\pi}^k_t \in \hat{\Pi}_t$,
        \begin{equation*}
        \dtv{\mu_t F^\rho_t(\mu_t, \nu_t, \hat{\pi}^k_t), \mu_t' F^\rho_t(\mu'_t, \nu'_t, \hat{\pi}_t)} \leq \big(1+ \frac{1}{2} L_{f_t}\big) \dtv{\mu_t, \mu'_t} + \frac{1}{2}L_{f_t} \dtv{\nu_t, \nu'_t}.
    \end{equation*}
    From Proposition~\ref{appdx-lmm:Rset-convex-hull}, we know that the reachable set can be characterized as the convex hull
    \begin{equation*}
        \R^\rho_{\mu,t}(\mu_t, \nu_t) = \Co{\{\mu_t F_t^\rho(\mu_t,\nu_t, \hat{\pi}^k\}_{k=1}^{|\U|^{|\X|}}}.
    \end{equation*}
    Leveraging Lemma~\ref{appdx-lmm:vertex-movement}, we can conclude that
    \begin{equation*}
        \distH{\mathcal{R}_{\mu,t}(\mu_t, \nu_t), \mathcal{R}_{\mu,t}(\mu'_t, \nu'_t))} \leq \big(1+ \frac{1}{2} L_{f_t}\big) \big(\dtv{\mu_t, \mu'_t} + \dtv{\nu_t, \nu'_t}\big).
    \end{equation*}

\end{proof}

\subsection{Proof of Theorem~\ref{thm:value-L-cont}}
\label{appdx-sec:Continuity-RSet-thm-2}

\vlc*
\begin{proof}
    The proof is given by induction.
    
    \textit{Base case: } At the terminal time $T$, the optimal coordinator value function is given by  
    \begin{equation*}
        \lowervalue_{\cor,T}^{\rho*}(\mu^\rho_T, \nu^\rho_T) = r^\rho_T(\mu^\rho_T, \nu^\rho_T).
    \end{equation*}
    From Assumption~\ref{assmpt:lipschitiz-rewards}, the Lipschitz constant for $\lowervalue_{\cor, T}^{\rho*}$ is $L_r$.
    
    \vspace{+5pt}
    \textit{Inductive hypothesis: }
    Suppose that the Lipschitz constant is given by~\eqref{eqn:value-L-constant} at time $t+1$.
    
    \vspace{+5pt}
    \textit{Induction: }
    Recall the definition of the optimal coordinator lower value at time $t$:
    \begin{equation*}
        \lowervalue_{\cor,t}^{\rho*}(\mu^\rho_t, \nu^\rho_t) = r_t(\mu^\rho_t, \nu^\rho_t) +
        \max_{\mu_{t+1}^\rho \in \mathcal{R}_{\mu,t}^\rho(\mu_t^\rho,\nu_t^\rho)}~ \min_{\nu_{t+1}^\rho \in \mathcal{R}_{\nu,t}^\rho(\mu_t^\rho,\nu_t^\rho)}
        \lowervalue_{\cor,t+1}^{\rho*}(\mu^\rho_{t+1}, \nu^\rho_{t+1}),
    \end{equation*}
    where the second term takes the form of a max-min marginal function defined in~\eqref{eqn:max-min-marginal-def}, where $\Gamma = \R^\rho_{\mu,t}$, $\Theta = \R^\rho_{\nu,t}$ and $g = \lowervalue_{\cor,t+1}^{\rho*}$, while the first term is Lipschitz continuous with Lipschitz constant $L_r$.
    Applying Lemma~\ref{lmm:min-max-marginal-L-cont}, it follows that, for all $\mu^\rho_t, \mu^{\rho\prime}_t \in \P(\X)$ and $\nu^\rho_t, \nu^{\rho\prime}_t \in \P(\Y)$,
    %
    \begin{align*}
        &\Big \vert \lowervalue_{\cor,t}^{\rho*}(\mu^\rho_t, \nu^\rho_t) - \lowervalue_{\cor,t}^{\rho*}(\mu^{\rho\prime}_t, \nu^{\rho\prime}_t)\Big \vert  \leq \abs{r_t(\mu^\rho_t, \nu^\rho_t) - r_t(\mu^{\rho\prime}_t, \nu^{\rho\prime}_t)} + \\
        & \abs{
        \max_{\mu_{t+1}^\rho \in \mathcal{R}_{\mu,t}^\rho(\mu_t^\rho,\nu_t^\rho)}~ \min_{\nu_{t+1}^\rho \in \mathcal{R}_{\nu,t}^\rho(\mu_t^\rho,\nu_t^\rho)}
        \lowervalue_{\cor,t+1}^{\rho*}(\mu^\rho_{t+1}, \nu^\rho_{t+1})
        -
        \max_{\mu_{t+1}^{\rho\prime} \in \mathcal{R}_{\mu,t}^{\rho\prime}(\mu_t^{\rho\prime},\nu_t^{\rho\prime})}~ \min_{\nu_{t+1}^{\rho\prime} \in \mathcal{R}_{\nu,t}^{\rho\prime}(\mu_t^{\rho\prime},\nu_t^{\rho\prime})}
        \lowervalue_{\cor,t+1}^{\rho*}(\mu^{\rho\prime}_{t+1}, \nu^{\rho\prime}_{t+1})}\\
        &\leq \bigg(L_r + (L_{\R^\rho_{\mu,t}} + L_{\R^\rho_{\nu,t}}) L_r \Big( 1+ \sum_{k=t+1}^{T-1} \prod_{\tau =t+1}^{k} (L_{\R^\rho_{\mu,\tau}} + L_{\R^\rho_{\nu,\tau}})\Big)\bigg) \Big(\dtv{\mu^\rho_t, \mu^{\rho\prime}_t} + \dtv{\nu^\rho_t, \nu^{\rho\prime}_t}\Big)\\
        & = \bigg(L_r \Big( 1+ (L_{\R^\rho_{\mu,t}} + L_{\R^\rho_{\nu,t}}) + \sum_{k=t+1}^{T-1} (L_{\R^\rho_{\mu,t}} + L_{\R^\rho_{\nu,t}}) \prod_{\tau =t+1}^{k}  (L_{\R^\rho_{\mu,\tau}} + L_{\R^\rho_{\nu,\tau}})\Big)\bigg) \Big(\dtv{\mu^\rho_t, \mu^{\rho\prime}_t} + \dtv{\nu^\rho_t, \nu^{\rho\prime}_t}\Big)\\
        &=\bigg(L_r \Big( 1+ (L_{\R^\rho{_\mu,t}} + L_{\R^\rho_{\nu,t}}) + \sum_{k=t+1}^{T-1} \prod_{\tau =t}^{k} (L_{\R^\rho_{\mu,\tau}} + L_{\R^\rho_{\nu,\tau}})\Big)\bigg) \Big(\dtv{\mu^\rho_t, \mu^{\rho\prime}_t} + \dtv{\nu^\rho_t, \nu^{\rho\prime}_t}\Big) \\
        &=\bigg(L_r \Big( 1+ \sum_{k=t}^{T-1} \prod_{\tau =t}^{k} (L_{\R^\rho_{\mu,\tau}} + L_{\R^\rho_{\nu,\tau}})\Big)\bigg) \Big(\dtv{\mu^\rho_t, \mu^{\rho\prime}_t} + \dtv{\nu^\rho_t, \nu^{\rho\prime}_t}\Big),
    \end{align*}
    which completes the induction.
\end{proof}

\section{Proof of Existence of Game Values}
\label{appdx-sec:game-value}

We first examine the convexity of the reachability correspondences. 
\begin{definition}[\citep{kuroiwa1996convexity}]
    Let $\X$, $\Y$ and $\Z$ be convex sets. 
    The correspondence $\Gamma: \X \times \Y \rightsquigarrow \Z$ is convex with respect to $\X$ if, for all $x_1, x_2 \in \X$, $y \in \Y$,  $z_1 \in \Gamma(x_1, y)$, $z_2 \in \Gamma(x_2, y)$, and $\lambda \in (0,1)$, there exists $z \in \Gamma(\lambda x_1 + (1-\lambda) x_2, y)$ such that 
    \begin{equation*}
        z = \lambda z_1 + (1-\lambda) z_2.
    \end{equation*}
\end{definition}

\begin{remark} \label{rmk:corespondence-convexity}
    The above definition of convexity is equivalent to  the following set inclusion
    \begin{equation}
        \label{eqn:convex-correspondence-def}
        \Gamma(\lambda x_1 + (1-\lambda) x_2, y) \supseteq \lambda \Gamma(x_1,y) + (1-\lambda) \Gamma(x_2, y),
    \end{equation}
    for all $\lambda \in (0,1)$, $x_1,x_2 \in \X$ and $y \in \Y$,
    where the Minkowski sum is defined as     
    \[
    \lambda \Gamma(x_1,y) + (1-\lambda) \Gamma(x_2, y) = \{z| z=\lambda z_1 + (1-\lambda) z_2 \, \text{ where } \, z_1 \in \Gamma(x_1, y), z_2 \in \Gamma(x_2, y)\}.
    \]
\end{remark}

\begin{definition}
    The correspondence $\Gamma: \X \times \Y \rightsquigarrow \Z$ is constant with respect to $\Y$ if
    \begin{equation}
        \label{eqn:constant-correspondence-def}
        \Gamma(x, y_1) = \Gamma(x, y_2) \quad \forall x \in \X, ~\forall y_1, y_2 \in \Y.
    \end{equation}
    
\end{definition}

We have the following concave-convex result for a max-min marginal function defined in~\eqref{eqn:max-min-marginal-def}. 
\begin{restatable}{lemma}{lcc}
    \label{lmm:concave-convex}
    Consider two compact-valued correspondences $\Gamma: \X \times \Y \rightsquigarrow  \X$ and $\Theta: \X \times \Y \rightsquigarrow \Y$. 
    Let $\Gamma$ be convex with respect to $\X$ and constant with respect to $\Y$, and let $\Theta$ be constant with respect to $\X$ and convex with respect to $\Y$.
    Let $g: \X \times \Y \to \mathbb{R}$ be concave-convex. 
    Then, the max-min marginal function $f(x, y) = \max_{p\in \Gamma(x,y)} \min_{q \in \Theta(x,y)} g(p,q)$ 
    is also concave-convex.
\end{restatable}

\begin{proof}
    Fix $y \in \Y$, and consider $x_1, x_2 \in \X$ and $\lambda \in (0,1)$. Denote $x = \lambda x_1 + (1-\lambda) x_2$. It follows that
    \begin{align*}
        f(\lambda x_1 + (1-\lambda)x_2, y ) &= \max_{p \in \Gamma(\lambda x_1 + (1-\lambda)x_2, y)} \; \min_{q\in \Theta(\lambda x_1 + (1-\lambda)x_2, y)} g(p,q)\\
        &\stackrel{\text{(i)}}{\geq} 
        \max_{p \in \lambda \Gamma(x_1, y) + (1-\lambda) \Gamma(x_2, y)} ~~ \min_{q\in \Theta(x, y)} g(p, q) \\
        &= \max_{\substack{p_1 \in \Gamma(x_1,y)\\ p_2 \in \Gamma(x_2,y)}} \min_{q\in \Theta(x, y)} g(\lambda p_1 + (1-\lambda)p_2, q)\\
        &\stackrel{\text{(ii)}}{\geq}\max_{\substack{p_1 \in \Gamma(x_1,y)\\ p_2 \in \Gamma(x_2,y)}} \min_{q\in \Theta(x,y)} (\lambda g(p_1,q) + (1-\lambda)g(p_2, q)) \\
        &\stackrel{\text{(iii)}}{\geq} \max_{\substack{p_1 \in \Gamma(x_1,y)\\ p_2 \in \Gamma(x_2,y)}}  \Big( \underbrace{\min_{q_1\in \Theta(x,y)}\lambda g(p_1,q_1)}_{A(p_1)} + \underbrace{\min_{q_2\in \Theta(x,y)}(1-\lambda)g(p_2, q_2)}_{B(p_2)} \Big) \\
        & = \lambda \max_{p_1 \in \Gamma(x_1,y)} \min_{q_1 \in \Theta(x,y)} g(p_1, q_1) + (1-\lambda) \max_{p_2 \in \Gamma(x_2,y)} \min_{q_2 \in \Theta(x,y)} g(p_2, q_2)\\
        & \stackrel{\text{(iv)}}{=} \lambda \max_{p_1 \in \Gamma(x_1,y)} \min_{q_1 \in \Theta(x_1,y)} g(p_1, q_1) + (1-\lambda) \max_{p_2 \in \Gamma(x_2,y)} \min_{q_2 \in \Theta(x_2,y)} g(p_2, q_2)\\
        &= \lambda f(x_1, y) + (1-\lambda) f(x_2, y),
    \end{align*}
    where inequality (i) holds from restricting the maximization domain using the convexity of $\Gamma$ in~\eqref{eqn:convex-correspondence-def}; 
    inequality (ii) holds from the assumption that $g$ is concave with respect to its $p$-argument; 
    inequality (iii) is the result of distributing the minimization;
    equality (iv) is due to $\Theta$ being constant with respect to $\X$.
    The above result implies the concavity of $f$ with respect to its $x$-argument.
    
    Fix $x \in \X$, and let $y_1, y_2 \in \Y$ and $\lambda \in (0,1)$. Denote $y = \lambda y_1 + (1-\lambda) y_2$. Then,
    \begin{align*}
        f(x, \lambda y_1 + (1-\lambda) y_2) &= \max_{p \in \Gamma(x,\lambda y_1 + (1-\lambda) y_2)} \; \min_{q\in \Theta(x,\lambda y_1 + (1-\lambda) y_2)} g(p,q)\\
        &\leq \max_{p\in \Gamma(x,y)} \min_{\substack{q_1\in \Theta(x,y_1) \\ q_2 \in \Theta(x, y_2)}} g(p, \lambda q_1 + (1-\lambda) q_2)\\
        &\leq \max_{p\in \Gamma(x,y)} \min_{\substack{q_1\in \Theta(x, y_1) \\ q_2 \in \Theta(x, y_2)}}(\lambda g(p,q_1) + (1-\lambda)g(p, q_2)) \\
        & \leq \lambda \max_{p_1 \in \Gamma(x, y)} \min_{q_1 \in \Theta(x,y_1)} g(p_1, q_1) + (1-\lambda) \max_{p_2 \in \Gamma(x, y)} \min_{q_2 \in \Theta(x,y_2)} g(p_2, q_2)\\
        & = \lambda \max_{p_1 \in \Gamma(x, y_1)} \min_{q_1 \in \Theta(x,y_1)} g(p_1, q_1) + (1-\lambda) \max_{p_2 \in \Gamma(x, y_2)} \min_{q_2 \in \Theta(x,y_2)} g(p_2, q_2)\\
        &= \lambda f(x, y_1) + (1-\lambda) f(x, y_2),
    \end{align*}
    which implies that $f$ is convex with respect to its $y$-argument.
\end{proof}

Recall the definition of independent dynamics:
\IDD*

The following lemma shows the convexity of the reachability correspondences $\R^\rho_{\mu,t}$ and $\R^\rho_{\nu,t}$ under independent dynamics.

\begin{restatable}{lemma}{lrc}
    \label{lmm:rset-indep-dynamics}
    Under the independent dynamics in Definition~\ref{def:independent-dynamics}, the reachability correspondences satisfy
    \begin{subequations}
        \begin{align}
            \R^\rho_{\mu,t}(\mu_t, \nu_t) &= \R^\rho_{\mu,t}(\mu_t, \nu'_t),\quad \forall \mu_t \in \P(\X), ~~\forall \nu_t, \nu'_t \in \P(\Y), \label{eqn:constant-blue-reachability}
            \\
            \R^\rho_{\nu,t}(\mu_t, \nu_t) &= \R^\rho_{\nu,t}(\mu'_t, \nu_t),\quad \forall \mu_t, \mu'_t \in \P(\X), ~~ \forall \nu_t \in \P(\Y).
            \label{eqn:constant-red-reachability}
        \end{align}
    \end{subequations}
    Furthermore, for all $\lambda \in [0,1]$, 
    \begin{subequations}
        \begin{align}
            \R^\rho_{\mu,t}\big( \lambda \mu_t + (1-\lambda) \mu'_t, \nu_t \big) &= \lambda \R^\rho_{\mu,t}\big( \mu_t, \nu_t \big) + (1-\lambda) \R^\rho_{\mu,t}\big( \mu'_t, \nu_t \big),
            \quad \forall \mu_t, \mu'_t \in \P(\X), ~~ \forall \nu_t \in \P(\Y), \label{eqn:concave-blue-reachability}
            \\
            \R^\rho_{\nu,t}\big(\mu_t, \lambda \nu_t + (1-\lambda) \nu'_t \big) &= \lambda \R^\rho_{\nu,t}\big(\mu_t, \nu_t \big) + (1-\lambda) \R^\rho_{\nu,t}\big(\mu_t, \nu'_t \big),
            \quad \forall \mu_t \in \P(\X), ~~\forall \nu_t, \nu'_t \in \P(\Y). \label{eqn:concave-red-reachability}
        \end{align}
    \end{subequations}
\end{restatable}

\begin{proof}
    Due to symmetry, we only prove the results for the Blue reachability correspondence.
    The results for the Red team can be obtained through a similar argument. 
    
    Consider an arbitrary pair of distributions $\mu_t \in \P(\X)$ and $\nu_t \in \P(\Y)$. 
    If $\mu_{t+1} \in \R_{\mu,t}^\rho (\mu_t, \nu_t)$ under the independent dynamics in~\eqref{eqn:independent-dynamics},  there exists a local policy $\pi_t \in \Pi_t$ such that
    \begin{equation*}
        \mu_{t+1}(x') = \sum_{x \in \X} \Big[\sum_{u \in \U} \bar{f}_t(x'|x,u)\pi_t(u|x)\Big] \mu_t(x),
    \end{equation*}
    which is independent of $\nu_t$. 
    Consequently, 
    \begin{equation*}
        \R_{\mu,t}^\rho (\mu_t, \nu_t) = \R_{\mu,t}^\rho (\mu_t, \nu'_t), \quad \forall~ \mu_t \in \P(\X) \text{ and } \forall~\nu_t, \nu'_t \in \P(\Y).
    \end{equation*}
    
    Next, we prove the property in~\eqref{eqn:concave-blue-reachability}.
    For the $\subseteq$ direction, consider two Blue mean-fields $\mu_t, \mu'_t \in \P(\X)$ and a Red mean-field $\nu_t \in \P(\Y)$. 
    Let $\bar{\mu}_{t+1} \in \R^\rho_{\mu,t}(\lambda \mu_t + (1-\lambda)\mu'_t, \nu_t)$. 
    From the definition of the reachable set, there exists a local policy $\pi_t \in \Pi_t$ such that
    \begin{align*}
        \bar{\mu}_{t+1}(x') &= \sum_{x \in \X} \Big[\sum_{u \in \U} \bar{f}_t(x'|x,u) \pi_t(u|x) \Big] (\lambda \mu_t(x) + (1-\lambda)\mu_t'(x))  \\
        &=\lambda \underbrace{\sum_{x \in \X} \Big[\sum_{u \in \U} \bar{f}_t(x'|x,u) \pi_t(u|x) \Big]  \mu_t(x) }_{\mu_{t+1}(x')}
        +(1-\lambda) \underbrace{\sum_{x \in \X} \Big[\sum_{u \in \U} \bar{f}_t(x'|x,u) \pi_t(u|x) \Big]  \mu_t'(x)}_{\mu'_{t+1}(x')},
    \end{align*}
    where $\mu_{t+1}\in \R^\rho_{\mu,t} (\mu_t,\nu_t)$ and $\mu'_{t+1}\in \R^\rho_{\mu,t} (\mu'_t,\nu_t)$.
    Hence, we have 
    \[\R^\rho_{\mu,t}(\lambda \mu_t + (1-\lambda)\mu'_t, \nu_t) \subseteq \lambda \R^\rho_{\mu,t}\big( \mu_t, \nu_t \big) + (1-\lambda) \R^\rho_{\mu,t}\big( \mu'_t, \nu_t \big).\]
    
    \vspace{+5pt}
    
    Let now $\bar{\mu}_{t+1} = \lambda \mu_{t+1} + (1-\lambda) \mu_{t+1}'$, where $\mu_{t+1} \in \R^\rho_{\mu,t}(\mu_t, \nu_t)$ and $\mu_{t+1}' \in \R^\rho_{\mu,t}(\mu_t', \nu_t)$. 
    Further, define $\bar{\mu}_t = \lambda \mu_t + (1-\lambda) \mu'_t$.
    From the definition of the reachable set, there exists local policies $\pi_t$ and $\pi'_t$ such that 
    \begin{align*}
        \bar{\mu}_{t+1} &= \lambda \sum_{x \in \X} \Big[\sum_{u \in \U} \bar{f}_t(x'|x,u) \pi_t(u|x) \Big]  \mu_t(x)
        +(1-\lambda) \sum_{x \in \X} \sum_{u \in \U} \bar{f}_t(x'|x,u) \pi'_t(u|x) \Big]  \mu_t'(x) \\
        &= \sum_{x \in \X} \sum_{u \in \U} \bar{f}_t(x'|x,u) \Big[ \lambda \pi_t(u|x) \mu_t(x) + (1-\lambda) \pi'_t(u|x) \mu'_t(x)\Big] \\
        &= \sum_{x \in \X} \sum_{u \in \U} \mathds{1}_{\bar{\mu}_t(x) >0} \; \bar{f}_t(x'|x,u) \Big[ \lambda \pi_t(u|x) \mu_t(x) + (1-\lambda) \pi'_t(u|x) \mu'_t(x)\Big] \\
        & + \sum_{x \in \X} \sum_{u \in \U} \mathds{1}_{\bar{\mu}_t(x) =0}  \; \bar{f}_t(x'|x,u) \underbrace{\Big[ \lambda \pi_t(u|x) \mu_t(x) + (1-\lambda) \pi'_t(u|x) \mu'_t(x)\Big]}_{=0 \, \footnotemark} \\
        &= \sum_{x \in \X} \sum_{u \in \U} \mathds{1}_{\bar{\mu}_t(x) >0} \; \bar{f}_t(x'|x,u) \frac{ \lambda \pi_t(u|x) \mu_t(x) + (1-\lambda) \pi'_t(u|x) \mu'_t(x)}{\bar{\mu}_t(x)} \bar{\mu}_t(x) \\
        &+\sum_{x \in \X} \sum_{u \in \U} \mathds{1}_{\bar{\mu}_t(x) =0} \; \bar{f}_t(x'|x,u) \underbrace{\Big[ \lambda \pi_t(u|x) + (1-\lambda) \pi'_t(u|x) \Big] \bar{\mu}_{t}(x)}_{=0} \\
        &= \sum_{x \in \X} \left[\sum_{u \in \U} \bar{f}_t(x'|x,u) \bar{\pi}_t(u|x)\right] \bar{\mu}_t(x),
    \end{align*}
    \footnotetext{Since the condition $\bar{\mu}_t(x)=0$ together with $\lambda \in (0,1)$ implies that $\mu_t(x) =0$ and $\mu_t'(x) = 0$.}
    
    where the ``averaged" local policy $\bar{\pi}_t$ is given by
    \begin{equation*}
        \bar{\pi}_t (u|x) = \mathds{1}_{\bar{\mu}_t(x) > 0}  \frac{\lambda\pi_t(u|x)\mu_t(x) + (1-\lambda) \pi'_t(u|x) \mu'_t(x)}{\lambda \mu_t(x) + (1-\lambda) \mu'_t(x)} + \mathds{1}_{\bar{\mu}_t(x) = 0} \Big(\lambda\pi_t(u|x) + (1-\lambda) \pi'_t(u|x) \Big).
    \end{equation*}
    Consequently, we have
    \begin{equation*}
        \bar{\mu}_{t+1} \in \R^\rho_{\mu,t} (\lambda \mu_t + (1-\lambda) \mu_{t}', \nu_t).
    \end{equation*}
     Hence,  
     \[\R^\rho_{\mu,t}(\lambda \mu_t + (1-\lambda)\mu'_t, \nu_t) \supseteq \lambda \R^\rho_{\mu,t}\big( \mu_t, \nu_t \big) + (1-\lambda) \R^\rho_{\mu,t}\big( \mu'_t, \nu_t \big).\]

\end{proof}


\EGV*
\begin{proof}
The game value at time $T$ is directly the mean-field reward, which is assumed to be concave-convex.
Hence, one can provide an inductive proof leveraging Lemma~\ref{lmm:concave-convex} and Lemma~\ref{lmm:rset-indep-dynamics} to show that both the lower and upper game values are concave-convex at each time step. Note that the lower game value in~\eqref{eqn:lower-j-ab-t-rset} has the form of a max-min marginal function. 
Finally, as the reachable sets are compact and convex, one can apply the minimax theorem~\citep{owen2013game}, which guarantees that the max-min optimal value is the same as the min-max optimal value, and thus ensures the existence of the game value.
\end{proof}

\section{Policy Extraction from Reachability Set}
\label{appdx-sec:LP}

In Section~\ref{subsec:equivalent-opt}, we changed the optimization domain from the policy space to the reachability set and obtained the optimal next-to-go mean-field.
We now address the natural question regarding the construction of the coordination strategy that achieve the desired next-to-go mean-field. 

We consider the following general policy extraction problem. 

\begin{problem}
    Given a desired next-to-go Blue mean-field $\mu^\rho_{t+1} \in \RMSet{t}{\rho}$, find the local policy $\pi_t \in \Pi_t$ such that
    \begin{equation*}
        \mu^\rho_{t+1}  = \mu^\rho_t F^\rho(\mu^\rho_t, \nu^\rho_t, \pi_t).
    \end{equation*}
\end{problem}

From the convex hull characterization of the reachability set (see~Proposition~\ref{appdx-lmm:Rset-convex-hull}), we have that the mean-field $\mu^\rho_{t+1} \in \RMSet{t}{\rho}$ can be written as 
\begin{equation*}
    \mu_{t+1}^\rho = \sum_{k} \lambda^k \mu_{t}^\rho F^\rho_t(\mu_t^\rho , \nu_t^\rho , \hat{\pi}_t^k),
\end{equation*}
where $\lambda^k \geq 0$ and $\sum_{k} \lambda^k = 1$, and $\{\hat{\pi}^k\}_k$ are the pure local policies (see Definition~\ref{appdx-def:pure-policy}).

Since $\mu_t^\rho$, $\nu_t^\rho$ and $\{\pi^k_t\}$ are all known variables at time $t$, the vector-form coefficients $\bm{\lambda}$ can be found by solving the following linear feasibility problem:
\begin{equation}
    \label{appdx-eqn:lp}
    \bm{M}_t \bm{\lambda} = \mu^\rho_{t+1}, \quad \text{s.t.} ~ \bm{\lambda}\geq 0 ~\text{and}~ \sum_k \lambda^k = 1,
\end{equation}
where the matrix $\bm{M}_t \in \mathbb{R}^{|\X| \times |\hat{\Pi}_t|}$ has $\mu_t^\rho F^\rho_t(\mu_t^\rho , \nu_t^\rho , \hat{\pi}_t^k)$ as its columns. 
Again, the feasibility of~\eqref{appdx-eqn:lp} is guaranteed due to the construction of the reachable set. 

From the linearity of the mean-field dynamics in~\eqref{appd-eqn:linear-mf-dynamics}, we have that the local policy $\pi_t = \sum_k \lambda^k \hat{\pi}_t^k$ achieves the desired next-to-go mean-field.
Specifically, 
\[
\mu_{t+1}^\rho = \mu^\rho_t F^\rho(\mu^\rho_t, \nu^\rho_t, \pi_t) = \sum_{k} \lambda^k \mu_{t}^\rho F^\rho_t(\mu_t^\rho , \nu_t^\rho , \hat{\pi}_t^k).
\]

The Blue coordination strategy can then be constructed to select the local Blue policy solved from~\eqref{appdx-eqn:lp} to achieve the optimal next-to-go Blue mean-field.

\end{appendices}

\bibliographystyle{apalike}
\bibliography{references}
\end{document}